\documentclass[twoside,a4paper]{article}
\usepackage{amsmath,graphicx,amssymb,fancyhdr,amsthm,enumerate,textcomp,verbatim,cite}

\usepackage[dvips]{color}

\newtheorem{thm}{Proposition}[section]
\newtheorem{cor}[thm]{Corollary}
\newtheorem{lem}[thm]{Lemma}
\newtheorem{prop}[thm]{Proposition}

\theoremstyle{definition}
\newtheorem{defn}[thm]{Definition}

\theoremstyle{remark}
\newtheorem{rem}[thm]{Remark}
\newtheorem{ex}[thm]{Example}

\def\beq{\begin{eqnarray}}
\def\eeq{\end{eqnarray}}
\def\bsp{\begin{split}}
\def\esp{\end{split}}

\def\d{\mathrm{d}}
\def\diag{\mathrm{diag}}
\def\Tt{\tilde{T}}

\def\T{T}

\def\Rt{R}
\def\Rc{{\cal R}}

\def\a{\alpha}
\def\b{\beta}

\def\diag{\mbox{diag}}

\def\l{\lambda}

\def\bl {\mbox{\boldmath{$\ell$}}}
\def\bn {\mbox{\boldmath{$n$}}}

\def\R{\mathbb{R}}

\def\VC{V^{\mathbb{C}}}
\def\GC{G^{\mathbb{C}}}

\newcommand{\mf}[1]{{\mathfrak #1}}

\newcommand{\mb}[1]{{\mathbb #1}}

\newcommand{\mbold}[1]{\mbox{\boldmath{\ensuremath{#1}}}}

\newcommand{\inn}[2]{\left\langle {#1},{#2}\right\rangle}

\def \bl {\mbox{{$\mbold\ell$}}}
\def \bn {\mbox{{$\bf n$}}}
\def \bm {\mbox{{$\bf m$}}}
\def \bu {\mbox{{$\bf u$}}}
\def \bv {\mbox{{$\bf v$}}}
\newcommand{\be}{\begin{equation}}
\newcommand{\ee}{\end{equation}}
\newcommand{\beqn}{\begin{eqnarray}}
\newcommand{\eeqn}{\end{eqnarray}}
\newcommand{\pa}{\partial}
\newcommand{\ba}{\begin{array}}
\newcommand{\ea}{\end{array}}
\newcommand{\pp}{{\it pp\,}-}

\def \bg {\mbox{\boldmath{$g$}}}

\begin{document}

\title{\Large\textbf{Minimal tensors and purely electric or magnetic spacetimes of arbitrary dimension.}}

\author{{\large Sigbj\o rn Hervik$^\diamond$, Marcello Ortaggio$^\star$ and Lode Wylleman$^{\diamond,\dagger,\natural}$}\\
\vspace{0.05cm} \\
{\small $^\diamond$ Faculty of Science and Technology, University of Stavanger}, {\small  N-4036 Stavanger, Norway}  \\
{\small $^\star$ Institute of Mathematics, Academy of Sciences of the Czech Republic}, {\small \v Zitn\' a 25, 115 67 Prague 1, Czech Republic} \\
{\small $^\dagger$ Faculty of Applied Sciences TW16, Ghent
University},  {\small  Galglaan 2, 9000 Gent, Belgium}\\
{\small $^\natural$ Department of Mathematics, Utrecht University,
Budapestlaan 6, 3584 CD Utrecht, The Netherlands}\\
{\small E-mail: \texttt{sigbjorn.hervik@uis.no,
ortaggio@math.cas.cz, lode.wylleman@ugent.be}} }

\date{\today}
\maketitle
\pagestyle{fancy}
\fancyhead{} 
\fancyhead[EC]{S. Hervik, M. Ortaggio, L. Wylleman}
\fancyhead[EL,OR]{\thepage}
\fancyhead[OC]{Minimal tensors and PE/PM spacetimes}
\fancyfoot{} 

\begin{abstract}

We consider time reversal transformations to obtain twofold
orthogonal splittings of any tensor on a Lorentzian space of
arbitrary dimension {$n$}. Applied to the Weyl tensor of a spacetime, this
leads to a definition of its electric and magnetic parts relative to
an observer (defined by a unit timelike vector field $\bu$), in any
dimension. We study the cases where one of these parts vanishes in
particular, i.e., purely electric (PE) or magnetic (PM) spacetimes.
We generalize several results from four to higher dimensions and
discuss new features of higher dimensions. For instance, we prove
that the only permitted Weyl types are G, I$_i$ and D, and discuss
the possible relation of $\bu$ with the Weyl aligned null directions
(WANDs); we provide invariant conditions that characterize PE/PM
spacetimes, such as Bel-Debever-like criteria, or constraints on
scalar invariants, and connect the PE/PM parts to the kinematic
quantities of $\bu$; we present conditions under which direct
product spacetimes (and certain warps) are PE/PM, which enables us
to construct   explicit examples. In particular, it is also shown
that all static spacetimes are necessarily PE, while stationary
spacetimes (such as spinning black holes) are in general neither PE
nor PM. Whereas ample classes of PE spacetimes exist, PM solutions
are elusive; specifically, we prove that PM Einstein spacetimes of
type D do not exist, in any dimension. Finally, we derive
corresponding results for the electric/magnetic parts of the Riemann
tensor, which is useful when considering spacetimes with matter
fields, and moreover leads to first examples of PM spacetimes in
higher dimensions. We also note in passing that PE/PM Weyl (or
Riemann) tensors provide examples of {\em minimal tensors}, and we
make the connection hereof with the recently proved alignment
theorem~\cite{Hervik11}. This in turn sheds new light on {{the}}
classification of the Weyl tensors based on null alignment,
providing a further invariant characterization that distinguishes
the (minimal) types G/I/D from the (non-minimal) types II/III/N.

\end{abstract}

\newpage \tableofcontents \newpage

\section{Introduction}\label{sec: Intro}

Decompositions of tensors relative to an observer (identified here
with its normalized time-like four-velocity $\bu$) are of great
import in contemporary theoretical physics. One of the most
notorious insights, coming along with Einstein's Special Relativity
already, is that the separate electric and magnetic (henceforth also
abbreviated to EM) fields in Maxwell's electromagnetism are in fact
the electric and magnetic {\em parts}, relative to an observer
$\bu$, of one unified object, the Maxwell tensor $F_{ab}$.
Conversely, given a Maxwell tensor and any observer $\bu$, one may
split the tensor into its electric and magnetic parts relative to
$\bu$. Although the precise value of the EM components clearly
depends on the observer's frame of reference, the property of a
field of being (or not), e.g., purely electric { (PE)} or purely
magnetic {(PM)} is in fact intrinsic and can be easily
determined using the two Lorentz invariants $F_{ab}F^{ab}$ and
$F_{ab}F^{*ab}$ (see \cite{syngespec}). Furthermore, when one
considers the electromagnetic field generated by an isolated,
bounded source, the associated conserved charges can be computed,
via Gauss' law, as specific surface integrals at infinity, to which
only the leading (``Coulomb'') terms of the corresponding electric
and magnetic parts will contribute.

As the twofold EM splitting can be
performed pointwise at any event and for any $\bu$, the procedure
applies in General Relativity as well. Given that the latter
explains the gravitational interaction through the curved spacetime
structure, one may ask
whether gravitational quantities exist playing a role analogous to
the EM fields in classical electromagnetism, and whether a PE
or PM gravitational field can be given an intrinsic meaning and an
invariant characterization. Matte \cite{Matte53} showed that the
answer to the first question is affirmative, by introducing
the electric and magnetic parts, relative to an observer $\bu$, of
the Riemann tensor of a vacuum metric. For general energy-momentum
content (Ricci tensor) this generalizes to the EM parts of the Weyl
tensor. In terms of these parts, the decomposed trace-free second
Bianchi identities indeed take a form analogous to Maxwell's
equations (see, e.g., \cite{Ellis71}).  A positive answer to the second question was
supplied by the work of McIntosh {\em et al.}~\cite{McIntoshetal94},
who deduced an invariant criterion for deciding whether a given Weyl
tensor has PE or PM character (see Remark~\ref{remark: IJM} below).
 In addition, building on the analogy with the electromagnetic
field, the EM decomposition of the Weyl tensor has proven to be a
very useful and, by now, standard tool in the initial-value
formulation of the gravitational field, as well as in the definition
of conserved charges and asymptotic symmetries (see, e.g.,
\cite{Geroch77,AshHan78,Asthekar80} and references
therein).\footnote{It should be noted that, in this context, the EM
splitting is sometimes meant wrt a {\em spacelike} vector field
\cite{AshHan78,Asthekar80}.} It has also played an important role in the study of cosmological models \cite{Ellis71,Stephanibook}.

With the emergence of higher-dimensional physical theories such as
string theory, the interest in general $n$-dimensional spacetimes
with Lorentzian signature has grown rapidly. In this paper we
propose a general viewpoint to the splitting of tensors, deduced
from the theory of Cartan involutions of a semi-simple Lie group. In
the case of the Lorentz group, these involutions are simply
reflections of unit timelike vectors (``$n$-velocities''), $\bu\mapsto-\bu$. As we will
see this leads to a twofold splitting of any tensor (see
\cite{Senovilla00,Senovilla01} and remark \ref{remark: Senovilla}
for a comparison with Senovilla's approach).
When applied to
the Weyl tensor, the splitting provides a natural definition of its
electric and magnetic parts relative to $\bu$. We show that this
definition is sound, by proving that several four-dimensional
results concerning purely electric (PE) or purely magnetic (PM) Weyl
tensors or spacetimes generalize to higher dimensions.  In
addition to Senovilla's papers mentioned above, a similar splitting
of the Weyl tensor in higher dimensions has also been considered in
the study of asymptotic properties at spatial infinity and of
conserved charges, see, e.g., \cite{ManMar06,TanTanShi09}. Our work
does not overlap with the results of such references.

Recently, one of us proved the {\em alignment theorem}, stating that
(direct sums of) {\em tensors which are {\em not} characterized by
their invariants are precisely the ones of aligned type II or more
special, but not D}~\cite{Hervik11}. In this paper we stress another
equivalent fact, in the realm of the splitting relative to $\bu$:
such tensors are precisely the ones which do {\em not} have a {\em
minimal} tensor relative to $\bu$, in their orbit under the (active)
Lorentz group action on tensors (see below). As we will see, if a
tensor equals one of its parts in the splitting wrt $\bu$, it is
itself minimal wrt $\bu$ (but not viceversa, in general). In
particular, a Weyl tensor which is PE/PM wrt a unit timelike vector
$\bu$ is minimal wrt the same $\bu$, but more stringent conditions
than those based on the alignment theorem will be deduced.

The structure of the paper is as follows.
Preliminary basic results and definitions necessary for our work,
such as theorems about tensors characterized by their invariants,
the twofold splitting of a tensor relative to an arbitrary unit
timelike vector $\bu$, and null alignment theory, are relegated to Appendix~\ref{section: Preliminaries}, since mostly known -- however, this can be the starting section for a reader not familiar with such concepts.
In section~\ref{sec_minimal} we
present an algebraic criterion for a tensor to be minimal wrt $\bu$,
provide sufficient conditions and examples, and make the connection
with the alignment theorem of \cite{Hervik11}. The twofold splitting is applied to the
Weyl tensor in section~\ref{sec_Weyl}, and to the Ricci and Riemann tensors in
section~\ref{PE_PM_Riemann}. In both parts we derive several useful results and
examples of spacetimes for which the tensors in question are
purely electric or magnetic. We end with conclusions and a discussion.
In Appendix~\ref{app_minimal} we present an alternative, more explicit proof of the general Proposition~\ref{th_minimal_ODIG} for the special case of Ricci- and
Maxwell-like rank 2 tensors (namely: they are minimal if and only if they are not of alignment type II (but not D) or more special). In Appendix~\ref{app_congruences} we summarize standard definitions of the
kinematic quantities of a unit timelike congruence, and write parts of the Riemann and Weyl tensors in terms of these.
\\

{\em Notation.} The symbol ${\cal F}_M$ denotes the set of smooth
scalar functions of an $n$-dimensional spacetime $M$.  We will write
$A^\bot$ for the orthogonal complement of a set $A$, and denote a
tensor either in index-free notation ($\T$) or abstract index
notation, with lowercase, possibly numbered Latin letters
$a,\,a_1,\,a_2,\,b,\,c,\,d,\,\ldots$ ($T^{a_1\ldots
a_r}{}_{b_1\ldots b_s}$), or clumping the abstract indices ($T^{\bf
a}{}_{\bf b}$), whatever is more convenient in the context. In the
index-free notation a metric tensor in use will be denoted by $\bg$;
likewise tangent vectors and one-forms will be bolded, and
$\bv\sim\bf{w}$ means that $\bv$ is proportional to $\bf{w}$. The
Riemann, Ricci and Weyl tensors of a spacetime will be denoted by
${R}_{abcd}$, $R_{ab}\equiv R^c{}_{acb}$ and $C_{abcd}$,
respectively, while ${\cal R}\equiv R^a{}_a$ symbolizes the Ricci
scalar.

A component of a tensor $T$ in an unspecified frame
$\{\bm_{\alpha=1,\ldots, n}\}$ of tangent space, with dual frame
$\{\bm^{\alpha=1,\ldots, n}\}$, is denoted  by
$T^{\alpha_1\ldots\alpha_r}{}_{\beta_1\ldots\beta_s}$ (or
$T^{\boldsymbol \alpha}{}_{\boldsymbol\beta}$). An {\em orthonormal}
frame (henceforth, ONF) is of the form $\{{\bf u},{\bf
m}_{i=2,...,n}\}$, where we will use the frame label `$u$' for the
timelike vector $\bu$ (instead of $1$) and $i, j, k, \ldots=2,
\ldots ,n$ for the spacelike frame vectors.  When $\bu$ is a
specific timelike vector we will call any ONF $\{{\bf u},{\bf
m}_{i=2,...,n}\}$ a {\em $\bu$-ONF}. In general
$(T^{\text{sp}})_{a\ldots b}\equiv h_a{}^c\ldots h_b{}^dT_{c\ldots
d}$ denotes the {\em purely spatial part} of a tensor $T$ wrt $\bu$
(see~(\ref{g and theta}) for the definition of the projector
$h_a{}^c$); if $\T=\T^{\text{sp}}$ the tensor is called {\em purely
spatial} (relative to $\bu$), and in any $\bu$-ONF only components
$T_{ij...}$ can be non-zero for such.

When also one of the spacelike vectors of a $\bu$-ONF is
selected or preferred, say ${\bf m}_2$, we shall indicate the
remaining labels with $\hat i, \hat j, \hat k, \ldots=3, \ldots ,n$
instead.  The null vectors
\begin{equation}
\label{null frame corr}
  {\mbold\ell}=\frac{\bu+{\bm}_{2}}{\sqrt{2}} , \qquad {\bf
  n}=\frac{-\bu+{\bm}_{2}}{\sqrt{2}} ,
\end{equation}
are normalized by $l^an_a=1$ and generate the respective null directions of the
timelike plane spanned by $\bu$ and $\bm_2$. The null frame
$\{\bm_0=\bl,\bm_1=\bn,{\bf m}_{\hat i=3,...,n}\}$ will then be
called {\em adapted} to the $\bu$-ONF $\{{\bf u},{\bf
m}_{i=2,...,n}\}$; notice that we use the frame labels `0' for $l^a$
and `1' for $n^a$ (the dual frame consisting of
${m^0}{}_a=n_a,\;{m^1}{}_a=l_a$ and $m^{\hat i}{}_a=(m_{\hat
i})_a,\,{\hat i=3,...,n}$).

\section{Minimal tensors}

\label{sec_minimal}

\subsection{Definition and algebraic criterion}\label{subsec: minimal def}

The definition of a Cartan involution $\theta$ and of the associated Euclidean product (cf. \eqref{def_inn}) and norm are recalled in Appendix~\ref{subsec: Pre Cartan invol}. Now, wrt the Euclidean product associated to the Cartan involution
$\theta$, the standard definition of a {\em minimal vector} of
a tensor space $V$ is the following.

\begin{defn}
A vector (tensor) $T\in V$ is called minimal iff $||g(T)||\geq ||T||$, for all $g\in G$.
\end{defn}

Since the norm $||.||$ is $K$-invariant such a minimal tensor is not
necessarily unique; i.e., if $T$ is minimal, so is $k(T)$ for $k\in
K$. Moreover, for a tensor $T$ the property of being minimal
obviously depends on the norm $||.||$ and thus on the choice of
$\theta$ (i.e., of $\bu$).

An algebraic criterion for when a tensor is minimal was given in
\cite{RicSlo90}. Let us specify it to our situation,  culminating to
Proposition \ref{prop: crit minimal} below.

Recall that for a Lie group, $G$, we can identify the tangent space
of the identity element, $T_{\sf 1}G$ as its Lie algebra, ${\mf g}$;
i.e., $T_{\sf 1}G\cong {\mf g}$. Furthermore, there is an analytic
map, $\exp: {\mf g}\mapsto G$, along with a local inverse
$\exp^{-1}:U\mapsto {\mf g}$, where $U\subset G$ is some
neighbourhood of the identity ${\sf 1}\in G$. This map, along with
its inverse, enables us to write any element $g\in U$ as
$g=\exp(\mathcal{X})$, for some ${\mathcal X}\in{\mf g}$. Moreover,
given any ${\mathcal X}\in{\mf g}$ we can generate a one-parameter
subgroup of $G$ by $g_\tau=\exp(\tau\mathcal{X})$.

In our situation $G=O(1,n-1)$, and we denote the Lie
algebra by $\mf{o}(1,n-1)$. Then the action of an element ${\mathcal{X}}\in \mf{o}(1,n-1)$ on $V$
is defined via the one-parameter subgroup $g_\tau=\exp(\tau\mathcal{X})$; explicitly:
\begin{eqnarray}
\mathcal{X}(T)\equiv
\lim_{\tau\rightarrow 0}\frac{1}{\tau}[g_\tau(T)-T]. \label{X(v)}
\end{eqnarray}
If $({\cal X}^{\alpha}{}_{\beta})$ is the representation matrix of
${\mathcal X}$ acting on tangent space wrt a basis
$\{\bm_{\alpha=1,\ldots, n}\}$ we get by (\ref{g(v)}):
\begin{equation}\label{X action}
{\mathcal
X}(T)^{\alpha_1...\alpha_r}_{\phantom{\alpha_1...\alpha_r}\beta_1...\beta_s}=-\sum_{k=1}^r
{\cal
X}^{\alpha_k}{}_{\alpha_k'}T^{\alpha_1...\alpha_k'...\alpha_r}{}_{\beta_1...\beta_s}+\sum_{l=1}^s
{\cal
X}^{\beta_l'}{}_{\beta_l}T^{\alpha_1...\alpha_r}{}_{\beta_1...\beta_l'...\beta_s}.
\end{equation}
Furthermore, we may split $\mf{o}(1,n-1)$  into eigenspaces of
$\theta$:
\begin{equation}\label{decomp o(1,n-1)}
\mf{o}(1,n-1)=\mf{B}\oplus\mf{K},
\end{equation}
where the $+1$ eigenspace $\mf{K}$ is the Lie algebra
of the maximal compact subgroup $K$, while the $-1$ eigenspace
$\mf{B}$ is the vector space
consisting of the generators of the boosts in planes through
$\bu$.
Moreover, since the elements ${\cal X}\in\mf{o}(1,n-1)$ are
antisymmetric wrt the inner product $\bg$ it follows from remark
\ref{rem selfadjoint} that ${\cal X}_+\in\mf{K}$ and ${\cal
X}_-\in\mf{B}$ are the antisymmetric, respectively, symmetric part
of ${\cal X}$ wrt the inner product $\inn{-}{-}$ (cf.\
(\ref{B-orthonorm}) and (\ref{B-null}) regarding the boost
generators). Hence, $\inn{ {\mathcal{X}}(T)}{T}=0$ for all
$\mathcal{X}\in \mf{K}$.
For $\mathcal{X}\in \mf{B}$ this is not necessarily zero, but

\begin{prop}\label{prop: crit minimal} A covariant tensor $T\equiv T_{a_1\ldots a_m}$ is minimal iff
\begin{eqnarray}
\inn{\mathcal{X}(T)}{T}=0,\;\; \forall \mathcal{X}\in\mf{B}.
\label{crit_minimal}
\end{eqnarray}
In a $\bu$-ONF $\{\bm_\alpha\}=\{\bm_1=\bu,{\bf m}_{i=2,...,n}\}$
this is equivalent with
\begin{equation}\label{minimal crit bu-ONF}
\sum_{k=1}^m\sum_{\a_1=1}^n\cdots \sum_{\a_{m-1}=1}^n
T_{\a_1..\a_{k-1}i\a_k\ldots \a_{m-1}}T_{\a_1..\a_{k-1}u\a_k\ldots
\a_{m-1}}=0,\quad \forall i=2,...,n.
\end{equation}
\end{prop}
\begin{proof} The criterion (\ref{crit_minimal}) was proved in
\cite{RicSlo90} (Theorem 4.3) in a more general context. The
component form (\ref{minimal crit bu-ONF}) follows straightforwardly
from (\ref{innST}), (\ref{X action}) and the fact that $\mf{B}$ is
spanned by the boost generators ${\cal X}_i$ in the
$(\bu,\bm_i)$-planes, ${\cal X}_i^{ab}\equiv 2u^{[a}m_i^{b]}$.
\end{proof}

To write (\ref{minimal crit bu-ONF}) in a covariant way one replaces
$T_{\ldots i\ldots}$ by $h_a^{~b}T_{\ldots b\ldots}$ and $T_{\ldots
u\ldots}$ by $u^c T_{\ldots c\ldots}$ for the free indices $i$ and
$u$, but one has to be careful with the $\a_l$'s since raising $u$
gives a minus sign.

\begin{ex}\label{ex:minimal} For covectors $v_a$ the criterion (\ref{minimal crit
bu-ONF}) becomes simply
\begin{equation}\label{minimal vector}
v_iv_u=0,\;\forall i=2,...,n\qquad\Leftrightarrow\qquad h_a^{~b}v_b
(u^cv_c)=0.
\end{equation}
Hence, ${\bf v}$ is minimal wrt $\bu$ iff, relative to $\bu$, it is
either purely temporal (i.e., proportional to $\bu$) or purely spatial (i.e.,
orthogonal to $\bu$).

For symmetric (Ricci-like) rank 2 tensors $R_{ab}=R_{(ab)}$ we get
\begin{equation}
\label{minimal 2tensor symm}
0=\tfrac 12\sum_{\alpha=1}^n(R_{i\alpha}R_{u\alpha}+R_{\alpha
i}R_{\alpha u})=\sum_{\alpha=1}^n
R_{i\alpha}R_{u\alpha}=R_{ij}R_u{}^j+R_{iu}R_{uu}=R_{ia}R^a{}_u+2R_{uu}R_{iu}.
\end{equation}

Likewise, for antisymmetric (Maxwell-like) rank 2 tensors
$F_{ab}=F_{[ab]}$ (\ref{minimal crit bu-ONF}) reduces to
\begin{equation}
\label{minimal 2tensor antisymm}
0=F_{ij}F_u{}^j=-F_{ia}F^a{}_u\qquad\text{or}\qquad
u_{[a}F_{b]c}F^c{}_d u^d=0,
\end{equation}
i.e., $\bu$ is an eigenvector of $F^a{}_bF^b{}_c$.

Finally, for rank
4 tensors $C_{abcd}$ satisfying the first two parts of the
Riemann-like symmetries (\ref{symm1}) we get
\begin{equation}
\label{minimal 4tensor Riem}
0=\sum_{\a=1}^n\sum_{\b=1}^n\sum_{\gamma=1}^n
C_{i\a\b\gamma}C_{u\a\b\gamma}=C_{ijkl}C_u{}^{jkl}+2C_{ijku}C_u{}^{jk}{}_u=C_{iabc}C_u{}^{abc}+4C_{iabu}C_u{}^{ab}{}_u.
\end{equation}

The above examples already show an interesting analogy in the four cases. Obviously, a covector is minimal iff it is not null (which could be dubbed ``type N'' in the sense of alignment theory). Rewriting the conditions (\ref{minimal 2tensor symm})--(\ref{minimal 4tensor Riem}) in a null frame, one also immediately sees that: if $R_{ab}$ is minimal it can not be of any of the types II (not D), III and N (i.e., only the types G, I and D can be minimal); if $F_{ab}$ is minimal it can not be neither type II (not D) nor N (i.e., only the types G and D can be minimal);  if $C_{abcd}$ is minimal it can not be of any of the types II (not D), III and N (i.e., only the types G, I and D can be minimal). One can show that the converse is also true (i.e., the admitted types are also sufficient conditions to ensure minimality) and that, in fact, a more general such result holds for {\em any} tensor, as we shall show below in Proposition~\ref{th_minimal_ODIG} (see also Appendix~\ref{app_minimal} in the case of $R_{ab}$ and $F_{ab}$).
\end{ex}

\subsection{Sufficient conditions and examples}\label{subsec: minimal suff conds}

Using  $\theta$, any tensor $T$ can be split as (see Appendix~\ref{subsec: Pre Cartan invol} for more details)
\begin{equation}
    \label{orthsplit_text}
 T=T_+ + T_-, \qquad T_{\pm}=\frac 12[T\pm\theta(T)] ,
\end{equation}
which will be used in the following.

Consider a $\bu$-ONF  ${\cal F}_u=\{{\bf u},{\bf m}_{i=2,...,n}\}$.
In such a frame any element ${\cal X}\in\mf{B}$ acting on
$T_pM$ is represented by a symmetric matrix of the form:
\begin{eqnarray} [{\cal X}]_{{\cal F}_u}=({\cal X}^\alpha{}_\beta)=
\begin{bmatrix}
0 & z_2 &  ... & z_n \\
z_2 & 0 &  ... & 0 \\
\vdots & \vdots& & \vdots \\
z_n & 0 & ... & 0
\end{bmatrix}.
 \label{B-orthonorm}
\end{eqnarray}
In the  null frame ${\cal F}'=\{{\mbold\ell},{\bf n},{\bf m}_{\hat
i=3,...,n}\}$ adapted to ${\cal F}_u$ (see (\ref{null frame corr}))
${\cal X}$ is represented by the symmetric matrix (in $1+1+(n-2)$
block-form):
\begin{eqnarray} [{\cal X}]_{{\cal F}'}= \frac{1}{\sqrt{2}}
\begin{bmatrix}
\lambda & 0 &  z^t_{\hat i} \\
0 & -\lambda &   -z^t_{\hat i} \\
z_{\hat i} & -z_{\hat i}  & {\sf 0}
\end{bmatrix},\quad \lambda=\sqrt{2}\,z_2.
\label{B-null}\end{eqnarray}

In what follows boost-weight decompositions will refer to the
adapted null frame ${\cal F}'$,
and given a tensor $T$ the collection of its components of
boost-weight $b$ will be denoted as $(T)_b$ (see Appendix
\ref{subsec: null alignment}, also for the nomenclature regarding
(null alignment) types of tensors in the subsequent text).

Let us split ${\mathcal X}\in\mf{B}$ using { a} vector space basis $\{
{\mathcal X}_B,{\mathcal X}_{\hat i} \}$ of $\mf{B}$, where
${\cal X}_B$ is the generator of the boost (\ref{boost}). Hence,
eq.~(\ref{B-null}) becomes ${\mathcal X}=\lambda {\mathcal
X}_B+z_{\hat i}{\mathcal X}_{\hat i}$.
We note that the boost-weight decomposition of $T$ is the eigenvalue
decomposition with respect to ${\mathcal X}_B$:
\begin{eqnarray} {\mathcal X}_B(T)=\sum_{b}b(T)_b.
 \label{eigenvalue_b}
\end{eqnarray} This may serve as a definition of the boost-weight $b$
components of $T$: $(T)_b$  is the eigenvector of ${\mathcal X}_B$
with eigenvalue $b$.

The action of ${\mathcal X}_{\hat i}$ on an arbitrary tensor $T$ is
a bit more complicated, but can be derived from (\ref{X action}) and
(\ref{B-null}), with $\lambda=0$ and $z_{\hat k}=0,\,{\hat k}\neq
\hat{i}$. Also, using (\ref{B-null}) we note that ${\mathcal
X}_{\hat i}$ raises and lowers the b.w.\  by 1; i.e.,
\begin{eqnarray}\label{Xi boost weights} {\mathcal X}_{\hat
i}((T)_b)=\left({\mathcal X}_{\hat
i}((T)_b)\right)_{b-1}+\left({\mathcal X}_{\hat
i}((T)_b)\right)_{b+1}. \end{eqnarray}
Since  $\inn{-}{-}$ is bilinear we have \begin{eqnarray}
\inn{{\mathcal X}(T)}{T}= \lambda\inn{ {\mathcal X}_B(T)}{T}+z_{\hat
i}\inn{{\mathcal X}_{\hat i}(T)}{T}.
 \label{minimal_decomp}
\end{eqnarray}
Thus, to check minimality we can consider ${\mathcal X}_B$ and
${\mathcal X}_{\hat i}$ also separately.
Based on these observations we have
\begin{thm}
\label{th_minimal_ex} Any of the following conditions is sufficient
for a tensor $T\in V$ to be minimal:
\begin{enumerate}
\item{} $T$ is a $\theta$-eigenvector, i.e., $T=T_+$ or $T=T_-$;
\item{} $T$ has the boost-weight decomposition $T=(T)_0$ {(and thus is of type D)}.
\end{enumerate}
\end{thm}
\begin{proof}
1. was proven in \cite{RicSlo90}: essentially, for any $T$, we have
$\inn{{\mathcal X}(T)}{T}=2\inn{{\mathcal X}(T_+)}{T_-}$,
${\mathcal{X}}\in\mf{B}$ (using $\mathcal{X}(T_{\pm})\in V_{\mp}$).
Thus, if $T_-=0$ or $T_+=0$ the criterion (\ref{crit_minimal}) is
fulfilled.

\noindent 2. If $T=(T)_0$ then (\ref{eigenvalue_b}) implies ${\cal
X}_B(T)=0$, while (\ref{Xi boost weights}) and
$\inn{X}{Y}=\sum_{b}\inn{(X)_b}{(Y)_b}$ give $\inn{{\mathcal
X}_{\hat i}(T)}{T}=0$. Thus $\inn{{\mathcal X}(T)}{T}=0$ from
(\ref{minimal_decomp}) and again (\ref{crit_minimal}) is fulfilled.
\end{proof}

\begin{rem}
The two conditions of Proposition~\ref{th_minimal_ex} are only
sufficient conditions and they are, in general, independent. An
exception to this statement is the special case $V=T_pM$, for which
 $T=T_\pm$ is also necessary to be minimal, see  (\ref{minimal vector}),
and $T=(T)_0$ is equivalent to $T=T_+$ (i.e., $T$ is a spacelike
vector). However, if $T$ represents a Maxwell-like tensor (bivector,
$T_{ab}=F_{ab}=F_{[ab]}$) we have $F=(F)_0\Leftrightarrow F_{u\hat
i}=0=F_{2\hat i}$, but $F_{u2}$ and $F_{\hat i\hat j}$ can be
non-zero, so that $F_+\neq F\neq F_-$, in general (here we assume we
are in four or higher dimensions). Similarly, it is easy to see that
$F=F_-$ implies $F=(F)_0$ (the direction of ${\bm}_{2}$ being
defined by $F_{ui}$), whereas $F=F_+$ can be of type G if $n$ is odd
(see Remark \ref{rem: antisymm class} in Appendix~\ref{app_minimal}
for a complete discussion). Moreover, as we shall discuss below
(Proposition~\ref{th_minimal_ODIG}), all Weyl tensors of type G, I
or D contain a minimal Weyl tensor in their orbit, with no need to
satisfy either 1.\ or 2.\ above.
\end{rem}

\begin{ex}
As an example of a more generic minimal tensor,
choose a tensor $T=(T)_{-2} + (T)_0+(T)_{+2}$ where both  $(T)_{-2}$
and $(T)_{+2}$ are non-zero.\footnote{In fact, with no essential
change in the following argument we could more generally also
consider a tensor of the form $T=(T)_{-k} + (T)_0+(T)_{+k}$, where
$k>1$.} First, consider (\ref{B-null}) using ${\mathcal X}_B$; then
we get $\mathcal{X}_B(T)=\sum_{b}b(T)_b=-2(T)_{-2}+2(T)_{+2}$.
Consequently,
\[ \inn{\mathcal{X}_B(T)}{T}=-2\inn{(T)_{-2}}{(T)_{-2}}+2\inn{(T)_{+2}}{(T)_{+2}}, \]
which is in general not zero. However, by a boost of the frame,
\[ \inn{\mathcal{X}_B(T)}{T}=-2e^{-{4\lambda}}\inn{(T)_{-2}}{(T)_{-2}}+2e^{4\lambda}\inn{(T)_{+2}}{(T)_{+2}};\]
therefore, {\em there exists a boost of the frame such that
$\inn{\mathcal{X}_B(T)}{T}=0$ in which case
$\inn{(T)_{-2}}{(T)_{-2}}=\inn{(T)_{+2}}{(T)_{+2}}$. }

Consider next (\ref{B-null}) using ${\mathcal{X}_{\hat i}}$:
$\mathcal{X}_{\hat i}(T)$ has only odd boost-weight, so
$(\mathcal{X}_{\hat i}(T))_b=0$ for $b$ even. Thus, since $T$ has
only even b.w.\ components:
\[ \inn{\mathcal{X}_{\hat i}(T)}{T}=0.\]
Thus we reach the conclusion that for any $T=(T)_{-2} +
(T)_0+(T)_{+2}$ there exists a boost generated by ${\mathcal X}_B$
such that it is minimal. For this minimal vector, we have the
condition $\inn{(T)_{-2}}{(T)_{-2}}=\inn{(T)_{+2}}{(T)_{+2}}$.

We still notice that it is important in this example that both
$(T)_{-2}$ and $(T)_{+2}$ are non-zero (alternatively, both zero for
which $T=(T)_0$ and it falls under the spell of Proposition
\ref{th_minimal_ex}). Indeed, if one of these parts were zero
while the other is not, there would not be any minimal $T$. This is
connected to the fact that tensors  which are of type II or more
special, but not D nor O, do not have a minimal vector in their
orbit (see Proposition~\ref{th_minimal_ODIG}).

Furthermore, we should emphasize that the minimal example
$T=(T)_{-2} + (T)_0+(T)_{+2}$ does not need to fulfill condition 1.\
nor 2.\ in Proposition~\ref{th_minimal_ex} showing, again, that
these conditions are only sufficient.
\end{ex}

\subsection{Minimal tensors and null alignment type}

\label{subsec: minimal alignment}

In this subsection we revisit the `alignment theorem' for tensors
over a Lorentzian space of any dimension proved in \cite{Hervik11},
giving a more streamlined proof and adding the connection with
minimal tensors. Version A is the contrapositive of version B.  The
statements (1) assume a chosen unit timelike vector $\bu$ and
associated Euclidean product, and the abbreviations (Act) and (Pass)
refer to the active and passive viewpoints.
\begin{thm}
\label{th_minimal_ODIG}
For a tensor $T$ the following are equivalent:\\
{[Version A]}
\begin{enumerate}
\item[(1a)] {(Act)} There exists a minimal tensor $v$ in the orbit $\mathcal{O}(T)$;
 (Pass) There exists a possibly different vector $\bu'$ such
that the representation $\tilde T$ of $T$ in a $\bu'$-ONF is minimal
in $\mathcal{O}(\tilde{T})$.
\item[(2a)] $T$ is of type O, D, or any other type which is {\em not} type II or
more special.
\item[(3a)] $T$ is characterised by its invariants.
\end{enumerate}
[Version B]
\begin{enumerate}
\item[(1b)] {(Act)} There exists {\em no} minimal tensor $v$ in the orbit $\mathcal{O}(T)$;
 (Pass) No ONF-representation $\tilde T$ of $T$ is minimal.
\item[(2b)] $T$ is of type II  or more special, but not D nor O.
\item[(3b)] $T$ is \emph{not} characterised by its invariants.
\end{enumerate}
\end{thm}
\begin{proof}

\noindent (1a)$\Leftrightarrow$ (3a): Let $\mf{M}\subset V$ denote the set of minimal vectors. In \cite{RicSlo90} it was proved that
\be\label{closed 2}
\mf{M}\cap {\mathcal{O}}(\Tt)\neq \emptyset \Leftrightarrow  {\mathcal{O}}(\Tt) \in \mf{C}
\ee
and the equivalence follows from corollary \ref{cor: closed is invarchar}.

\noindent (2b)$\Leftarrow$ (3b): From \cite{RicSlo90} we have that,
if $\mathcal{O}(T)$ is not closed, then there exists a vector $v_0$
in the closure $\overline{\mathcal{O}(T)}$ and $\mathcal{X}\in
\mf{B}$ such that $e^{\tau\mathcal{X}}(T)\rightarrow v_0$, as
$\tau\rightarrow+\infty$. By considering the boost-weight
decomposition with respect to the boost
$B(\tau)=e^{\tau\mathcal{X}}$, we get \cite{Hervik11}
\[  e^{\tau\mathcal{X}}(T)=\sum_be^{\tau b}(T)_b.\]
Since $v_0$ is finite we need $b\geq 0$, or $b\leq 0$. By the
isomorphism $b\leftrightharpoons -b$, we can assume $b\leq 0$.
Moreover, type D is ruled out by Proposition~\ref{th_minimal_ex},
and the result follows.

\noindent (2b)$\Rightarrow$ (3b): If the tensor $T$ is of type II or
more special, but not type D nor O, then there exists a boost-weight
decomposition \be\label{type II} T=\sum_{b\leq 0}(T)_b.\ee By the
action of the boost ($\tau\rightarrow+\infty$):
\[
e^{\tau\mathcal{X}}(T)=\sum_{b\leq 0}e^{\tau b}(T)_b\rightarrow (T)_0.
\]
If $v_0\equiv (T)_0$ is in $\mathcal{O}(T)$, then there exists a
frame such that $T=(T)_0$, hence type D or O which is a
contradiction. Thus $v_0$ is not in $\mathcal{O}(T)$, which is
consequently not closed, and corollary \ref{cor: closed is
invarchar} concludes the proof.
\end{proof}

Explicitly, in the case of the  Weyl tensor condition 2a covers type
G, strict type I (including subtypes such as I$_i$), type D and O,
while condition 2b covers strict type II (not D), III (including
subtypes such as III$_i$) and N.

The general result remains valid for a collection (or direct sum) of
tensors instead of a single one. Here a collection $(T_i)$ is called
`of (aligned) type II or more special' if {\em all} $T_i$ are of the
form (\ref{type II}) in {\em the same} null frame. For example, if
the Weyl and Ricci tensors of a metric at a spacetime point are both
type N wrt the same null vector, then the corresponding Riemann
tensor will be of type N as well, as follows from Proposition
\ref{prop: Riemann boost order}. If, however, they are both type N
but wrt different null-vectors, then they are not aligned and there
is a minimal vector: if $R=(R)_{-2}$ and $C=(C)_{+2}$, then we
formally write the Riemann tensor as $T=[R,C]$ and we have
$T=(T)_{-2}+(T)_{+2}$ such that we are back in the example
considered in $\S$ \ref{subsec: minimal suff conds}, $T$ being
minimal in a frame such that $\inn{R}{R}=\inn{C}{C}$.

{In appendix~\ref{app_minimal} we give more explicit proofs of Proposition~\ref{th_minimal_ODIG}
in the case of vectors and Ricci- or Maxwell-like rank 2 tensors.}

\section{The Weyl tensor: purely electric (PE) or magnetic (PM) spacetimes}

\label{sec_Weyl}

In the context of General Relativity and its higher dimensional
extensions, the Weyl tensor is a natural object to consider, e.g.,
in the classification of exact solutions {(in particular, of
Einstein spacetimes $R_{ab}={\cal R}g_{ab}/n$)}, in the study of
gravitational radiation, of asymptotic properties of spacetimes,
etc.. We now apply the general orthogonal splitting of tensors
relative to an observer with timelike vector field $\bu$, outlined
in Appendix \ref{subsec: Pre Cartan invol}, to the Weyl tensor
$C_{abcd}$ at a point of a spacetime of dimension $n\geq 4$.
This enables us to define purely
electric and magnetic Weyl tensors and spacetimes, to work out
several useful results such as Bel-Debever criteria, the
structure of the associated Weyl bivector operator and null
alignment properties, and
to provide illustrative examples. We will see that several
well-known results in four dimensions generalize to arbitrary
dimensions.
In the next section we shall apply a similar analysis to the
Ricci and Riemann tensors, which is relevant in the study of
spacetimes which contain matter fields.

\subsection{Electric and magnetic parts}

\label{subsec: Weyl EM parts}

As before, we consider a fixed unit timelike vector $\bu$ and the
corresponding Cartan involution $\theta$.

\begin{defn} The tensor $(C_+)_{abcd}$ ($(C_-)_{abcd}$) is
called the {\em electric (magnetic) part} of the Weyl tensor wrt
$\bu$.
\end{defn}


Recall the definition (eq.~(\ref{g and theta})) of the orthogonal projector
\[
    h_{ab}\equiv g_{ab}+u_au_b .
\]

Define the tensor
\begin{equation}\label{Edef}
E_{ab}\equiv C_{aebf}u^eu^f=h_a{}^ch_b{}^dC_{cedf}h^{ef},
\end{equation} where
(\ref{tracefree}) implies the last equality. Obviously, this is a
trace-free symmetric rank~2 tensor which is moreover purely spatial
relative to $\bu$: $E_{ab}=(E^{\text{sp}})_{ab}$. Using (\ref{g and theta}), (\ref{orthsplit}) and the symmetries
(\ref{symm1}) one obtains
\begin{eqnarray}
&&(C_+)^{ab}{}_{cd}=h^{ae}h^{bf}h_c{}^gh_d{}^h C_{efgh}+4u^{[a}u_{[c}C^{b]e}{}_{d]f}u_eu^f= (C^{\text{sp}})^{ab}{}_{cd}+4u^{[a}u_{[c}E^{b]}{}_{d]},\label{C+}\\
&&(C_-)^{ab}{}_{cd}=2h^{ae}h^{bf}C_{efk[c}u_{d]}u^k+2u_ku^{[a}C^{b]kef}h_{ce}h_{df}.\label{C-}
\end{eqnarray}

In any ONF  $\{\bu,\bm_{i=2,...,n}\}$ the non-identically vanishing
electric (magnetic) part accounts for the components of the Weyl
tensor with an even (odd) number of indices $u$ (cf.\ $\S$
\ref{subsec: Pre Cartan invol}.1). The first, purely spatial term of
the Weyl electric part (\ref{C+}) covers the $C_{ijkl}$ components,
of which there are $N_0(n)=(n^2-2n+4)(n+1)(n-3)/12$ independent
ones; the last term covers the $N_2(n)= (n+1)(n-2)/2$ independent
$C_{uiuj}$ components; however, the latter are fully determined by
the former since
\begin{equation}
    \label{N2 vs N0}
    C_{uiuj}=C_{ikj}{}^k,
\end{equation}
which is the component form of the trace-free property
(\ref{tracefree}) also expressed in (\ref{Edef}); thus there are
$N_0(n)-N_2(n)=n(n^2-1)(n-4)/12$ extra independent purely electric
components $C_{ijkl}$ in addition to the $C_{uiuj}$ ones. The Weyl
magnetic part (\ref{C-}) has $N_1(n)=(n^2-1)(n-3)/3$ independent
components $C_{uijk}$.
Together these add up to the $(n-3)n(n+1)(n+2)/12$
independent components of the Weyl tensor in $n$ dimensions (see
also \cite{Senovilla01}).

\begin{rem}
The already known four-dimensional case $n=4$ has somewhat special
properties, which we now briefly review. One has $N_0(4)=N_2(4)=5$,
such that the relations (\ref{N2 vs N0}) can be inverted to give \be
 C_{ijkl}=2(\delta_{i[k}C_{l]uju}-\delta_{j[k}C_{l]uiu}) \qquad (n=4).
 \label{Cijkl_4D}
\ee Using (\ref{Edef}) this reads $(C^{\text{sp}})^{ab}{}_{cd} =
4h_{[c}^{[a}E^{b]}{}_{d]}$ in covariant form. Thus (\ref{Edef}) and
(\ref{C+}) imply that the tensors $(C_+)_{abcd}$ and $E_{ab}$, both
having 5 independent frame components, are in biunivocal relation:
\be
 \label{C+ 4D} E_{ab}=(C_+)_{acbd}u^cu^d \quad\leftrightarrow\quad
(C_+)^{ab}{}_{cd}=
4\left(h_{[c}^{[a}+u^{[a}u_{[c}\right)E^{b]}{}_{d]}  \qquad (n=4).
\ee
We have $N_1(4)=5$ as well. Define
\begin{equation}
    \label{Hdef} H_{ab}\equiv
\tfrac 12 \varepsilon_{acef}C^{ef}{}_{bd}u^cu^d \qquad (n=4) ,
\end{equation}
where $\varepsilon_{abcd}$ is the volume element. Just as
$E_{ab}$, $H_{ab}$ is a purely spatial, symmetric
and trace-free rank 2 tensor which thus has 5 independent
components. Then, by virtue of the identity
$\varepsilon^{abef}u_e\varepsilon_{cdfg}u^g=2h^a{}_{[c}h^b{}_{d]}$
(\ref{C-}) can be rewritten, and $(C_-)_{abcd}$ and $H_{ab}$ are in
biunivocal relation: \be\label{C- 4D} H_{ab}\equiv \tfrac 12
\varepsilon_{acef}(C_-)^{ef}{}_{bd}u^cu^d\quad\leftrightarrow\quad
(C_-)^{ab}{}_{cd}=2\varepsilon^{abef}u_e
u_{[c}H_{d]f}+2\varepsilon_{cdef}u^e u^{[a}H^{b]f}\qquad (n=4). \ee
Adding the expressions in (\ref{C+ 4D}) and (\ref{C- 4D}) for $C_+$
and $C_-$, one obtains the well-known formula for the Weyl tensor in
four-dimensional General Relativity
in terms of  $E_{ab}$ and $H_{ab}$~\cite{Stephanibook}, which are
usually referred to as the electric and magnetic parts of the Weyl
tensor. Since they are respectively equivalent with $C_+$ and $C_-$
this justifies the above definition of Weyl electric and magnetic
parts, for general $n$.
\end{rem}

\begin{rem}\label{remark: Senovilla} In \cite{Senovilla00} and \cite{Senovilla01} Senovilla proposed
a construction for generalizing the electro-magnetic decomposition
relative to a unit timelike vector $\bu$, applicable to any tensor
$T$ and based on the consideration of maximal antisymmetric index
slots. If the number of such slots is $r$ then one constructs $2^r$
different tensors from $T$, by taking for each slot a contraction
with either $u^a$ (yielding an electric ``E''-contribution for that
slot) or $u_a\varepsilon^{aa_1\ldots a_{n-1}}$ (yielding a magnetic
``H''-contribution). However, by the antisymmetry of the slots this
is equivalent to contraction (over $b,d,...,f$) with
$u_au^bh_c{}^d\ldots h_e{}^f$ and $h_a{}^bh_c{}^d\ldots h_e{}^f$,
respectively. Then, our $T_+$ ($T_-$) part collects the $2^{r-1}$
tensors constructed in this way with an even (odd) number of
E-parts. For instance, when $T_{a[bc]}\neq T_{[abc]}$ in
(\ref{Tex1})-(\ref{Tex2}) then $r=2$, and the first and second term
in the second of (\ref{Tex1}) represent the associated HH- and
EE-tensors associated to $T_{a[bc]}\neq T_{[abc]}$ ($r=2$),
respectively, while the second of (\ref{Tex2}) contains respective
equivalents of the HE- and EH-tensors. For the Weyl tensor we also
have $r=2$; our magnetic part $C_-$ corresponds to Senovilla's EH
and HE tensors, which are equivalent due to the symmetry
$C_{abcd}=C_{cdab}$; our electric part $C_+$ covers the EE and HH
tensors, where the former can be seen as a part of the latter due to
(\ref{N2 vs N0}). Notice that one has a reversed situation for a
Maxwell field $F_{ab}$ since $r=1$, i.e., $F_+$ ($F_-$) covers the
ONF components $F_{ij}$ ($F_{ui}$) and is the magnetic (electric)
part.  As another example, for symmetric rank 2 tensors $T$ like the
energy-momentum tensor one has $r=2$, and the electric part, $T_+$,
then assembles the stress-pressure two-tensor ($T_{ij}$) and scalar
energy density ($T_{uu}$) as measured by $\bu$, while the magnetic
part $T_-$ represents the heat flux vector. Equivalently for these
situations, the `electric (magnetic) part' collects the $2^{r-1}$
tensors with an even (odd) number of H-parts. This leads us to the
following definition for general tensors, where we thus
refer to the definition of H/E-parts in \cite{Senovilla00,Senovilla01} and the above explanation:\\

\noindent {\bf Definition 3.3.} Let $T$ be any tensor with $r$
maximal antisymmetric index slots.   The {\em electric (magnetic)
part} of $T$ relative to a unit timelike vector $\bu$ is the
collection of the $2^{r-1}$ tensors with an even (odd) number of
H-parts; this part equals $T_+$ ($T_-$) when $r$ is even, and $T_-$
($T_+$) when $r$ is odd.
\end{rem}

\subsection{PE/PM condition at a point}\label{subsec: Weyl Bel-Debever}

\begin{defn}\label{def PE/PM} At a point $p$, the Weyl tensor $C$ is called purely electric (magnetic)
[henceforth, PE (PM)] wrt $\bu$ if $C=C_+\leftrightarrow C_-=0$ ($C=C_-\leftrightarrow C_+=0$).
If such a $\bu$ exists the Weyl tensor is called PE (PM);  it is called {\em properly PE (PM) wrt $\bu$ if $C=C_+\neq 0$ ($C=C_-\neq 0$)}. A spacetime, or an open region thereof, is called (properly) PE (PM)
if the Weyl tensor is (properly) PE (PM), everywhere.
\end{defn}

In any ONF  $\{\bu,\bm_{i=2,...,n}\}$ a non-zero Weyl tensor is PE
wrt $\bu$ iff $C_{uijk}=0$, $\forall\, i,j,k =2,...,n$; in view of
(\ref{N2 vs N0}), it is PM wrt $\bu$ iff $C_{ijkl}=0$, $\forall\,
i,j,k,l =2,...,n$.

In analogy with the Bel-Debever criteria for null
alignment~\cite{Stephanibook,Ortaggio09}, and using the {
properties (\ref{symm1}) and (\ref{tracefree})}, one may rewrite
this in the following covariant way.

\begin{prop}(Weyl PE/PM Bel-Debever criteria) Let $\bu$ be a unit timelike vector,
\be\label{unit}
g^{ab}u_au_b=-1.
\ee
Then a Weyl tensor $C_{abcd}$ is
\begin{itemize}
\item PE wrt $\bu$ iff
\be\label{Bel PE}
    u_ag^{ab}C_{bc[de}u_{f]}=0;
\ee
\item PM wrt $\bu$ iff
\be\label{Bel PM}
    u_{[a}C_{bc][de}u_{f]}=0.
\ee
\end{itemize}
\label{prop_Bel-Debever}
\end{prop}

These Bel-Debever criteria are covariant tensor equations, only
involving the metric inverse $g^{ab}$, the Weyl tensor $C_{abcd}$
and the one-form $u_a$. The big advantage of this format of the
PE/PM conditions is that one may take {\em any} basis
$\{\bm_{\alpha=1,...,n}\}$ of $T_pM$, with dual basis
$\{\bm^{\alpha=1,...,n}\}$ of $T_p^*M$, and consider the components
$g^{\alpha\beta}$, $C_{\alpha\beta\gamma\delta}$ and $u_\alpha$.
E.g., when the metric is given in coordinates over a neighbourhood
of $p$, one may take the corresponding holonomic frames of
coordinate vector fields and differentials. One then considers
(\ref{unit}) and (\ref{Bel PE}) [(\ref{unit}) and (\ref{Bel PM})] as
a system of quadratic equations in the $n$ unknowns $u_\alpha$; if a
solution to this system exists then the Weyl tensor is PE (PM)
relative to the corresponding $\bu$. However, since $n\geq 4$ the
number $1+N_1(n)=(n-2)(n^2-n-3)/3$ [$1+N_0(n)=n(n-1)^2(n-2)/12$] of
independent equations in this system exceeds $n$, with degree of
overdeterminacy $d_1(n)=(n-1)(n^2-2n-6)/3$
[$d_0(n)=n(n^3-4n^2+5n-14)/12$. For $n=4$ we already have
$d_0(4)=d_1(4)=2$, and we note that
$d_0(n)-d_1(n)=(n+1)(n-2)(n-3)(n-4)/12$  in general, which
increases with $n$. Hence, for a {\em generic} metric and Weyl
tensor no solution $\bu$ to the PE or PM conditions can be found,
not even at a point $p$, and for $n>4$ the situation is worse for PM
(the number of equations then being quartic in $n$ while cubic in
the PE case).

\begin{rem}\label{rem: PM4D} In the case where only the contraction (\ref{Edef}) vanishes ($C_{uiuj}=0$), we will say
that the Weyl tensor is ``PM'' (note that the quotes are part of the name); this is only equivalent to PM for
$n=4$, but gives a weaker condition for $n>4$ dimensions (since there are no restrictions on $C_{ijkl}$).
\end{rem}

In the next paragraph we will meet easily computable necessary
conditions for the above PE and PM equations to have solutions
$\bu$. In $\S$ \ref{subsec: Weyl alignment} we will discuss the
alignment types for PE/PM Weyl tensors and discuss the uniqueness of
solutions $\bu$. In $\S$ \ref{subsec: PE spacetimes} we will see
that ample classes of PE {\em spacetimes} exist, whereas PM
spacetimes are most elusive ($\S$ \ref{subsec: PM spacetimes}).

\subsection{PE/PM Weyl bivector operators}\label{subsec: Weyl bivector}

Consider the  real $N$-dimensional vector space $\wedge^2T_pM$
of contravariant bivectors (antisymmetric tensors $F^{ab}=F^{[ab]}$)
at $p$,  where $N=n(n-1)/2$. In view of the first three
symmetries in (\ref{symm1}) the map
\begin{equation}\label{Weyl operator}
{\sf C}:\quad F^{ab}\mapsto \tfrac 12 C^{ab}{}_{cd}F^{cd}=\tfrac 12 F^{cd} C_{cd}{}^{ab}
\end{equation}
is a linear operator (=endomorphism) of $\wedge^2T_pM$, referred to
as the {\em Weyl bivector operator} \cite{ColHer09}
, which is symmetric (self-adjoint)
wrt the restriction ${\sf g}$ to $\wedge^2 T_p M$ of the inner
product $\bg$ on ${\cal T}^2_0$, cf.\ (\ref{innerprod on V}):
\[
{\sf g(C(F),G)=g(F,C(G))},\qquad {\sf F}, {\sf G}\in \wedge^2T_pM.
\]

Consider a unit timelike vector $\bu$. Through the tensor structure
of bivector space, the corresponding $\theta$ acts on it by
$F^{ab}\mapsto \theta^a_{~c}\theta^b_{~d}F^{cd}$. We can then repeat
the constructions of $\S$ \ref{subsec: Pre Cartan invol} replacing
$T_p M$ by $\wedge^2T_pM$ and $\bg$ by ${\sf g}$.~\footnote{Notice
that the inner product ${\sf g}$ has now the signature
$\frac{(n-1)(n-2)}{2}-(n-1)=\frac{(n-1)(n-4)}{2}$, cf.\ (\ref{ONF
bivector}). However, the map $\theta$ transforms $\{\bu\wedge
\bm_i,\bm_j\wedge\bm_k\}$ into $\{-\bu\wedge
\bm_i,\bm_j\wedge\bm_k\}$ and the corresponding inner product
$\inn{-}{-}$ on $\wedge^2T_pM$ is again Euclidean.}
In particular, the Weyl bivector operator ${\sf C}$ is viewed as a
type (1,1) tensor over $\wedge^2T_pM$ and can be decomposed into its
{\em electric} and {\em magnetic parts} ${\sf C}_\pm$, which are
also symmetric wrt ${\sf g}$. Here, ${\sf C}_\pm$ are the
endomorphisms of $\wedge^2T_pM$ obtained by replacing $C_{abcd}$ by
$(C_\pm)_{abcd}$ in (\ref{Weyl operator}). Hence, by Remark \ref{rem
selfadjoint}, ${\sf C}_+$ and ${\sf C}_-$ are the symmetric and
antisymmetric parts of ${\sf C}$ wrt $\inn{-}{-}$, respectively.
Hence, whereas it is cumbersome to say something general about the
eigenvector-eigenvalue structure of Weyl operators (in particular in
the type I/G case), for {\em purely electric} (${\sf C}={\sf C}_+$)
or {\em purely magnetic} (${\sf C}={\sf C}_-$) Weyl operators, we
have the following:

\begin{thm}\label{th_Weyl_PE} A purely electric (PE) or purely magnetic (PM)
Weyl operator is diagonalizable, i.e., a basis of eigenvectors for
$\wedge^2T_pM$ exists. A PE (PM) Weyl operator has only real (purely
imaginary) eigenvalues. Moreover, a PM Weyl
operator has at least  $s=\frac{(n-1)(n-4)}{2}$ zero eigenvalues ($s$ being the signature of ${\sf g}$).
\end{thm}
\begin{proof}
The first and second statements follow immediately from the fact
that ${\sf C}_+$ (${\sf C}_-$) are symmetric (antisymmetric) linear
operators wrt a Euclidean inner product on $\wedge^2T_pM$. To make
this more explicit and to prove the third statement, consider the
ONF \be\label{ONF bivector} {\cal B}=\{[ui]\equiv \bu\wedge
\bm_i,[jk]\equiv\bm_j\wedge\bm_k\} \ee of $\wedge^2T_pM$ induced by
the ONF $\{\bu,\bm_{i=2,...,n}\}$ of $T_pM$. Using a $(p+q)$-block
form, where $p=n-1$ and $q=\frac{(n-1)(n-2)}{2}$ ($p\leq q$ for
$n\geq 4$), the matrix representations of $\theta$ and ${\sf C}_\pm$
wrt ${\cal B}$ are
\begin{eqnarray}
[\theta]_{\cal B}=[{\sf
g}]_{\cal B}=\begin{bmatrix} -\mf{1}_p & 0 \\ 0 & \mf{1}_q
\end{bmatrix},\qquad
[{\sf C}_+]_{\cal B}=\begin{bmatrix}
S & 0 \\ 0 & T
\end{bmatrix},
 \qquad [{\sf C}_-]_{\cal B}=\begin{bmatrix}
0 & U \\ -U^t & 0
\end{bmatrix}.
\end{eqnarray}
Here $S$ and $T$ are trace-free symmetric square matrices with
components $S^{[ui]}{}_{[uj]}=-C_{uiuj}=-E_{ij}$ and
$T^{[ij]}{}_{[kl]}=C_{ijkl}$, while $U$ is a $p\times q$ matrix with
components $U^{[ui]}{}_{[jk]}=-C_{uijk}$. In the PE case, the
eigenvalues are the eigenvalues of $S$ and $T$, which are clearly
real.
In the PM case, we note that the matrix $U$ can be decomposed (using
the singular value decomposition) as $U=g_1Dg_2$, where $g_1$ and
$g_2$ are  $SO(p)$ and $SO(q)$ matrices, respectively, and $D$ is a
diagonal $p\times q$ matrix $D=\diag(\lambda_1,\lambda_2,...,
\lambda_p)$. Consequently, a PM Weyl bivector operator has
eigenvalues $\{0,...,0,\pm i\lambda_1,..., \pm i \lambda_p\}$, where
the number of zero-eigenvalues is at least
$N-2p=q-p=s=\frac{(n-1)(n-4)}{2}$. This proves the proposition.
\end{proof}

\begin{rem}\label{remark: IJM} In four dimensions, the original Petrov type classification is
equivalent with the Jordan-Segre classification of the Weyl bivector
operator (see e.g.\ \cite{Stephanibook}), where the latter is
diagonalizable iff the Petrov type is I, D or O. Referring to the
above we have $p=q=3$, the eigenvalues for $S$ and $T$ are the same
(cf.\ $\S$ \ref{subsec: Weyl EM parts}) and the conditions of the
theorem are also sufficient, i.e., if a Weyl operator is
diagonalizable and has only real (purely imaginary) eigenvalues then
it is PE (PM) wrt a certain $\bu$ (see \cite{Stephanibook,WyllVdB06}
and references therein). This can be expressed in terms of
polynomial invariants of the self-dual Weyl operator ${\sf C}_s$
acting on the 3-dimensional {\em complex} space of {\em self-dual}
bivectors: defining the quadratic and cubic invariants $I\equiv
\text{tr}({\sf C}_s^2)$ and $J\equiv \text{tr}({\sf C}_s^3)$ and the
adimensional invariant $M\equiv I^3/J^2-6$, the Weyl operator is properly PE
(PM) iff it is diagonalizable,
$M\in\R^+\cup\{\infty\}$ and $I\in\R^+_0\,(\R^-_0)$. Here $M=0$
corresponds to Petrov type D; the Petrov type I cases were
symbolized I$(M^+)$ and I$(M^\infty)$ in the extended
Petrov-classification by Arianrhod and McIntosh, where
$M=\infty\Leftrightarrow J=0$ corresponds to Petrov type I with a
zero eigenvalue~\cite{ArianrhodMcIntosh92,McIntoshetal94}.

In dimension $n>4$, it is no longer sufficient that the eigenvalues
are real (purely imaginary) in order for the Weyl operator to be PE
(PM), even if it is diagonalizable. Counterexamples to the
sufficiency for $n=5$ are provided in \cite{Coleyetal12} in the
type D case (cf.\ also { Proposition \ref{PE_PM_types} below}).
However, {\em necessary} conditions can be deduced from the fact
that, by virtue of Proposition~\ref{th_Weyl_PE}, the characteristic
equations
\begin{equation}
\label{general char eq}
    \sum_{k=0}^Na_kx^{{N}-k}=0 \qquad (a_0=1)
\end{equation}
of the operator ${\sf C}$
acting on the full bivector space $\wedge^2T_pM$ are of the form
\begin{equation}\label{char eq PEPM}
\prod_{i=1}^p(x-\lambda_i)\prod_{j=1}^q(x-\mu_j)=0\qquad \mbox{or}
\qquad x^{q-p}\prod_{i=1}^p (x^2+\lambda_i^2)=0,
\end{equation}
in the PE and PM case, respectively, where the $\lambda_i$ and
$\mu_j$ are real. Define
\be
 A_k\equiv \text{tr}({\sf C}^k) .
\ee

\begin{itemize}
\item In the PM case, for instance, one
has $a_{2l+1}=0\Leftrightarrow A_{2l+1}=0$
($2l+1\leq {N}$),
and $a_{2l}=0$ for all $l>p$ ($2l\leq {N}$), where ( cf.\ e.g.\
\cite{MacDonaldbook} or \cite{ColHer11})
\[
a_{2l}=\frac{(-1)^l}{2^ll!}\text{det}
\begin{bmatrix}A_2&2&0&\cdots&0\\A_4&A_2&4&\ddots&\vdots\\A_6&A_4&A_2&\ddots&0\\\vdots&\ddots&\ddots&\ddots&2(l-1)\\A_{2l}&\cdots&A_6&A_4&A_2
\end{bmatrix}.
\]
In addition $A_{4l+2}<0<A_{4l}$ {for properly PM}.
\item In the properly PE case: $A_{2l}>0$.
\end{itemize}
In particular, a {\em nilpotent} Weyl operator is thus neither PE
nor PM. Further necessary conditions on the $A_k$'s can be derived
along the line of \cite{ColHer11}.

\end{rem}

\subsection{Null alignment properties}\label{subsec: Weyl alignment}

In four dimensions, a properly PE or PM Weyl tensor is of Petrov type D or
I~\cite{Trumper65,Haddow95}. This follows immediately from the
Weyl-Petrov classification in terms of the operator
$Q^a{}_{b}=E^a{}_{b}+iH^a{}_{b}$ on tangent space  (defined wrt
{\em any} $\bu$), which has the same Segre type and eigenvalues as
the Weyl bivector operator $\sf{C}$~\cite{Stephanibook}. Indeed, if
a non-zero Weyl tensor is PE (or PM) wrt $\bu$, i.e., if the parts
$H_{ab}$ (or $E_{ab}$) defined wrt this $\bu$ vanish, then in any
$\bu$-ONF $\{\bu,\bm_{i=2,...,n}\}$ the non-zero part $[Q^i{}_j]$ of
the representation matrix is a real symmetric matrix (or a complex unit times
such a matrix). Thus $Q^a{}_b$, whence $\sf{C}$ is diagonalizable
and the Petrov type is I or D (cf.\ remark  \ref{remark: IJM}).

Another classification of the Weyl tensor is the one based on its
Debever-Penrose principal null directions (PNDs) which, in four
dimensions, coincides with the bivector approach. In higher
dimensions, however, both approaches are highly non-equivalent (see
\cite{Coleyetal12} for a detailed verification of this in five
dimensions). The PNDs approach was worked out in \cite{Coleyetal04}
for the Weyl tensor, leading to the concept of Weyl aligned null
directions (WANDs) replacing the PNDs and being part of the {\em
(null) alignment theory} for general tensors~\cite{Milsonetal05},
succinctly revised in section \ref{subsec: null alignment}.

 In this section we deduce the possible null alignment types for
PE/PM Weyl tensors in general $n\geq 4$ dimensions, and the
uniqueness and relative position to possible (multiple) WANDs of the
vectors $\bu$ realizing the PE/PM property. We do this in a direct
way, i.e., without relying on properties of the corresponding Weyl
bivector operator.

\subsubsection{Admitted alignment types}\label{subsub:admittedtypes}

It immediately follows from  Propositions~\ref{th_minimal_ex}
and \ref{th_minimal_ODIG} that a properly PE or PM Weyl tensor is minimal and
thus of one of the null alignment types D, I or G, in any dimension
$n>4$. However, one can be more specific by giving a different proof.

To this end, the following general observation is essential. Given a
unit timelike vector $\bu$, a {\em $\bu$-adapted null frame} is a
null frame $\{\bm_0=\bl,\bm_1=\bn,\bm_{\hat{i}=3,...,n}\}$ for which
we have
\begin{equation}
 \bu=\frac{{\mbold\ell}-\bn}{\sqrt{2}}.
 \label{u_WANDs}
\end{equation}
In {\em any} such frame, the involution (\ref{u reflection}) is
represented by
\begin{equation}\label{ln interchange}
\theta:\quad{\mbold\ell}\leftrightarrow\bn,\quad \bm_{\hat i}\mapsto
\bm_{\hat i},\;\;\forall \hat i =3,...,n,
\end{equation}
and the (passive) action hereof on a tensor $S$ simply interchanges
the frame labels $0$ and $1$. This implies that in the case
$S=S_+\Leftrightarrow \theta(S)=S$ ($S=S_-\Leftrightarrow
\theta(S)=-S$) the components of $S$ in any such frame should be all
invariant (change sign). Notice that if a null vector $\bl$ is
given, satisfying the normalization condition $l^au_a=-1/\sqrt{2}$ (but it can be otherwise arbitrarily chosen),
then (\ref{u_WANDs}) should be read as the definition
$\bn=\bl-\sqrt{2}\bu=\theta(\bl)$ of the time-reflected
$\bl$,
being a null vector lying
along the second null direction of the timelike plane
$\bu\wedge\bl$.

Conversely, if the components of a tensor $S$ in a certain null
frame $\{\bl,\bn,{\bf m}_{i=3,...,n}\}$ are invariant (change sign)
under a $0\leftrightarrow 1$ interchange, then $S=S_+$ ($S=S-$) in
the orthogonal splitting wrt the unit timelike vector
(\ref{u_WANDs}).

Applied to the Weyl tensor we obtain the following.

\begin{thm}\label{prop Weyl_PE_PM} If a Weyl tensor is PE/PM wrt $\bu$ then the following component relations hold
in any $\bu$-adapted null frame
$\{\bm_0=\bl,\bm_1=\bn,\bm_{\hat{i}=3,...,n}\}$:
\begin{eqnarray}
\label{PE cond}&\text{PE}:&\quad C_{0\hat i0\hat j}=C_{1\hat i1\hat
j},\qquad C_{0\hat
i\hat j\hat k}=C_{1\hat i\hat j\hat k},\qquad C_{01\hat i\hat j}=0,\\
\label{PM cond}&\text{PM}:&\quad C_{0\hat i0\hat j}=-C_{1\hat i1\hat
j},\qquad C_{0\hat i\hat j\hat k}=-C_{1\hat i\hat j\hat k},\qquad
C_{\hat i\hat j\hat k\hat l}=0.
\end{eqnarray}
Conversely, if a null frame
$\{\bm_0=\bl,\bm_1=\bn,\bm_{\hat{i}=3,...,n}\}$ exists for which
(\ref{PE cond}), respectively, (\ref{PM cond}) are satisfied then
the Weyl tensor is PE, respectively, PM wrt
$\bu=(\bl-\bn)/\sqrt{2}$.
\end{thm}

\begin{proof} Due to the properties (\ref{symm1}) and (\ref{tracefree}), the identities \cite{Coleyetal04}
\begin{eqnarray}
C_{010\hat i}=C_{0\hat k\hat i}{}^{\hat k}, \quad
C_{0101}=
-\tfrac{1}{2}C_{\hat{i}\hat{j}}{}^{\hat{i}\hat{j}},\quad
2C_{0(\hat{i}\hat{j})1}=C_{\hat{i}\hat{k}\hat{j}}{}^{\hat k},\qquad
2C_{0[\hat{i}\hat{j}]1}=- C_{01\hat{i}\hat{j}},\quad C_{101\hat
i}=C_{1\hat k\hat i}{}^{\hat k}
\end{eqnarray}
hold in any null frame
$\{\bm_0=\bl,\bm_1=\bn,\bm_{\hat{i}=3,...,n}\}$, and the components
of a certain b.w.\ are fully determined by the following ones:
\begin{equation}
\text{b.w.\ 2}:\; C_{0\hat i 0\hat j},\quad \text{b.w.\ 1}:\;
C_{0\hat i \hat j\hat k}, \quad \text{b.w.\ 0}:\; C_{01\hat i \hat
j},\,C_{\hat i\hat j\hat k\hat l}, \quad \text{b.w.\ -1}:\; C_{1\hat
i \hat j\hat k},\quad \text{b.w.\ -2}:\; C_{1\hat i 1\hat j}.
\end{equation}
The thesis follows from the general considerations above and by
observing that under $0\leftrightarrow 1$ the components $C_{\hat
i\hat j\hat k\hat l}$ are invariant while $C_{01\hat i \hat j}$
change sign.
\end{proof}

As a simple consequence we have
\begin{thm}
\label{PE_PM_types}
A Weyl tensor which is properly PE or PM wrt a certain $\bu$ is of alignment
type D, I$_i$ or G. In the type I$_i$ and D cases, the vector $\bu$
``pairs up'' the space of WANDs, in the sense that the second null
direction of the timelike plane spanned by $\bu$ and any WAND is
also a WAND with the same multiplicity.
Furthermore, a type D Weyl tensor is PE iff it is type D(d), and PM
iff it is type D(abc).
\end{thm}

\begin{proof}
From (\ref{PE cond}) and (\ref{PM cond}) it follows that if in a
$\bu$-adapted null frame all b.w.\ +2 components are zero, then also
all b.w.\ -2 components, and similarly for b.w.\ +1/-1 components;
if the b.w.\ 0 components additionally vanished then the Weyl tensor
would be zero (type O). Hence, WANDs of a properly PE/PM Weyl tensor (if
there exist any) must go in pairs: if $\bl$ spans a WAND then so
does $\bn$ and with the same multiplicity,
which is either 1 (type I$_i$) or 2 (type D). This proves the first
two statements. The last one follows from these considerations
and the definition of type D(d) and D(abc) (see section \ref{subsec: null alignment}).
\end{proof}

\begin{rem}
Proposition~\ref{prop Weyl_PE_PM} can be considered as an extension
of the observation in $n=4$ dimensions that the Weyl tensor is PE/PM
iff in a certain Newman-Penrose null tetrad the relations
\begin{equation}\label{PE/PM 4D conds}
\Psi_0=c\overline{\Psi}_4,\qquad \Psi_1=-c\overline{\Psi}_3,\qquad
\Psi_2=c\overline{\Psi}_2
\end{equation}
hold, where $c=+1$ in the PE and $c=-1$ in the PM case (see, e.g.,
\cite{LozCarm02}). For a Petrov type I Weyl tensor one can always
take a Weyl canonical {\em transversal} ($\Psi_0=\Psi_4\neq
0,\,\Psi_1=\Psi_3=0$) or {\em longitudinal}
($\Psi_0=\Psi_4=0,\,\Psi_1=\Psi_3\neq 0$) frame and add these to the
PE/PM conditions (\ref{PE/PM 4D conds}). Regarding type D, the last
part of Proposition \ref{PE_PM_types} is an extension of the
four-dimensional Theorem~4 of \cite{McIntoshetal94}, stating that a
Petrov type D Weyl tensor is PE (PM) iff in a canonical null frame
($\Psi_0=\Psi_1=\Psi_3=\Psi_4=0$) the scalar $\Psi_2\neq 0$ is real
(purely imaginary). Such simplifying choices have been proved
crucial for deducing classification or uniqueness results for
four-dimensional PE or PM spacetimes (see, e.g., \cite{VdBWyll06}).
\end{rem}

\begin{rem}
Spacetimes of type N (such as vacuum type N \pp waves) are usually
understood as describing transverse gravitational waves. The
interpretation of type N fields as ``radiative'' is supported, also
in higher dimensions, by the peeling behavior of asymptotically flat
spacetimes \cite{GodRea12} (in spite of significant differences with
respect to the four dimensional case, see \cite{GodRea12} and
references therein). From Proposition~\ref{prop Weyl_PE_PM} it thus
also follows that a spacetime containing {\em gravitational waves}
necessarily contains both an electric and a magnetic field
component. This resembles a well-known similar property of
electromagnetic waves, and in four dimensions was discussed, e.g.,
in \cite{Matte53,Bel62}. Conversely, we shall show below
(section~\ref{subsubsec_sheartwistfree}) that {\em static} fields
(and thus, in particular, the Coulomb-like field of the
Schwarzschild solution) are PE.
\end{rem}

\subsubsection{Uniqueness of $\bu$}\label{subsub: uniqueness}

The following facts are well known in $n=4$ dimensions (see, e.g.\
\cite{Barnes04,WyllVdB06}):
\begin{itemize}
\item if a PE/PM Weyl tensor is of Petrov type D, it is
PE/PM precisely wrt any $\bu$ lying in the plane ${\cal L}_2$
spanned by the two double WANDs (then also called principal null
directions (PNDs)~\cite{Stephanibook,Coleyetal04});
\item if a PE/PM Weyl tensor is of Petrov type I, then it is PE/PM precisely wrt the timelike Weyl principal
vector, which is unique up to sign;
\item a Weyl tensor can never
be properly PE and PM at the same time, even wrt different timelike
directions.
\end{itemize}

We shall see (Proposition \ref{prop uniqueness}) that these
properties suitably generalize to any dimension, thus giving further
support to the soundness of our PE/PM definitions. Recall that for
$n>4$, a type D Weyl tensor may have more than two double WANDs
(see, e.g., \cite{GodRea09,Durkee09,DurRea09,OrtPraPra11} for
examples). In \cite{WyllemanWANDs} it is shown that for general $n$,
the set of multiple WANDs of a type D Weyl tensor is homeomorphic to
a sphere ${\cal S}_k$, the dimension $k$ being at most $n-4$. This
is the sphere of null directions of a (proper) Lorentzian subspace
${\cal L}_{k+2}$ (the latter being defined as the space spanned by
all multiple WANDs) of the full space ${\cal L}_n$ (generated by the
full sphere of null directions ${\cal S}_{n-2}$.
However, regarding types I$_i$ and
$G$, no analog of the concept of Weyl principal vector is presently
known.

Hence, it is natural to ask whether a PE or PM Weyl tensor of type I$_i$ or G may admit a non-unique $\bu$ when $n>4$. However, we shall show that the
answer is negative.

In order to prove our results we will be considering two timelike
directions spanned by $\bu$ and $\bu'$, where $u^au'_a<0,\,\bu'\neq
\bu$. These vectors define two observers in relative motion in the
timelike plane $\bu\wedge\bu'$. Suppose that $\bl$ and $\bl'$ are
two parallel null vectors spanning the first null direction of this
plane, while the parallel null vectors $\bn$ and $\bn'$ span the
second one, such that
\begin{equation}\label{two us}
\bu=\frac{\bl-\bn}{\sqrt{2}},\qquad \bu'=\frac{\bl'-\bn'}{\sqrt{2}}.
\end{equation}
Then $\bu'=b_\lambda(\bu)$ for a certain
positive Lorentz boost (\ref{boost}), $\lambda\neq 0$, which
transforms a $\bu$-adapted null frame
${\cal F}=\{\bm_0=\bl,\bm_1=\bn,\bm_{\hat{i}=3,\ldots,n}\}$ into the
$\bu'$-adapted null frame
\begin{equation}\label{bu' boost image}
{\cal F}'=b_\lambda({\cal F})=\{\bm_{0'}=\bl'=e^{\lambda}\bl,\bm_{1'}=\bn'=e^{-\lambda}\bn,\bm_{\hat{i}=3,...,n}\}.
\end{equation}

\begin{thm}\label{prop uniqueness} A Weyl tensor $C$ at a point of a $n$-dimensional spacetime cannot be properly PE and PM at the
same time, even wrt two different timelike directions. If $C$
is properly PE or PM, then it is PE/PM precisely wrt any
$\bu$ belonging to the space ${\cal L}_{k+2}$ spanned by all multiple WANDs in the type D case, and wrt a unique $\bu$
(up to sign) in the type I$_i$ and $G$ cases.
\end{thm}

\begin{proof} Suppose that $C$ is PE/PM wrt to different timelike
directions, spanned by $\bu$ and $\bu'$ (where we take $u^au'_a<0$
and consider all 4 possibilities PE/PE, PM/PM, PE/PM and PM/PE).
Define $\bu$- and $\bu'$-adapted null frames ${\cal F}$ and ${\cal F}'$ as  above. By the PE/PM
assumptions we have
\begin{eqnarray}\label{PEPM cond bubu'}
C_{0\hat i0\hat j}=\pm C_{1\hat i1\hat j},\qquad C_{0\hat i\hat
j\hat k}=\pm C_{1\hat i\hat j\hat k},\qquad C_{0'\hat i0'\hat j}=\pm
C_{1'\hat i1'\hat j},\qquad C_{0'\hat i\hat j\hat k}=\pm C_{1'\hat
i\hat j\hat k}.
\end{eqnarray}
However, by (\ref{bu' boost image}) and the definition of boost
weight we also have
\begin{eqnarray}\label{rel bw comps}
C_{0'\hat i0'\hat j}=e^{2\lambda}C_{0\hat i0\hat j}, \qquad
C_{1'\hat i1'\hat j}=e^{-2\lambda}C_{1\hat i1\hat j}, \qquad
C_{0'\hat i\hat j\hat k}=e^{\lambda} C_{0\hat i\hat j\hat k}, \qquad
C_{1'\hat i\hat j\hat k}=e^{-\lambda} C_{1\hat i\hat j\hat k},
\end{eqnarray}
By comparison of (\ref{PEPM cond bubu'}) and (\ref{rel bw
comps}) and the fact that $e^\lambda\neq 1$ we immediately obtain
\begin{equation}
 C_{0\hat i0\hat j}=0=C_{1\hat i1\hat j} , \qquad C_{0\hat i\hat j\hat k}=0=C_{1\hat i\hat j\hat k},
\end{equation}
i.e., the type D condition is fulfilled relative to $\bl$ and $\bn$,
which thus span double WANDs. This already proves uniqueness of the
$\bu$-direction in the type I$_i$ and $G$ cases.

Next, suppose that
$C$ is of type D and PE/PM wrt $\bu$. By the second sentence in
Proposition \ref{PE_PM_types}, such a $\bu$ necessarily lies in a
plane of double WANDs and thus in ${\cal L}_{k+2}$. Conversely,
consider any other timelike direction in ${\cal L}_{k+2}$, spanned
by a vector $\bu'$ ($u^au'_a<0,\,\bu'\neq \bu$). Then, by definition
of the vector space ${\cal L}_{k+2}$, the null directions of the
timelike plane $\bu\wedge \bu'$ are double WANDs. Hence, defining
again $\bu$- and $\bu'$-adapted null frames ${\cal F}$ and ${\cal F}'$ as above, the only non-zero Weyl
components in the ${\cal F}$-frame are comprised in the b.w.\ 0 components
which are invariant or change sign under (\ref{u reflection}),
namely $C_{\hat i\hat j\hat k\hat l}$ (PE case) or $C_{01\hat i\hat
j}$ (PM case). Since the boost $b_\lambda$ leaves these components
invariant (by definition of boost weight), the same holds in the
${\cal F}'$-frame, and thus the Weyl tensor is also PE/PM also wrt $\bu'$.
\end{proof}

\begin{rem}\label{rem: uniqueness} The proof of this proposition can be readily generalized
to arbitrary tensors $S$. We notice that if $S$ is of type D (cf.\
$\S$ \ref{subsec: minimal alignment}) then the set of null directions along which the boost order of $S$ is zero is again homeomorphic to a sphere ${\cal
S}_k$ generating a Lorentzian space
${\cal L}_{k+2}$~\cite{Wyllemancomment}. We obtain that {\em any tensor
$S\neq 0$ cannot be $S_+$ and $S_-$ at the same time, even wrt
two different timelike directions. If $S=S_{\pm}$ wrt a certain
$\bu$, then either $S$ is {\em not} of type II or more
special, in which case $S=S_{\pm}$ is realized by a unique timelike direction, or $S$ is of type D, in which case  $S=S_{\pm}$  is realized by any $\bu\in{\cal
L}_{k+2}$.}
\end{rem}

\begin{rem}

\label{rem:CtypeD_PE}

More specifically for a type D Weyl tensor $C$, it also follows from
the results of \cite{WyllemanWANDs} that {\em if $C$ has more than
two double WANDs (i.e., we have $k\geq 1$ for the dimension of
${\cal S}_k$), then $C=C_+$ wrt any $\bu$ lying in ${\cal L}_{k+2}$,
i.e., $C$ is PE (type D(d)).} Let us emphasize once more that in
this case the PE property is realized precisely by any $\bu\in{\cal
L}_{k+2}$ (i.e., by any $\bu$ lying in any plane spanned by multiple
WANDs, and by no other timelike vectors); hence, since $k\leq n-3$
for any $n$ \cite{WyllemanWANDs}, {\em a Weyl tensor can never be PE
wrt {\em all} timelike directions in ${\cal L}_n$}.  By
contraposition, we have that a type D spacetime that is not PE
admits exactly two multiple WANDs. This is true, in particular, for
a type D PM Weyl tensor, which is thus PM wrt all timelike
directions in the 2-plane $\bl\wedge\bn$, and {\em only} wrt those
(i.e., $k=0$ for PM Weyl in Proposition~\ref{prop uniqueness}).

\end{rem}

\begin{rem} For PE or PM type I$_i$ Weyl tensors, the second statement of Proposition
\ref{PE_PM_types} becomes particularly meaningful when combined with
the $\bu$-uniqueness result:
any single WAND is associated to exactly one other single WAND under
the uniquely defined time-reflection $\theta$, the relation being
symmetric and where (\ref{u_WANDs}) should be read as
$\sqrt{2}\bu=\bl-\theta(\bl)=\theta(\theta(\bl))-\theta(\bl)$. This
is exemplified clearly, e.g., by the four single WANDs of static
black rings \cite{PraPra05} (which are PE, see below). In $n=4$ dimensions a Petrov
type I spacetime has always four PNDs, and it was known that these
span a 3-dimensional vector space in the PE and PM
cases~\cite{Trumper65,Narain70,McIntoshetal94}; this is now a simple
consequence of the ``pairing'' property (second statement of
Proposition \ref{PE_PM_types}).
\end{rem}

\subsection{PE spacetimes}
\label{subsec: PE spacetimes}

Large classes of PE spacetimes exist.
It is not our purpose to deduce classifications of, for instance, PE
Einstein spacetimes here;
even in four dimensions this is a very difficult task which is still
far from completion.
Instead, we mention generic conditions which imply that the
spacetimes in question are PE wrt some $\bu$. These generic
conditions hold in arbitrary dimensions and often generalize known
ones in four dimensions. Hence, this again supports the soundness of
the Weyl PE definition, cf.\ Sec.\ \ref{subsub: uniqueness}.
Evidently, all examples remain PE, with the same Weyl alignment
type, when subjected to a conformal transformation (this will be
important in section~\ref{PE_PM_Riemann}).

\subsubsection{Spacetimes with a shear-free normal $\bu$, static metrics and warps with
a one-dimensional timelike factor}

\label{subsubsec_sheartwistfree}

Given a unit timelike vector field $\bu$, we refer to (\ref{def
hab})-(\ref{uab}) of Appendix~\ref{app_congruences} for the usual
definitions of the kinematic quantities of $\bu$. In particular, a
vector field $\bu$ and the timelike congruence of curves it
generates, are called {\em shear-free} if $\sigma_{ab}=0$, and {\em
normal} (or {\em non-rotating} or {\em twist-free} or {\em
hypersurface-orthogonal}) if $\omega_{ab}=0$. We have


\begin{prop}
 \label{prop_shear_twist_free}
All spacetimes admitting a shear-free, normal
unit timelike vector field $\bu$
are PE wrt $\bu$. These are precisely the spacetimes which admit a
line element of the form
\begin{equation}
    \label{shearfree normal line element}
    \d s^2=-V^2(t,x^\gamma)\d t^2 + P^2(t,x^\gamma)\xi_{\alpha\beta}(x^\gamma)\d x^\alpha \d x^\beta.
\end{equation}
In these coordinates we have $\bu=\pa_t/V$, and the remaining kinematic
quantities are given by
\begin{equation}\label{kinem quant shearfree}
\tilde\Theta=\frac{1}{V}\, \partial_t \ln P,\qquad \dot{u}_\alpha=
\pa_\alpha \ln V.
\end{equation}
\end{prop}

\begin{proof} Eq.\ (\ref{Weyl_u}) gives the magnetic part of the Weyl tensor in
terms of the kinematic quantities. As an immediate consequence, the
existence of $\bu$ for which $\sigma_{ab}=\omega_{ab}=0$ implies
that the magnetic part vanishes and the spacetime is PE wrt $\bu$.
The proof of Since $\bu$ is hypersurface-orthogonal one has
$u_a=-V(t,x^\gamma)\d_at$ and the line-element can be written as $\d
s^2=-V(t,x^\gamma)^2\d t^2+ h_{\alpha\beta}(t,x^\gamma)\d x^\alpha
\d x^\beta$, for certain coordinates $\{t,x^\gamma\}$. Then the
shear-free property of $u^a$ translates to $\tilde\Theta
h_{\alpha\beta}=u_{\alpha;\beta}=\tfrac{1}{2V}\partial_t
h_{\alpha\beta}$ (the labels referring to coordinate components
here), whence
$h_{\alpha\beta}=P(t,x^\gamma)^2\xi_{\alpha\beta}(x^\gamma)$ (and
vice versa; this is a direct extension of the observations in
\cite{Trumper62} from four to arbitrary dimensions). The expressions
(\ref{kinem quant shearfree}) follow by direct computation.
\end{proof}

\begin{rem}\label{rem: PE and shearfree normal} One may ask the converse question: does
every PE spacetime necessarily admit one or more shear-free normal
timelike congruences? In conformally flat (type O) spacetimes the
answer is yes: there are as many of them as in Minkowski spacetime,
since the conditions $\sigma_{ab}=0,\,\omega_{ab}=0$ are conformally
invariant (see, e.g., \cite{Stephanibook}). In four dimensions,
partial answers are known for the other admitted Petrov types (D and
I). In the Petrov type D case, it was shown in \cite{WyllBeke10}
that in PE Einstein spacetimes and aligned Einstein-Maxwell
solutions there is a one-degree freedom of shear-free normal
timelike congruences.
Notorious examples
of PE type D Einstein spacetimes are the Schwarzschild and C metric
solutions (see \cite{WyllBeke10} for a complete survey). For
instance, in the interior (non-static) region $u(r)\equiv 2m/r-1
>0$ of the Schwarzschild solution $\d s^2=-\d
r^2/u(r)+u(r)\d t^2+r^2(\d\theta^2+\sin(\theta)^2\d\phi^2)$, two
particular families of shear-free normal vector fields $\bu$ are
given by
\begin{eqnarray}\label{shearfree normal int_Schw}
\bu=\sqrt{E}r\partial_r+(2m)^{1/3}\left(\frac{1}{\sqrt{u(r)q(r)}}-\frac{\sqrt{E}r}{u(r)}\right)\partial_t, \qquad
q(r)=\frac{Er^2}{u(r)}-1\pm\sqrt{\left(\frac{Er^2}{u(r)}-1\right)^2-1},
\end{eqnarray}
where, for a given $r$, the constant $E>0$ is large enough such that
$Er^2>u(r)$; shear-free normal congruences also exist in the
exterior regions, where they generalize the static observers. In
passing, we note that all Petrov type D perfect fluids with
shear-free normal fluid velocity, comprising the type D PE Einstein
spacetimes as a limiting subcase, were classified by
Barnes~\cite{Barnes73} (see also \cite{WyllBeke10} for a
clarification and a correction). However, the answer to the question
is negative in general. For instance, G\"{o}del's rotating perfect
fluid universe and the Szekeres non-rotating dust models (see
\cite{Stephanibook} and references therein), both of type D, are PE
but do not admit a shear-free normal $\bu$ (since the conditions of
proposition B.1 in \cite{WyllBeke10} are not fulfilled). In the
Petrov type I case the same is true for, e.g., the generic Kasner
vacuum spacetimes and the rotating `silent' dust models of
\cite{Wylleman08}; here the field $\bu$ realizing the PE condition
is unique (Proposition \ref{prop uniqueness}) and one verifies that
it is not shear-free normal, while
Proposition~\ref{prop_shear_twist_free} ensures that there cannot be
any other shear-free normal timelike congruences.
\end{rem}

Special cases of the spacetimes (\ref{shearfree normal line
element}) are the following {\em warped (cases (a) and (b) below),
direct product (case (c)) and doubly-warped (case (d)) spacetimes
with a one-dimensional timelike factor} (see also
\cite{PraPraOrt07}; we add the expressions of the corresponding
expansion scalar and acceleration vector between square brackets, a
prime denoting an ordinary derivative):
\begin{enumerate}[(a)]
\item $V=V(t)$, $P=P(t)$\; [$\tilde\Theta=P'(t)/(P(t)V(t)),\,\dot{u}_a=0$];
\item $V=V(x^\gamma)$, $P=P(x^\gamma)$\; [$\tilde\Theta=0,\,\dot{u}_a=\ln(V)_{;a}$];
\item $V=V(t)$, $P=P(x^\gamma)$\;[$\tilde\Theta=0,\,\dot{u}_a=0$];
\item $V=V(x^\gamma)$ {non-constant}, $P=P(t)$ { non-constant}\;[$\tilde\Theta=P'(t)/(P(t)V(x^\gamma)),\,\dot{u}_a=\ln(V)_{;a}$].
\end{enumerate}
 Notice that if $V=V(t)$, we may rescale the coordinate $t$ such
that $V=1$; if $P=P(x^\gamma)$ we can put $P=1$ by absorption in
$\xi_{\alpha\beta}(x^\gamma)$. Hence, the direct product case (c)
can be considered as a subcase of both (a) and (b). Case (d)
describes doubly-warped spacetimes; see \cite{RamVaz03} for a definition and for a discussion of their
properties in four dimensions.

It is easy to see (cf. appendix~A of \cite{OrtPraPra11} and
references therein) that for {\em Einstein spacetimes} case (a)
reduces to Brinkmann's warp ansatz \cite{Brinkmann25}
\begin{equation}
 \d s^2=-f(t)^{-1}\d t^2+f(t)\d\tilde s^2 , \qquad f(t)=\lambda t^2-2dt-b,
 \label{brinkmann_metric}
\end{equation}
where $\lambda$ is the cosmological constant (up to a positive
numerical factor), $b$ and $d$ are constant parameters and $\d\tilde
s^2$ is any $(n-1)$-dimensional Euclidean Einstein space with Ricci
scalar $\tilde {\cal R}=-(n-1)(n-2)(\lambda b+d^2)$. This can be used to
produce a number of explicit examples (see \cite{OrtPraPra11} for a
recent analysis of such warps).

Case (b) precisely covers the {\em static} spacetimes ($\bu$ being parallel to the
hypersurface-orthogonal timelike Killing vector field $\partial_t$).
In fact, the argument in the proof of  Proposition~\ref{PE_PM_types}
is a simple extension of the one used in \cite{PraPraOrt07} to prove
that static spacetimes can only be of the Weyl types O, D(d), I$_i$
or G. Let us note that in $n>4$ dimensions explicit static {\em
vacuum} solutions of the last three types are known (type O
just giving flat space): the Schwarzschild black hole (type D
\cite{Coleyetal04,Pravdaetal04,PraPraOrt07}), the static black ring
(type I$_i$ \cite{PraPra05}) and the static KK bubble (type G
\cite{GodRea09}). In four dimensions, the static type D vacua were
invariantly classified by Ehlers and Kundt~\cite{EhlersKundt} and
comprise, e.g., the exterior regions of the Schwarzschild and $C$
metrics; static type I examples are comprised in, e.g., the
Harrison metrics (see \cite{Stephanibook}).

\begin{rem}
In four dimensions, and in the line of Remark \ref{rem: PE and
shearfree normal},
the following spacetimes are necessarily {\em static} ($\bu$ being parallel to the
hypersurface-orthogonal timelike Killing vector field $\partial_t$):
\begin{itemize}

\item Petrov type D Einstein spacetimes with a non-rotating rigid $\bu$
(i.e., $\omega_{ab}=0$, $\sigma_{ab}=0=\tilde\Theta=0$)~\cite{WyllBeke10};

\item Petrov type I Einstein spacetimes with a shear-free normal $\bu$~\cite{Trumper62,Trumper65}, and type I perfect fluids with
shear-free normal fluid velocity $\bu$~\cite{Barnes73}.
\end{itemize}
\end{rem}

\begin{rem}{\em Stationary spacetimes.}
\label{rem: stationary} Although stationary PE spacetimes do exist,
and in four-dimensions have been constructed in
\cite{Collins84,Sklavenites85,Senovilla87,Senovilla92}, this is now
not the only possibility (contrary to the static case discussed
above). First, the existence of four-dimensional Petrov type I,
stationary spacetimes with a PM Weyl tensor was shown in
\cite{Arianrhodetal94}.  Moreover, {\em stationary, non-static
spacetimes are in general ``hermaphroditic'', i.e., neither PE nor
PM.} For instance, in four dimensions generic (Petrov type D)
locally rotationally symmetric (LRS) spacetimes of class I in the
Stewart-Ellis classification~\cite{StewartEllis} have this property
(cf.\ \cite{LozCarm03} for the additional PE and PM conditions). The
same is true for the exterior (Petrov type I or II) vacuum region of
a van Stockum rotating dust cilinder (if the mass per unit length is
large enough), and the (Petrov type D) Kerr metric~\cite{Bonnor95b}.
We additionally point out here that the higher dimensional
generalization of the latter, i.e.\ the black hole solution of Myers
and Perry \cite{MyePer86}, shares a similar property: the components
of the type D Weyl tensor in a canonical null frame are such that
$C_{ijkl}\neq0\neq C_{01ij}$ (see section~6.4 of \cite{PraPraOrt07}
and section~5.5 of \cite{OrtPraPra09}), such that the {spacetime} is
neither PE {[type D(d)]} nor PM {[type D(abc)]}, cf.\
Proposition~\ref{PE_PM_types} (more generally, the same comment
applies to all vacuum Kerr-Schild spacetimes with a twisting
Kerr-Schild null vector, see section~5.5.1 of \cite{OrtPraPra09}).
Moreover, the (generically type I$_i$) five-dimensional spinning
black rings of \cite{EmpRea02prl} (reducing to a Myers-Perry black
hole under an appropriate limit) are also hermaphroditic in the
non-static regions, as can be shown by making use of the Bel-Debever
criteria of Proposition \ref{prop_Bel-Debever}. These thus provide
explicit examples of spacetimes with a minimal Weyl tensor (cf.\
Proposition~\ref{th_minimal_ODIG}) which are, however, neither PE
nor PM.
\end{rem}

\subsubsection{More general direct products and warped spacetimes}

\label{subsubsec_PEwarps}

We have seen above that warped metrics with a one-dimensional
timelike factor are examples of PE spacetimes  and thus can only
be of type G, I$_i$, D(d) or O (Proposition \ref{PE_PM_types}). This latter
result was stated in Proposition 3 of \cite{PraPraOrt07}. Here we
discuss similar properties in the case of other possible warps
$(M,\bg)$ for which, by definition:
\begin{itemize}
\item $M$ is a direct product manifold $M^{(n)}=M^{(n_1)}\times
M^{(n_2)}$ of factor spaces $M^{(n_1)}$ and $M^{(n_2)}$, where
$n=n_1+n_2$, $n_1\geq 2$ and $M^{(n_1)}$ represents the Lorentzian
(timelike) factor;
\item  $\bg$ is conformal to a direct sum
metric,
\begin{equation}
\label{warped metric}
\bg=e^{2\theta}\left(\bg^{(n_1)}\oplus \bg^{(n_2)}\right),
\end{equation}
where $\bg^{(n_i)}$ is a metric on $M^{(n_i)}$ ($i=1,2$) and
$\theta$ is a smooth scalar function on {\em either} $M^{(n_1)}$ or
$M^{(n_2)}$.
\end{itemize}

Since we will be interested in PE/PM Weyl tensors of direct
products, it is useful first of all to recall a known result (see,
e.g., \cite{PraPraOrt07}) that tells us when the Weyl tensor of a
product metric vanishes (and is thus both, trivially, PE and PM):
{\em a product space is conformally flat iff both product spaces are
of constant curvature and}
\begin{equation}
\label{Riccisc_Weyldecomp}
 n_2(n_2-1){\cal R}^{(n_1)}+n_1(n_1-1){\cal R}^{(n_2)}=0.
\end{equation}

In the following analysis we shall mostly rely on the results of \cite{PraPraOrt07}.
First, combining Propositions~4 and 5 (and the explanation on
top of page~4415 of \cite{PraPraOrt07}) with our Proposition
\ref{prop uniqueness} we obtain:

\begin{prop}\label{prop: warp 2D factor} Warped spacetimes with a two-dimensional Lorentzian
factor $(M^{(n_1)},\bg^{(n_1)})$, $n_1=2$, are {at each point} either type
O, or type D and PE wrt {any} unit timelike vector living in
$M^{(n_1)}$, the uplifts of the null directions of the tangent
space to $(M^{(n_1)},\bg^{(n_1)})$ being double WANDs of the complete
spacetime $(M,\bg)$. They include, in particular, all spherically, hyperbolically
or plane symmetric spacetimes.
\end{prop}

Here and below, a vector at a point of a factor space 
is said to ``live'' in a factor space 
if
it is spanned by uplifts of tangent vectors to this space.
For warped products in which the Lorentzian factor is at least
three-dimensional the above proposition does not hold, in general.
However we can find necessary and sufficient conditions for the
product space to be PE. Let us give results in the case of direct
products
($\theta=0$ in (\ref{warped metric})). This can be then extended to warped
products {(in fact, to all conformally related spaces)} by
introducing a suitable conformal factor, which does not affect the
properties of the Weyl tensor.  For direct products there is a biunivocal
relation between vectors ${\bf v}$ tangent to $M^{(n_i)}$ and their
uplifts ${\bf v}^*$ living in $M^{(n_i)}$ ($\bf v$ being the
$M^{(n_i)}$-projection of ${\bf v}^*$).
For brevity, we shall  identify
these objects and use the same notation for them; it will be clear
from the context to what quantity we are referring. Also, we
let lowercase Latin letters serve as abstract indices for the full space
as well as for the factor spaces.
We denote by $R^{(n_i)}_{ab}$ the Ricci tensor of $M^{(n_i)}$, and
similarly for other tensors defined in the factor geometries. In
addition, given a unit timelike $\bf U$ tangent to $M^{(n_1)}$
we define a $\bf U$-ONF $\{{\bf U},{\bf m}_{A}\}$ (with frame labels
$A,B,C,\ldots=2,...,n_1$) of $M^{(n_1)}$ and an ONF $\{{\bf
m}_{I}\}$ (with frame labels $I,J,K,\ldots=n_1+1,...,n$) of $M^{(n_2)}$. These in
turn enable us to define a composite ${\bf U}$-ONF $\{{\bf U},{\bf
m}_{i=2,...,n}\}$ of $M^{(n)}$.
Then, using the results of section~4 of \cite{PraPraOrt07} we easily
arrive at
\begin{prop}
 \label{prop_products u-in-M1}
A direct product spacetime $M^{(n)}=M^{(n_1)}\times M^{(n_2)}$ is PE
wrt a unit timelike vector ${\bf U}$ that lives in $M^{(n_1)}$ iff
${\bf U}$ is an eigenvector of $R^{(n_1)}_{ab}$ and $M^{(n_1)}$ is
PE wrt ${\bf U}$, i.e.,
\begin{equation}\label{RPE cond Mn1}
R^{(n_1)}_{UA}=0, \qquad C^{(n_1)}_{UABC}=0.
\end{equation}
Then, ${\bf U}$ is also an eigenvector of the Ricci tensor $R_{ab}$
of $M^{(n)}$ (i.e., $R_{{U}i}=0$).
\end{prop}
\begin{proof}
By (9) and (10) of \cite{PraPraOrt07}, the requirements
$C_{UIAJ}=0$ and $C_{UABC}=0$ are equivalent to (\ref{RPE cond
Mn1}), while the remaining magnetic Weyl components of $M^{(n)}$
are always identically zero thanks to eq.~(8) of \cite{PraPraOrt07}.
This proves the first part. The second part follows from the
well-known fact that the Ricci tensor of a direct product is a
`product tensor' (i.e., it is decomposable), such that
$R_{UA}=R^{(n_1)}_{UA}=0$ and $R_{UI}=0$.
\end{proof}

\begin{rem}\label{rem: PE direct} The proof makes use of eq.\ (9) of \cite{PraPraOrt07},
which is only valid for $n_1\geq 3$. However, the proposition
remains true for $n_1=1$ or $n_1=2$, since then the spacetime is
always PE (cf.\ above) and the conditions (\ref{RPE cond Mn1}) are
identically satisfied indeed. Further notice that in the case
$n_1=3$ the Weyl tensor of $M^{(n_1)}$ is identically zero, such
that $M^{(n)}$ is PE wrt $\bf U$ iff the Ricci tensor of
$M^{(n_1)}$ has $\bf U$ as an eigenvector.
In general, we shall be able to rephrase this proposition once we have
introduced the concept of Riemann purely electric spacetime in the
next section.
\end{rem}

One may further wonder whether direct products exist which are PE
wrt a vector $\bu$ {\em not} living in $M^{(n_1)}$, i.e., being
inherently $n$-dimensional. Since the $M^{(n_2)}$-projection of
$\bu$ is spacelike, the $M^{(n_1)}$-projection is timelike. Thus we
have
$\bu=\cosh\gamma{\mbox{{$\bf U$}}}+\sinh\gamma{\mbox{{$\bf Y$}}}$,
where ${\mbox{{$\bf U$}}}$ is a unit timelike vector living in
$M^{(n_1)}$, ${\mbox{{$\bf Y$}}}$ a unit spacelike vector living in
$M^{(n_2)}$ {and $\gamma\neq 0$}. We also define the unit
spacelike vector ${\mbox{{$\bf y$}}}=\cosh\gamma{\mbox{{$\bf
Y$}}}+\sinh\gamma{\mbox{{$\bf U$}}}$ and use  a further adapted
$\bu$-ONF $\{\bu,\bm_i\}=\{{\bf u},{\bf m}_{A},{{\bf y}},{\bf
m}_{\tilde I}\}$, where the $(n_1-1)$ ${\bf m}_{A}$ live in
$M^{(n_1)}$ and the $(n_2-1)$ ${\bf m}_{\tilde I}$ in $M^{(n_2)}$.

\begin{prop}
 \label{prop_products u-not-in-M1}
A direct product spacetime $M^{(n)}=M^{(n_1)}\times M^{(n_2)}$ is PE
wrt a unit timelike vector $\bu=\cosh\gamma{\mbox{{$\bf
U$}}}+\sinh\gamma{\mbox{{$\bf Y$}}}$ not living in $M^{(n_1)}$
($\gamma\neq 0$, ${\bf U}$ living in $M^{(n_1)}$ and ${\bf Y}$ in
$M^{(n_2)}$) iff the following relations hold:
\begin{eqnarray}
&&\label{RPE not Mn1 a} C^{(n_1)}_{UABC}=0,\qquad R^{(n_1)}_{UA}=0,\qquad (n_1-1)R^{(n_1)}_{UAUB}=R^{(n_1)}_{UU}\delta_{AB}\\
&&\label{RPE not Mn1 b} C^{(n_2)}_{Y\tilde I\tilde J\tilde K}=0,
\qquad R^{(n_2)}_{Y\tilde I}=0,\qquad (n_2-1)R^{(n_2)}_{Y\tilde I Y \tilde J}=R^{(n_2)}_{YY}\delta_{\tilde I\tilde J}\\
&&\label{RPE not Mn1 c}(n_2-1)R^{(n_1)}_{UU}=(n_1-1)R^{(n_2)}_{YY}.
\end{eqnarray}
In particular, $M^{(n)}$ is PE wrt $\bf U$ and thus belongs to the
class described by Proposition \ref{prop_products u-in-M1}.
Moreover, it is either type O, or type D and PE wrt any $\bu$ in the
plane spanned by $\bf U$ and $\bf Y$, i.e., wrt
$\bu=\cosh\gamma{\mbox{{$\bf U$}}}+\sinh\gamma{\mbox{{$\bf Y$}}}$
for {\em any} $\gamma$.
\end{prop}

\begin{proof}
The proof goes by splitting the equations $C_{uijk}=0$ in the
adapted frame $\{\bu,\bm_i\}=\{{\bf u},{\bf m}_{A},{{\bf y}},{\bf
m}_{\tilde I}\}$ and employing eqs. (8)--(11) of \cite{PraPraOrt07}.
Requiring $C_{u\tilde I A\tilde J}=0$ and $C_{uABC}=0$ one finds
(\ref{RPE cond Mn1}) so that, by Propostion~\ref{prop_products
u-in-M1}, $M^{(n)}$ is PE also wrt to the timelike unit vector field
${\mbox{{$\bf U$}}}$ living in $M^{(n_1)}$. Direct products which
are PE wrt a vector $\bu$ not living in $M^{(n_1)}$ are thus a
subset of those considered in Proposition~\ref{prop_products
u-in-M1}. Since they are PE wrt two distinct timelike vector fields,
by Proposition~\ref{prop uniqueness} they are necessarily
of type D (unless conformally flat) and thus PE wrt any unit
timelike vector in the plane spanned by $\bu$ and ${\bf U}$ (cf.\
Proposition \ref{prop uniqueness}). Dually, $C_{uA\tilde I B}=0$
and $C_{u\tilde I\tilde J\tilde K}=0$ are equivalent to the first
two relations in (\ref{RPE not Mn1 b}). Finally, $C_{uAyB}=0$ and
$C_{u\tilde I y \tilde J}=0$ give $C_{UAUB}+C_{YAYB}=0$ and
$C_{U\tilde I U\tilde J}+C_{Y\tilde I Y\tilde J}=0$, respectively.
Tracing the first relation over $A$ and $B$ (or the second over
$\tilde I$ and $\tilde J$) yields (\ref{RPE not Mn1 c}), and then
the respective relations reduce to the last equations of (\ref{RPE
not Mn1 a}) and (\ref{RPE not Mn1 b}). Under (\ref{RPE not Mn1
a})-(\ref{RPE not Mn1 c}) the remaining Weyl magnetic components
turn out to be identically zero. This proves the proposition.
\end{proof}

Simple examples are given by spacetimes $M^{(n)}=M^{(n_1)}\times
M^{(n_2)}$ with metric $\d s^2=\d s_1^2+\d s_2^2$, with
\begin{equation}\label{ex:PE}
 \d s_1^2=-\d t^2+\d\Xi^2, \qquad \d s_2^2=\d z^2+\d\Sigma^2, \qquad \mbox{$\d\Xi^2$ and $\d\Sigma^2$ Ricci-flat Euclidean
 spaces}.
\end{equation} Here ${\bf U}=\pa_t$ and ${\bf Y}=\pa_z$, and the full
space as well as the factors are Ricci flat with decomposable Weyl
tensor.~\footnote{\label{foot: RPE Ricci_flat}  Recall (see
\cite{Ficken39,PraPraOrt07}) that a direct product space is an
Einstein space iff both factors are Einstein spaces and ${\cal
R}/n={\cal R}^{(n_1)}/n_1={\cal R}^{(n_2)}/n_2$; it has decomposable
Weyl tensor iff both factors are Einstein spaces and
$n_2(n_2-1){\cal R}^{(n_1)}+n_1(n_1-1){\cal R}^{(n_2)}=0$. Hence, a
direct product is Ricci-flat iff both factor spaces are Ricci-flat
(in which case the Weyl tensor is automatically decomposable). This
applies to the factors and thus to the full space (\ref{ex:PE}).}

\subsubsection{Spacetimes with certain isotropies}

A spacetime with a high degree of symmetry clearly has a special
Weyl tensor. In particular, an isotropy of spacetime imposes
constraints on the Weyl tensor in the sense that the isotropy must
leave the Weyl tensor invariant; consequently, a non-trivial
isotropy implies that certain components of the Weyl tensor are
zero.  Recall that the isotropy group of an $n$-dimensional
spacetime must be isomorphic to a subgroup of the Lorentz group
$SO(1,n-1)$. The largest possible isotropy group is thus
$n(n-1)/2$-dimensional, in which case the spacetime must be of
constant curvature, and therefore also conformally flat (see, e.g.,
\cite{Stephanibook}). However, some (weaker) restrictions also arise
in the presence of a smaller isotropy. In the context of PE
spacetimes, an interesting result is the following:

\begin{thm}
\label{thm_isotropy} A spacetime which admits
$SO(p_1)\times ...\times SO(p_i)\times ...\times SO(p_k)$ isotropy,
where $p_i\geq 2$ and $\sum_{i=1}^kp_i=n-1$, is PE.

\end{thm}
\begin{proof}
Consider the orthonormal frame adapted to the isotropy group as
follows: the factor $SO(p_i)$ acts on (and leaves invariant) the
$p_i$-plane spanned by ${\bm}^{a_i}$. Let $(h_i)^{a}_{~b}$ be the
corresponding projection operators onto this $p_i$-plane and define
the spatial projector $h^a_{~b}=\sum_{i=1}^k(h_i)^{a}_{~b}$.  The
action of the isotropy group can thus be put on a block-diagonal
form; explicitly, for $G=(G_i,...,G_i,...,G_k)\in  SO(p_1)\times
...\times SO(p_i)\times ...\times SO(p_k)$, the isotropy acts on a
vector $v$ as:
\[ G(v)=\sum_{i=1}^kG_ih_i(v).\] Since $\sum p_i=n-1$ there
will be a time-like vector $\bu$ so that $g_{ab}=-u_au_b+h_{ab}$.

Consider then the tensor $T_{efg}\equiv u^a
C_{abcd}h^b_{~e}h^c_{~f}h^d_{~g}$. This is a purely spatial
tensor relative to $\bu$, with components $T_{ijk}={C}_{uijk}$ in any
$\bu$-ONF, such that it is necessary and sufficient to show that
$T_{efg}=0$.
The Weyl tensor is invariant under the spacetime isotropy group, and
using the results regarding invariant tensors under the action of
$SO(p)$ groups (see  \cite{GoodmanWallach}),
the only purely spatial tensors invariant under the group in
question are {linear combinations of} tensor products of
${(h_i)}_{ab}$ and the totally antisymmetric $p_i$-tensors
${\mbold\epsilon}^i=\bigwedge_{a_i}{\bm}^{a_i}$. Since the tensor
$T$ {is} a rank 3 tensor,  it follows that $T$ must be of
the form $T=\sum_{i=1}^k \alpha_i {\mbold\epsilon}^i$,
$\alpha_i\in{\cal F}_M$. Hence, $T_{efg}=T_{[efg]}$. However, due to the
first Bianchi identity (last equation in (\ref{symm1})) we have
$u^aC_{a[bcd]}=0$, whence $T_{[efg]}=0$, which proves the
proposition.
\end{proof}

Special instances of the above isotropy are the following.
\begin{itemize}

\item Spacetimes with an isotropy group $SO(n-1)$. They are conformally flat (see,
e.g., Theorem~7.1 of \cite{ColHer09}), i.e.,
Proposition~\ref{thm_isotropy} becomes ``trivial'' if we take one
single $SO(n-1)$ factor. If the spacetime is not of constant
curvature, the $SO(n-1)$ isotropy and the conformal flatness imply
that the Ricci tensor has Segre type $\{1,(11...1)\}$. Also, the
shear, rotation and acceleration of the preferred vector field $\bu$
must vanish, while the surfaces of the foliation orthogonal to $\bu$
are maximally isotropic and thus have constant
curvature~\cite{Stephanibook}. It follows that the spacetimes with
an isotropy group $SO(n-1)$ are given by the line elements
\begin{equation}
\label{FLRW extension}
 \d s^2=-\d t ^2+a(t)^2\d\Omega^{2}_{n-1,k}(x^1,...,x^{n-1}),
\end{equation}
where $\d\Omega^{2}_{n-1,k}$ is the metric on a $(n-1)$-dimensional
``unit'' space of constant curvature with sign $k$. Notice that they
are special instances of the warped metrics (\ref{shearfree normal
line element}), case (a).
For $n=4$ this gives
the Friedmann-Lema\^{i}tre Robertson-Walker (FLRW) model, which is
in fact the only possibility to satisfy the isotropy condition of
Proposition~\ref{thm_isotropy}. However, in all higher dimensions
spacetimes satisfying the assumptions of
Proposition~\ref{thm_isotropy} and admitting non-zero Weyl tensors
are possible (as is generically the case in the next examples).

\item In even dimensions, a possible isotropy is
$SO(3)\times SO(2)^{(n-4)/2}$. This is admitted, for example, by the
metric
\begin{equation}
 \d s^2=-\d t^2+a(t)^2
 \d\Omega^2_{3,k}(x,y,z)
 +\sum_{i=1}^{(n-4)/2}b_i(t)^2\d\Omega^{2}_{2,k_i} (y^1_i,y^2_i) \qquad (n \mbox{
 even}),
 \label{isotr_even}
\end{equation}
where the submanifolds $\{y^1_i\,\text{constant},\,y^2_i\,\text{
constant}\}$ clearly have a four-dimensional FLRW line element.

\item Similarly in odd dimensions, take all $p_i=2$, i.e., the isotropy group $SO(2)^{(n-1)/2}$.
An example is the line-element
\begin{equation}
 \d s^2=-\d t^2+\sum_{i=1}^{(n-1)/2}a_i(t)^2\d\Omega^{2}_{2,k_i} (v^1_i,v^2_i)
 \qquad (n \mbox{ odd}).
 \label{isotr_odd}
\end{equation}
In the case where all the $a_i(t)$ coincide, metric
(\ref{isotr_odd}) is a special subcase of~(\ref{shearfree normal
line element}) with (a) such that, in particular, Einstein
spacetimes are thus comprised.
\end{itemize}

One can easily construct other examples admitting different
isotropies compatible with Proposition~\ref{thm_isotropy}. Notice
that the above proposition could also be reexpressed in terms of
symmetries of the Weyl tensor alone, since the proof does not rely
on the presence of isometries. Other theorems regarding Weyl
tensors with large symmetry groups were deduced in \cite{ColHer09}
and serve to produce further examples of PE spacetimes.

We already mentioned that only zero Weyl tensors (and thus
conformally flat spacetimes) can admit $SO(n-1)$ isotropy
(Theorem~7.1 of \cite{ColHer09})). Next, Theorem~7.2 of
\cite{ColHer09} states that
\begin{prop}
If the Weyl tensor of a spacetime of dimension $n> 4$ admits
$SO(n-2)$ isotropy, then it is of type O or D(bcd), and thus PE.
\end{prop}

\begin{rem} The statement of the proposition is no longer valid for $n=4$, the counterexamples
being then precisely all non-PE Petrov type D spacetimes ({\em any}
four-dimensional type D Weyl tensor has boost isotropy in the plane
spanned by the PNDs, and spin isotropy in the plane orthogonal to it
\cite{Stephanibook}). For instance, Petrov type D Einstein
spacetimes (such as the Kerr solution), or their aligned
Einstein-Maxwell `electrovac' generalizations~\cite{Stephanibook}
are generically not PE (see also Remarks \ref{rem: PE and shearfree
normal} and \ref{rem: stationary}). We also observe that a metric
whose associated Weyl tensor is of Petrov type D (which thus admits
the above mentioned isotropies) is itself, nevertheless, generically
anisotropic. However, even if the spacetime itself (and not only the
Weyl tensor) is $SO(2)$-isotropic (i.e., LRS \cite{Stephanibook})
then it is still not necessarily PE: the LRS class I and III metrics
are generically not PE (nor PM; see  \cite{LozCarm03} for the PE and
PM conditions).

Yet, the LRS class II metrics, i.e., those admitting spherical,
hyperbolical or planar symmetry (in addition to the $SO(2)$ metric
isotropy), are {\em all PE}, just as their higher-dimensional
generalizations
\begin{equation}
\label{spher}
\d s^2=F(t,x)^2(-\d t^2+\d x^2)+G(t,x)^2\d\Omega^{2}_{n-2,k} \qquad (n\geq 4).
\end{equation}
For $n>4$ these are examples of the above proposition where (both
the Weyl tensor and) the metric itself admits $SO(n-2)$ isotropy,
and are special instances of Proposition \ref{prop: warp 2D factor}
(they include, in particular, the Schwarzschild(-Tangherlini)
metric, and its generalizations to include a cosmological constant
and/or electric charge).

\end{rem}

For arbitrary $n$ an $SO(n-3)$ isotropic Weyl tensor does {\em not}
require the spacetime to be PE (nor PM), in general: take, for
instance, the five-dimensional Myers-Perry spacetime (and, more
generally, see Theorem~7.4 of \cite{ColHer09}).

Finally, we note that the $2k+1$-dimensional spacetimes with
$U(k)$-symmetry ($k>1$) given in Theorem~7.5 of \cite{ColHer09} are
also PE and ``PM'', in the terminology of Remark \ref{rem: PM4D}.

\subsubsection{Higher-dimensional ``Bianchi type I'' spacetimes} \label{BTI}
We can generalize the well-known Bianchi type I spacetimes to
$n$-dimensions by \emph{a spacetime allowing for $(n-1)$-dimensional
space-like hypersurfaces, $\Sigma_t$, possessing a transitive
isometry group equal to the Abelian $\mathbb{R}^{n-1}$}.  Such
spacetimes will also be PE:
\begin{prop}\label{prop:BTI}
An $n$-dimensional spacetime possessing an Abelian $\mb{R}^{n-1}$ group of isometries acting transitively on space-like hypersurfaces is PE.
\end{prop}

\begin{proof}

Let us present two different proofs of this. First,
consider  the family of spatial hypersurfaces, $\Sigma_t$,
defined as the orbits of the Abelian $\mb{R}^{n-1}$. We choose $\bu$
to be the Gaussian normal to $\Sigma_t$. Consequently, $\bu$ is
vorticity-free, $u_{[a;b]}=0$ and geodesic, $u^{b}u_{a;b}=0$. Using
equation (\ref{Weyl_u}) in Appendix we see that the magnetic
components of Weyl reduce to:
\[   C^{dg}{}_{bc}u_dh^b_{~e}h^c_{~f}=  2h^{ag}h^b_{~e}h^c_{~f}\sigma_{a[b;c]}+\frac{2}{n-2}h^g_{~[e} h^b_{~f]}\sigma^a_{\ b;a}.
\]
Choosing a $\bu$-ONF consisting of left-invariant spatial vectors,
$\bm_{\hat i}$ in $\Sigma_t$ in the standard way
\cite{Milnor76,Hervik2010-Left,EllMac69,Hervik2002-5D}, the
commutators satisfy $[\bm_{\hat i},\bm_{\hat j}]=0$ due to the fact that
$\mb{R}^{(n-1)}$ is Abelian. In addition, $[\bu,\bm_{\hat i}]$ is
tangent to the hypersurfaces due to the fact that this is an ONF.
This further implies that the following connection coefficients are
zero: $ u_{a}\Gamma^a_{bc}u^c=\Gamma^{\hat i}_{\hat j\hat k}=0.$ An
explicit computation now gives that
$h^{ag}h^b_{~e}h^c_{~f}\sigma_{a[b;c]}=0$ and $\sigma^a_{\ b;a}=0$;
consequently, this spacetime is PE.

A second proof of Proposition \ref{prop:BTI} can also be given using
symmetries. The Abelian $\mb{R}^{n-1}$ implies also that we can, in
a suitable frame, write the metric as: \beq \d s^2=-\d
t^2+\sum_{i=1}^{(n-1)}a_i(t)(\d x^i)^2, \eeq where $\d t$ is the
dual one-form to the Gaussian normal vector $\bu$ above. Here, it is
obvious that the discrete map $\phi:(t,x^i)\mapsto (t,-x^i)$ is an
isotropy for a point with $x^i=0$. Since this space is spatially
homogeneous, this $\phi$ extends to an isotropy at any point in
space. Consider the point $p$ at the origin of $\Sigma_t$. Then it
is straightforward to see that $\phi$ gives rise to the map
$\phi^*=-\theta$ on $T^*_pM$, where $\theta$ is the Cartan
involution. Since this is an isotropy at any point, this must extend
to an isotropy of the Weyl tensor $C$ as well as all other curvature
tensors. Since $\phi^*=-\theta$, on $T^*_pM$ this implies that for a
curvature tensor, $T$, of rank $N$, we have the condition
$(-1)^N\theta(T)=T$. Hence, for the Weyl tensor, which is of rank 4,
$\theta(C)=C$, and consequently, $C=C_+$ and thus PE.

\end{proof}

Examples of such spacetimes have been considered in arbitrary
dimensions, for example, in \cite{Hervik2001-Discrete} (here, the
full group of discrete symmetries were considered).

\subsubsection{Type D spacetimes with more than two multiple WANDs}\label{subsub: morethan2WANDs}

Higher-dimensional type D spacetimes with more than two multiple
WANDs are PE (see Remark~\ref{rem:CtypeD_PE}). For instance, in
\cite{DurRea09} it was proved that all type D Einstein spacetimes
which admit a non-geodesic field of multiple WANDs over a region
necessarily posses more than two multiple WANDs at each point of
that region, and all five-dimensional such spacetimes were
explicitly listed. See also \cite{GodRea09,OrtPraPra11} for more
explicit examples.

\subsection{PM spacetimes}

\label{subsec: PM spacetimes}

Contrary to PE spacetimes, properly PM spacetimes are most elusive.
For instance, in four dimensions the only known Petrov type D
PM spacetimes are LRS and were obtained in \cite{LozCarm03}.
For $n=4$ we refer to \cite{Lozanovski07,LozWyll11} for recent
deductions of Petrov type I$(M^\infty)$ and I$(M^+)$ PM spacetimes
(cf.\ Remark \ref{remark: IJM}), and to \cite{WyllVdB06} for a
complete overview of the
PM literature prior to these investigations. Here we underline the
elusiveness of PM spacetimes in any dimension, by proving
Propositions \ref{thm: PM vacua D} and \ref{prop_PMproducts}; this
also supports the soundness of the Weyl PM definition. However, the
work of \cite{Lozanovski07} will enable us to construct examples of
higher-dimensional (non-vacuum) PM spacetimes in
section~\ref{subsubsec_RPM}.

\subsubsection{Restrictions on Einstein spacetimes}

In a frame approach to four- or higher-dimensional General
Relativity, the requirement of a spacetime to obey certain geometric
conditions
puts constraints on the closed Einstein-Ricci-Bianchi system of
equations. This may give rise to severe integrability conditions,
leading to non-existence or uniqueness results. Regarding the PM
condition in four dimensions it was shown, e.g., in \cite{VdBWyll06}
that PM, Petrov type D, {\em aligned} perfect fluids, i.e., for
which the Weyl tensor is PM wrt the fluid velocity, are necessarily
LRS and thus also comprised in the work of \cite{LozCarm03}. As
another example, aligned PM irrotational dust spacetimes have been
shown not to exist, irrespective of the Petrov
type~\cite{Wylleman06}. In the same line severe integrability
conditions arise for PM Einstein spacetimes (including the Ricci-flat
case), and up to now no such solution has been found, in any
dimensions. For $n=4$ it was therefore conjectured in
\cite{McIntoshetal94} that {\em no congruence of observers in an
Einstein spacetime exist which measures the Weyl tensor to be PM}.
Up to present a general proof has not been found, but the validity
of the conjecture was shown under a variety of additional
assumptions (see again \cite{WyllVdB06} for an overview), among
which the Weyl type D assumption~\cite{Hall73,McIntoshetal94}.
This last result can be generalized to arbitrary dimension:

\begin{thm}\label{thm: PM vacua D}
In any dimension, Einstein spacetimes with a type D, PM Weyl tensor
do not exist.
\end{thm}

\begin{proof}

Assume that a PM type D Einstein spacetime exists. Take a null frame
$(\bl,\bn,\bm_{\hat i})$ for which $\bl$ and $\bn$ span the (unique)
double WANDs. We work with the generalization of the
Geroch-Held-Penrose formalism introduced in \cite{Durkeeetal10}. In
the notation of \cite{Durkeeetal10}, the PM type D Einstein space
conditions translate into the vanishing of all curvature tensor
components, except for
$\Phi^A_{\hat{i}\hat{j}}=\Phi_{[\hat{i}\hat{j}]}\neq 0$ and possibly
$\phi_{\hat{k}\hat{k}}=\phi=\Lambda$ (no summation over $\hat{k}$,
$\Lambda$ being the cosmological constant up to
normalization). We denote $\boldsymbol{\Phi}$ for the matrix
$[\Phi^A_{\hat{i}\hat{j}}]$, and
$\mathbf{S}\equiv[\rho_{(\hat{i}\hat{j})}]$ and $\mathbf{A}\equiv
[\rho_{[\hat{i}\hat{j}]}]$ for the symmetric, resp.\ antisymmetric
part of the matrix $\boldsymbol{\rho}=[\rho_{\hat{i}\hat{j}}]$.
Since $\Phi\neq 0$ both double WANDs are geodetic by Proposition 6
of \cite{PraPraOrt07},
such that we can take the
simplified Ricci (`NP') and Bianchi equations displayed in Appendix
A of \cite{Durkeeetal10}, of which we shall only need (A.5-6) and
(A.10-13).

By considering the symmetric part of (A.10), and the symmetric and
antisymmetric parts of the $jk$-contraction of (A.11) we immediately
get
\begin{eqnarray}
&&\boldsymbol{\Phi S}=\boldsymbol{S\Phi},\qquad \boldsymbol{\Phi
A}=-\boldsymbol{A\Phi},\\
\label{Phi SA comm}&&\boldsymbol{\Phi}\left(\boldsymbol{S}+2\boldsymbol{A}-\tfrac
12 \rho \boldsymbol{1}_{n-2}\right)=0,\quad \rho\equiv \rho^{\hat
i}{}_{\hat i}=S^{\hat i}{}_{\hat i}.
\end{eqnarray}
Let $2p\geq 2$ be the rank of the antisymmetric matrix
$\boldsymbol{\Phi}$. Then, by rotation of the $\bm_{\hat i}$ we can
put $\Phi$ in normal $2\times 2$ block form
$\left[\begin{smallmatrix}x&0\\0&0\end{smallmatrix}\right]$, where
$x$ is an  antisymmetric, 2-block diagonal, invertible
$2p\times2p$ matrix. Write $\boldsymbol{S}$ and $\boldsymbol{A}$ in
the same kind of block form:
$\boldsymbol{S}=\left[\begin{smallmatrix}s_1&s_2\\s_2^t&s_3\end{smallmatrix}\right],\,
\boldsymbol{A}=\left[\begin{smallmatrix}a_1&a_2\\-a_2^t&a_3\end{smallmatrix}\right]$,
where $y^t$ is the transpose of $y$, and $s_1$ and $s_3$ are
symmetric whereas $a_1$ and $a_3$ are antisymmetric.  Performing the
matrix multiplication in (\ref{Phi SA comm}) in $2\times 2$ block
form and using the invertibility of $x$ one gets
\begin{equation}\label{conseq}
s_1+2a_1-\tfrac{\rho}{2}\boldsymbol{1}_{2p}=0,\;s_2+2a_2=0\quad\Rightarrow\quad
s_2=a_2=0,\quad a_1=0,\quad s_1=\tfrac 12 \rho\boldsymbol{1}_{2p},
\end{equation}
by taking symmetric and antisymmetric parts.
Next, taking $\hat{i}=1,\,\hat{j}=2$ and $\hat{k},\hat{l}>2p$ in
(A.11) produces $a_3=0$, whence $\mathbf{A}=0$. Now (A.12),
with $\hat{j}=1,\,\hat{k}=2$ gives $S_{\hat i\hat l}=0$, $\forall
\hat i$ and $\forall \hat l>2$, implying that either $\mathbf{S}=0$,
or $p=1$ and $s_3=0$ (notice that the latter is compatible with
the last equation in (\ref{conseq}) and $\rho=S^{\hat i}{}_{\hat
i}=(s_1)^{\hat i}{}_{\hat i}$). We conclude that
$\boldsymbol{\rho}=\left[\begin{smallmatrix}\tfrac\rho 2
\boldsymbol{1}_{2} & 0\\0&0\end{smallmatrix}\right]$. Priming the
above reasoning leads to
$\boldsymbol{\rho'}=\left[\begin{smallmatrix}\tfrac{\rho'}{2}
\boldsymbol{1}_{2}& 0\\0&0\end{smallmatrix}\right]$, such that
\begin{equation}
\label{rhorho'} \boldsymbol{\rho\rho'}=\boldsymbol{\rho'\rho}.
\end{equation}
Put $T_i\equiv\tau_{\hat i}-\tau'_{\hat i}$. Adding (A.13) to its
prime dual gives
\begin{equation}\label{B5}
\Phi_{\hat{i}[\hat{j}}T_{\hat{k}]}-T_{\hat
i}\Phi_{\hat{j}\hat{k}}=0.
\end{equation}
Tracing over $\hat{i}$ and $\hat{k}$ leads to $T^{\hat
i}\Phi_{\hat{i}\hat j}=0$; contracting now (\ref{B5}) with $T^{\hat
i}$ implies $T_{\hat i}=0$, i.e., $\tau_{\hat i}=\tau'_{\hat i}$.
Finally, subtracting (A.6) from (A.5), using (\ref{rhorho'}) and
taking the antisymmetric part yields the desired contradiction
$\boldsymbol{\Phi}=0$.
\end{proof}

\subsubsection{PM direct products}\label{subsub: PM direct products}

Similarly as in the PE case above, we now deduce necessary and
sufficient conditions for a product spacetime
$M^{(n)}=M^{(n_1)}\times M^{(n_2)}$ to be (properly) PM. Again, the
results can be translated immediately to, e.g., warped spacetimes. We
use the notation and conventions of $\S$ \ref{subsubsec_PEwarps}.

Firstly, recall that for $n_1\leq 2$ a direct product is PE, and
thus cannot be properly PM by the first sentence of
Proposition \ref{prop uniqueness}. Secondly, suppose that $M^{(n)}$
is properly PM wrt $\bu=\cosh\gamma{\mbox{{$\bf
U$}}}+\sinh\gamma{\mbox{{$\bf Y$}}}$, where ${\bf U}$ and ${\bf Y}$
live in the respective factor spaces. Consider the vector
${\mbox{{$\bf y$}}}=\cosh\gamma{\mbox{{$\bf
Y$}}}+\sinh\gamma{\mbox{{$\bf U$}}}$ and the composite $\bu$-ONF
$\{\bu,\bm_A,{\bf y},\bm_{\tilde I}\}$. By requiring $C_{y\tilde
IA\tilde J}=0$ and $C_{yABC}=0$ and using eqs.~(8)--(10) of
\cite{PraPraOrt07} one finds
$\sinh(\gamma)R^{(n_1)}_{UA}=\sinh(\gamma)C^{(n_1)}_{UABC}=0$. If
$\sinh\gamma\neq 0$ it would follow from
Propostion~\ref{prop_products u-in-M1} that $M^{(n)}$ is PE wrt
${\mbox{{$\bf U$}}}$, again in contradiction with Proposition
\ref{prop uniqueness}. Thus $\bu={\bf U}$ (i.e., $\gamma=0$). Thirdly, we state the
following lemma, which is proved by simple substitution in eqs.\
(9)--(11) of \cite{PraPraOrt07}; here and henceforth a composite
${\bf U}$-ONF $\{{\bf U},\bm_i\}=\{{\bf U},\bm_A,\bm_I\}$ is used.

\begin{lem} Let $M^{(n)}=M^{(n_1)}\times M^{(n_2)}$ be a direct
product spacetime with $n_1\geq 3$ and ${\bf U}$ a unit timelike
vector living in $M^{(n_1)}$. If
\begin{itemize}
\item
the Ricci tensor of $M^{(n_1)}$ is of the form
\begin{equation}\label{Rn1 conds}
R^{(n_1)}_{ab}=\frac{\Rc^{(n_1)}}{n_1}g^{(n_1)}_{ab}+U_{(a}q_{b)},\quad
U^aq_a=0\qquad \left(\mbox{i.e., } R^{(n_1)}_{AB}=\frac{\Rc^{(n_1)}}{n_1}\delta_{AB}
,\quad R^{(n_1)}_{UU}=-\frac{\Rc^{(n_1)}}{n_1}\right);
\end{equation}
\item
$M^{(n_2)}$ is an Einstein space:
\begin{equation}\label{Rn2 conds}
R^{(n_2)}_{ab}=\frac{\Rc^{(n_2)}}{n_2}g^{(n_2)}_{ab}\qquad
\left(R^{(n_2)}_{IJ}=\frac{\Rc^{(n_2)}}{n_2}\delta_{IJ}\right);
\end{equation}
\item the Ricci scalars of the factors are related by~(\ref{Riccisc_Weyldecomp}),
\end{itemize}
then the only possibly non-zero Weyl components of $M^{(n)}$ are
\begin{eqnarray}
&&C_{UIAJ}=-\frac{1}{n-2}R^{(n_1)}_{UA}\delta_{IJ},\qquad
C_{UABC}=C^{(n_1)}_{UABC}-\frac{2n_2}{(n-2)(n_1-2)}g^{(n_1)}_{A[B}R^{(n_1)}_{C]U},\label{PMcomps}\\
&&C_{UAUB}=C^{(n_1)}_{UAUB},\qquad C_{ABCD}=C^{(n_1)}_{ABCD},\qquad
C_{IJKL}=C^{(n_2)}_{IJKL}.\label{PMconds}
\end{eqnarray}
\end{lem}

Notice that under the conditions of the lemma $\Rc^{(n_1)}$ is
constant, as actually follows from the decomposability of the Ricci
scalar~\cite{Ficken39} and eq.\ (\ref{Riccisc_Weyldecomp}) an se. We
can now prove:

\begin{prop}\label{prop_PMproducts} A direct product spacetime $M^{(n)}=M^{(n_1)}\times M^{(n_2)}$ is PM wrt
a unit timelike vector ${\bf U}$ that lives in $M^{(n_1)}$ iff the
following conditions hold:
\begin{itemize}
\item[(a)] $M^{(n_1)}$ is PM wrt ${\bf U}$ and has Ricci tensor of the form
(\ref{Rn1 conds});
\item[(b)] $M^{(n_2)}$ is of constant curvature;
\item[(c)] the Ricci scalars of the factors are related by
(\ref{Riccisc_Weyldecomp}), i.e., $n_2 (n_2-1)\Rc^{(n_1)}+n_1(n_1-1) \Rc^{(n_2)}=0$.
\end{itemize}
In this case the Weyl (magnetic) components are given by
(\ref{PMcomps}). Moreover, if $M^{(n_1)}\times M^{(n_2)}$ is
properly PM wrt $\bu$
then $n_1\geq 3$ and $\bu$ necessarily lives in 
$M^{(n_1)}$.
\end{prop}

\begin{proof} Above we already proved the last sentence. Conditions
(a)-(c) are precisely those of the lemma, augmented by the vanishing
of the right hand sides in (\ref{PMconds}); hence (a)-(c) is
sufficient for the spacetime to be PM. Conversely, suppose that
$M^{(n)}$ is PM wrt ${\bf U}$, i.e., $C_{ijkl}=0$.
Expressing
$C^{IJ}{}_{IJ}=0$ and using eq.\ (11) of \cite{PraPraOrt07} one
immediately finds (\ref{Riccisc_Weyldecomp}). Next, $C_{AIBJ}=0$ and
eq.\ (9) of \cite{PraPraOrt07} yield for $A\neq B$, $I\neq J$ and
$A=B,\,I=J$ that $R^{(n_1)}_{AB}=0$, $R^{(n_2)}_{IJ}=0$ and
$(n-1)(R^{(n_1)}_{AA}+R^{(n_2)}_{II})=\Rc^{(n_1)}+\Rc^{(n_2)}$,
respectively. Summing the last relation over $I$ and separately over
$A$, and using (\ref{Riccisc_Weyldecomp}), one arrives at (\ref{Rn1
conds}) and (\ref{Rn2 conds}). This proves (c) and the Ricci part of
(a) and (b). Using the lemma, $M^{(n_1)}$ must be PM and $M^{(n_2)}$
conformally flat
since the left hand sides in (\ref{PMconds}) vanish, and the
remaining Weyl magnetic components are (\ref{PMcomps}).
\end{proof}

Notice that the PM condition for a direct product (or conformally
related) spacetime is much more stringent than the PE condition,
cf.\ Propositions \ref{prop_products u-in-M1} and \ref{prop_products
u-not-in-M1}. In addition, from footnote~\ref{foot: RPE Ricci_flat}
and (\ref{Riccisc_Weyldecomp}) it follows that {\em a product space
$M^{(n)}=M^{(n_1)}\times M^{(n_2)}$ is a properly PM Einstein space
iff it is the direct product of a properly PM Ricci flat spacetime
$M^{(n_1)}$ and a flat $M^{(n_2)}$}.

\section{The Ricci and Riemann tensors: Riemann purely electric (RPE) or magnetic  (RPM) spacetimes}

\label{PE_PM_Riemann}

Similarly as done above for the Weyl tensor, one can naturally
define the electric and magnetic parts of the Ricci and Riemann
tensors, and  deduce properties
of spacetimes which possess a purely electric or magnetic (Ricci or)
Riemann tensor. This is studied in the present section.

\subsection{Definitions and PE/PM conditions}

In accordance with Definition 3.3
we define:

\begin{defn}\label{def RPE/RPM} The electric part of the Ricci [Riemann]
tensor wrt $\bu$ is the tensor $(R_+)_{ab}$ [$(R_+)_{abcd}$]. The
Ricci tensor is called PE  (wrt $\bu$) if $R_{ab}=(R_+)_{ab}$. The
Riemann tensor or a spacetime is called {\em Riemann purely
electric} or {\em RPE} (wrt $\bu$) if $R_{abcd}=(\tilde
R_+)_{abcd}$. The definitions of a PM Ricci tensor and a {\em
Riemann purely magnetic (RPM)} Riemann tensor or spacetime are
analogous, replacing $+$ by $-$.
\end{defn}

Based on (\ref{Riemann decomp}) we have, in any $\bu$-ONF,
the following component relations between the different parts:
\begin{eqnarray}
 & & (C_+)_{ijkl}=(R_+)_{ijkl}-\frac{2}{n-2}(\delta_{i[k}(R_+)_{l]j}-\delta_{j[k}(R_+)_{l]i})+\frac{2{\cal R}}{(n-1)(n-2)}\delta_{i[k}\delta_{l]j}, \label{Cijkl} \\
 & & (C_+)_{uiuj}=(R_+)_{uiuj}+\frac{1}{n-2}\left\{(R_+)_{ij}-\left[(R_+)_{uu}+\frac{{\cal R}}{n-1}\right]\delta_{ij}\right\}, \label{Cuiuj} \\
 & & (C_-)_{uijk}=(R_-)_{uijk}-\frac{2}{n-2}\delta_{i[k}(R_-)_{j]u}. \label{Cuijk}
\end{eqnarray}

It is easy to see (cf. also \cite{Senovilla01}) that the independent
electric Riemann components consist of $n(n-1)/2$ components
$(R_+)_{uiuj}$ and $n(n-1)^2(n-2)/12$ components $(\tilde
R_+)_{ijkl}$, while the magnetic ones of $n(n-1)(n-2)/3$ components
$(R_-)_{ujik}$ (recall the index symmetries and the cyclicity).

From (\ref{Cijkl})--(\ref{Cuijk}) [or (\ref{comm +- contr}) and
(\ref{Riemann decomp})] it follows that

\begin{prop}\label{prop RPE/PM} The Riemann tensor is RPE (RPM) wrt $\bu$
iff both the corresponding Weyl and Ricci tensors are PE (PM) wrt
$\bu$, or vanish (but not both at the same time).
\end{prop}

In four dimensions, the RPE part of this proposition was proven in
\cite{Trumper65}. For the RPM part, the focus has usually been on
the weaker condition mentioned in remark \ref{rem RPEPM 4D} below,
for which the theorem does not hold in this form.

Just as for the Weyl tensor (Proposition~\ref{prop_Bel-Debever}) one
can easily derive PE/PM Bel-Debever criteria for the Ricci and
Riemann tensors. Regarding the latter, one can either make a
separate analysis (the only difference with the Weyl tensor being
that (\ref{tracefree}) does not hold, or use Propositions
\ref{prop_Bel-Debever} and \ref{prop RPE/PM}).

\begin{prop}(Ricci and Riemann PE/PM Bel-Debever criteria) Let $\bu$ be a unit timelike
vector.
Then a Ricci tensor $R_{ab}\neq
0$ is
\begin{itemize}
\item PE wrt $\bu$ iff $\;R_{ui}=0\;$ in a $\bu$-ONF, i.e.,
\begin{equation}\label{Bel Ricci PE}
    u_{[a}R_{b]c}u^c=0; 
\end{equation}
\item PM wrt $\bu$ iff $\;R_{uu}=R_{ij}=0\;$ in a $\bu$-ONF, i.e.,
\begin{equation}\label{Bel Ricci PM}
    R_{ab}u^au^b=u_{[a}R_{b][c}u_{d]}=0. 
\end{equation}
\end{itemize}
A Riemann tensor $\Rt_{abcd}\neq 0$ is
\begin{itemize}
\item RPE wrt $\bu$ iff $\;\Rt_{uijk}=0\Leftrightarrow C_{uijk}=R_{ui}=0\;$ in a
$\bu$-ONF, i.e.,
\begin{equation}\label{Bel Riemann PE}
    u_ag^{ab}\Rt_{bc[de}u_{f]}=0 \quad\Leftrightarrow \quad
    u_ag^{ab}C_{bc[de}u_{f]}=u_{[a}R_{b]c}u^c=0;
\end{equation}
\item RPM wrt $\bu$ iff $\;\Rt_{ijkl}=\Rt_{uiuj}=0\Leftrightarrow C_{ijkl}=R_{uu}=R_{ij}=0\;$ in a
$\bu$-ONF, i.e.,
\begin{equation}\label{Bel Riemann PM}
   u_{[a}\Rt_{bc][de}u_{f]}=\Rt_{abcd}u^bu^d=0 \quad\Leftrightarrow \quad
   u_{[a}C_{bc][de}u_{f]}=R_{ab}u^au^b=u_{[a}R_{b][c}u_{d]}=0. 
\end{equation}
\end{itemize}
\label{prop_Rie/Ric Bel-Debever}
\end{prop}

The PE/PM criteria for the Ricci tensor can be stated alternatively
in terms of conditions on the {\em Ricci operator} on tangent space:
\begin{equation}
{\sf R}:\;v^a\mapsto R^a{}_bv^b,
\end{equation}
for which the ${\sf R}_+$ and ${\sf R}_-$ parts wrt a unit timelike
$\bu$ have the following $1+(n-1)$ block form representations in any
$\bu$-ONF ${\cal F}_u=\{\bu,\bm_{i=2,...,n}\}$:
\begin{equation}\label{Ricci_op+-}
[{\sf R}_+]_{{\cal
F}_u}=\left[\begin{smallmatrix}-R_{uu}&0\\0&R^{sp}\end{smallmatrix}\right],\qquad
[{\sf R}_-]_{{\cal F}_u}=\left[\begin{smallmatrix}0&\alpha
q^t\\-\alpha q&0\end{smallmatrix}\right].
\end{equation}
Here $\alpha\in {\cal F}_M$, $R^{sp}$ is a real symmetric matrix
with components $(R^{sp})^i_j=R_{ij}$ and $q$ a unit column vector
($q^tq=1$). Write ${\bf q}\equiv q^i\bm_i\leftrightarrow q^a\equiv
q^im_{i}{}^a$. From Definition \ref{def RPE/RPM}, (\ref{Ricci_op+-})
and the classification of Ricci-like tensors into types A1, A2, A3
and B (see appendix A), we readily obtain
\begin{prop} A Ricci tensor $R_{ab}\neq 0$ is
 \begin{itemize}
    \item PE wrt $\bu$ iff ${\sf R}$ has $\bu$ as an eigenvector, ${\sf
    R}(\bu)=-R_{uu}\bu$. In this case all eigenvalues of ${\sf R}$
    are real.
    \item PM wrt $\bu$ iff it has the structure
    \begin{equation}\label{Ricci PM}
R_{ab}=2\alpha\, u_{(a}q_{b)}, \quad\alpha\neq 0,\quad
q^aq_a=1,\quad q^au_a=0,
\end{equation}
In this case ${\sf R}$ has eigenvalues 0 (multiplicity $n-2$) and
$\pm i\alpha\neq 0$, with corresponding eigenvectors $\bu\pm i{\bf
q}$. In particular the Ricci scalar vanishes, $\Rc=0$ (in agreement
with the general result of Proposition \ref{prop: PM contr}.)
    \end{itemize}
\label{prop_PEPM_Ricci}
\end{prop}

\begin{rem}
\label{rem_Einstein_PE}
The resemblance with Proposition \ref{th_Weyl_PE} is striking.
From appendix A we still have that {\em a Ricci tensor is PE iff it
has a timelike eigenvector, i.e., iff is of type R1.} In particular, in any dimension
{\em all proper Einstein spacetimes ($0\neq R_{ab}\sim g_{ab}$) have a
properly PE Ricci tensor.} In four dimensions this is also true for perfect fluids (Segre type $\{1,(111)\})$ and
Einstein-Maxwell fields (Segre type $\{(1,1)(11)\}$). However, Ricci tensors of types R3 or R4
(see Appendix A) have only real eigenvalues but are not PE.

A properly PM
Ricci tensor is a special instance of type R2. Referring to Remark
\ref{remark: IJM}, the mentioned eigenvalue properties are
equivalent with a characteristic equation for ${\sf R}$ of the form
$(x^2+\alpha^2)x^{n-2}=0$, which is equivalent to
\begin{equation}
\label{PM trace conds}
\text{tr}({\sf R}^{2k-1})=0,\qquad \text{tr}({\sf
R}^{2k})=\frac{\left(\text{tr}({\sf
R}^{2})\right)^k}{2^{k-1}},\qquad \text{tr}({\sf R}^{2})\neq 0
\qquad (k=1,...,\lfloor (n+1)/2 \rfloor).
\end{equation}
Conversely, if a Ricci operator satisfies these conditions, then it
is of type R2 with eigenvalues 0 (multiplicity $n-2$) and $\pm
i\alpha,\,\alpha=\sqrt{-\text{tr}({\sf R}^{2})/2}\neq 0$. Let ${\bf
v}_{\pm}\equiv \bu\pm i{\bf q}$ be corresponding eigenvectors of
$\pm i\alpha$ where, by multiplication with a complex scalar, we can
normalize $\bu$ to be unit timelike. Taking the real and imaginary
parts of ${\sf R}({\bf v}_{\pm})=\pm i\alpha {\bf v}_\pm$ we get
$R_{ab}u^b=-\alpha q^a$ and $R_{ab}q^b=\alpha u_a$. By considering
$q^aR_{ab}u^b$ and in view of the symmetry of $R_{ab}$ we obtain
$q^aq_a=1$.
However, in general $\bu$ and ${\bf q}$ are {\em not orthogonal},
but if they are then (\ref{Bel Ricci PM}) holds. We conclude that
{\em a Ricci tensor is properly PM iff (\ref{PM trace conds}) holds and the
real and imaginary parts of an eigenvector with non-zero eigenvalue
are orthogonal.}
\end{rem}

\begin{rem}
Replacing the Weyl by the Riemann tensor in (\ref{Weyl operator})
one gets the definition of the {\em Riemann bivector operator}. From
(\ref{Riemann decomp}) and the above results for the Ricci operator
it is easy to check that Proposition \ref{th_Weyl_PE} still holds
when replacing the Weyl by the Riemann tensor (the proof of the
proposition being independent of the tracefree property
(\ref{tracefree})).
\end{rem}

\begin{rem}\label{rem RPEPM 4D} (a) In the four-dimensional literature, a spacetime has been called ``RPM'' or ``RPE'' (the quotes being part of the name) if
\begin{equation}\label{RPM/PE 4D}
\Rt_{acbd}u^cu^d=0,\qquad \tfrac 12
\varepsilon_{acef}\Rt^{ef}{}_{bd}u^cu^d=0,
\end{equation}
respectively (see, e.g., \cite{Arianrhodetal94,Haddow95}). In a
$\bu$-ONF these become the respective sets of conditions
$\Rt_{uiuj}=0$ and $\Rt_{uijk}=0$. Hence, ``RPE'' coincides with our
RPE notion, whereas this is not the case for ``RPM'': there are no
restrictions on $\Rt_{ijkl}$ in the first of (\ref{RPM/PE 4D}),
i.e., it does not cover the $u_{[a}\Rt_{bc][de}u_{f]}=0$ part of
(\ref{Bel Riemann PM}), whence ``RPM'' is weaker than RPM  (in the
terminology of \cite{Senovilla00,Senovilla01}, the EE part of the
Riemann (Weyl) tensor vanishes, but not necessarily the HH part, cf.
Remark~\ref{remark: Senovilla}). This is analogous to the ``PM''
notion for the Weyl tensor (Remark \ref{rem: PM4D}).

(b) From~(\ref{Cuiuj}) one immediately deduces the following generalization
of Theorem~5 in \cite{Haddow95} from four to arbitrary dimensions,
wherein we also define a ``PM'' Ricci tensor.
\begin{prop}
Any two of the following three conditions imply the third:
\begin{enumerate}
\item[(i)] the Riemann tensor is ``RPM'', i.e., $\Rt_{acbd}u^cu^d=0$  ($\Rt_{uiuj}=0$);
\item[(ii)] the Weyl tensor is ``PM'', i.e., $C_{acbd}u^cu^d=0$  ($C_{uiuj}=0$);
\item[(iii)] the Ricci tensor is ``PM'', i.e., it has the form
\begin{equation}\label{Ricci 4D PM Haddow}
R_{ab}=u_{(a}q_{b)}+\frac{\Rc}{n-1}h_{ab},\quad u^aq_a=0\qquad
\left(R_{uu}=0,\,R_{ij}=\frac{\Rc}{n-1}\delta_{ij}\right).
\end{equation}
\end{enumerate}
\end{prop}
\noindent Hence, by comparison of (\ref{Ricci PM}) with (\ref{Ricci
4D PM Haddow}) a Ricci tensor is PM iff it is ``PM'' with vanishing
Ricci scalar.
In four dimensions, the spacetimes satisfying (i)-(iii), dubbed
`Haddow magnetic'~\cite{Lozanovski07}, are Weyl PM and ``RPM'', but
not RPM (in general). In \cite{Lozanovski07} a family of such
spacetimes was deduced, the RPM members
being given by the metrics (\ref{RPM metric}) below and giving rise
to RPM spacetimes in higher dimensions (section
\ref{subsubsec_RPM}). Examples satisfying (i) but not (ii), and vice
versa, were discussed in \cite{Arianrhodetal94}.

(c) Whereas the conjunction of the RPE and RPM conditions (even wrt
different timelike directions, see Propositions \ref{prop
uniqueness} and \ref{prop RPE/PM} below) only leads to flat
spacetime, the
RPE and ``RPM'' conditions {\em can} be realized wrt the same $\bu$.
This occurs iff
\begin{equation}\label{Ruabc=0}
 \Rt_{abcd}u^d=0 \quad (\Leftrightarrow \Rt_{uijk}=0=\Rt_{uiuj}\; \mbox{   in a $\bu$-ONF}).
\end{equation}
Vacuum spacetimes $R_{ab}=0$ satisfying (\ref{Ruabc=0}) are flat in
four dimensions (since $C_{uijk}=0\Leftrightarrow H_{ab}=0$ and
$C_{uiuj}=0\Leftrightarrow E_{ab}=0$) but can be non-trivial in
{five} or higher dimensions (see \cite{Ortaggio09} and section
\ref{subsub: RPE spacetimes} below). Finally notice that the trace
of (\ref{Ruabc=0}) gives $R_{ua}=0$ such that proper Einstein
spacetimes ($0\neq R_{ab}\sim g_{ab}$) are not allowed. (See also point~\ref{item_RPERPM} of section~\ref{subsub: RPE spacetimes}.)
\end{rem}

\subsection{Null alignment types}

For the Ricci and Riemann tensors we immediately get the following
analogue of Proposition \ref{prop Weyl_PE_PM}:

\begin{thm}
\label{prop_Ricci_PE_PM}
A Ricci tensor is PE/PM wrt $\bu=(\bl-\bn)/\sqrt{2}$ if in some $\bu$-adapted null frame
$\{\bm_0=\bl,\bm_1=\bn,\bm_{\hat{i}=3,...,n}\}$ the following
relations hold (in which case they hold in {any} such frame):
\begin{eqnarray}
\label{RicPE cond}&\text{PE}:&\quad R_{00}=R_{11},\qquad R_{0\hat
i}=R_{1\hat i}; \\
\label{RicPM cond}&\text{PM}:&\quad R_{00}=-R_{11},\qquad R_{0\hat
i}=-R_{1\hat i},\qquad R_{\hat i\hat j}=R_{01}=0.
\end{eqnarray}
A Riemann tensor is RPE (RPM) wrt $\bu=(\bl-\bn)/\sqrt{2}$ if in
some $\bu$-adapted null frame the relations (\ref{PE cond}) and
(\ref{RicPE cond}) ((\ref{PM cond}) and (\ref{RicPM cond})) hold (in
which case they hold in {any} such frame).
\end{thm}

From the beginning of section \ref{subsub:admittedtypes} it follows
that the only admitted alignment types of a PE or PM Ricci tensor,
and of a RPE or RPM Riemann tensor, are G, I$_i$, D or O (in the
terminological convention at the end of section \ref{subsec: null
alignment}). However, we will see that this can be further
constrained in the (R)PM case. Also, if the type is I$_i$ or G then
the vector $\bu$ realizing the (R)PE/(R)PM property is unique,
whereas it can be any vector in the Lorentzian space ${\cal
L}_{k+2}$ if the type is D. Moreover, the properties of being
properly PE and PM cannot be realized at the same time (cf.\
Proposition \ref{prop uniqueness} and Remark \ref{rem: uniqueness}).

A properly PE Ricci tensor can be of any of the types G, I$_i$ or D.
For instance, it is easy to see that $R_{ab}=\alpha(l_al_b+n_an_b)$
(with $\alpha\neq 0$) is of type G, while
\beq
 R_{ab}=\alpha_{\hat i}[(l_a+n_a)m^{\hat i}_{b}+m^{\hat
i}_{a}(l_b+n_b)]+R_{01}(l_an_b+n_al_b)\qquad (R_{01}\neq 0,\quad
\text{at least one}\; \alpha_{\hat i}\neq0)
\label{example_Ricci_Ii}\eeq
is of type I$_i$.~\footnote{Checking this is trivial for the type G
example. For the type I$_i$ example, consider a generic null
rotation~(\ref{null_rot}). In the new frame one finds
$R_{0'0'}=-z_{\hat i} R_{0'\hat i'}-(|z|^2/2+1)z_{\hat
i}\alpha_{\hat i}$ and $R_{0'\hat i'}=\alpha_{\hat
i}(-z^2/2+1)+z_{\hat i}(R_{01}-z_{\hat j}\alpha_{\hat j})$. The
existence of a doubly aligned null direction $\bl'$ requires
$R_{0'0'}=0=R_{0'\hat i'}$, which leads to $\alpha_{\hat i}z_{\hat
i}=0=2z_{\hat i}R_{01}+\alpha_{\hat i}(2-|z|^2)$. By contracting the
latter equation with $z_{\hat i}$ gives $|z|^2=0$, i.e., all
$z_{\hat i}=0$. This then implies $\alpha_{\hat i}=0$, leading to a
contradiction. Therefore in this case there do not exist any doubly
aligned null directions, so that the type is indeed I$_i$, as
claimed.} We also have

\begin{prop}\label{prop Ricci type I is Ii} If a Ricci-like tensor (over a vector space of dimension $n$) is of type I at a point, then it is of type
I$_i$ and possesses at least a $(n-3)$-dimensional surface of single
ANDs.
\end{prop}

\begin{proof} Recall that under a null rotation about a null vector $\bn$ with parameter
$z\equiv (z_{\hat i})\equiv (z^{\hat i})\in\R^{n-2}$, a null frame
$\{\bm_0=\bl,\bm_1=\bn,\bm_{\hat i}\}$ transforms to
$\{\bm_{0'}=\bl',\bm_{1'}=\bn,\bm_{\hat{i}'=3,...,n}\}$, where
($|z|^2=z^{\hat i}z_{\hat i}$):
\begin{equation}
 \bl'=\bl-z^{\hat i}\bm_{\hat i}-\frac{|z|^2}{2}\bn, \qquad \bn'=\bn, \qquad \bm_{\hat i}'=\bm_{\hat i}+z_{\hat
 i}\bn.
 \label{null_rot}
\end{equation}
In the new frame, the b.w.\ +2 component is given by
\begin{equation}\label{R0'0'}
R_{0'0'}=R_{00}-2z^{\hat i}R_{\hat i0}+z^{\hat i}z^{\hat j}R_{\hat
i\hat j}-|z|^2R_{01}+|z|^2 z^{\hat i} R_{1\hat
i}+\frac{|z|^4}{4}R_{11}.
\end{equation}
If $R_{ab}$ is of type I it possesses a single AND. Let it be
spanned by $\bn$. Hence, we have $R_{11}=0$ and
$\{R_{1\hat{i}}\}\neq \{0\}$. We may rotate the spatial frame
vectors such that $R_{13}\neq 0$ while all other $R_{1\hat{i}}$'s
vanish. The null vector $\bl'$ spans another single AND iff it
satisfies the alignment equation $R_{0'0'}=0$. By (\ref{R0'0'}) this
is a cubic equation in $z$, where the cubic term is $|z|^2 z_3
R_{1\hat 3}$. For any fixed value $(z^0_4,z^0_5,\ldots,z^0_{n})\in
\R^{n-3}$ we get a cubic equation in $z_3$, which has thus at least
one real solution $z^0_3$ depending continuously on
$(z^0_4,z^0_5,\ldots,z^0_{n})$.
\end{proof}

\begin{prop}\label{prop: PMpeculiar} Any type D Ricci-like tensor $R_{ab}$ is PE. For a type D PM Weyl
tensor $C_{abcd}$, any symmetrized rank 2 contraction of an odd
power vanishes:  $\text{Tr}_{4m+1}(C^{2m+1})_{(ab)}=0$.
\end{prop}
\begin{proof} The first statement is trivially seen by taking a null frame $\{\bl,\bn,\bm_{\hat i}\}$ where
 $\bl$ and $\bn$ are double aligned null vectors, such that $R_{ab}=R_{01}(l_an_b+n_al_b)+R_{\hat i\hat
j}m^{\hat i}_{a}m^{\hat i}_{b}$, with  $(R_{01},\{R_{\hat i\hat
j}\})\neq (0,\{0\})$. The second statement is a consequence of
Proposition \ref{prop: PM contr}, Corollary \ref{cor: type II-D
prod-contract} and the first statement.
\end{proof}

From Remark \ref{rem: uniqueness} it thus follows that a properly
PM Ricci tensor cannot be of type D. More specifically we have

\begin{prop}
\label{prop: PM Ricci} A PM Ricci tensor is of alignment type I$_i$
(i.e., types G and D are forbidden) and has a $(n-3)$-dimensional
sphere of single ANDs paired up by the unique unit timelike $\bu$
realizing the PM condition.
\end{prop}

\begin{proof}
Putting $q_a=\bm^2_a$ in (\ref{Ricci PM}) and
defining $\bl$ and $\bn$ by (\ref{null frame corr}) we have
\[
R_{ab}=\alpha (l_al_b-n_an_b),\quad \a\neq 0,
\]
cf.\ (\ref{R2 can form}). In a null frame
$\{\bm_0=\bl,\bm_1=\bn,\bm_{\hat{i}=3,...,n}\}$ the only surviving
components are $R_{11}=-R_{00}=\alpha\neq 0$. Under a null rotation
about $\bn$ with parameter $z\in\R^{n-2}$ the
components of positive b.w.\ in the new frame are given by
\[
R_{0'0'}=R_{00}+\frac{|z|^4}{4}R_{11}=\frac{\alpha}{4}(|z|^4-4),\qquad
R_{0'\hat{i}'}=-\alpha\frac{|z|^2}{2}z_{\hat i}\neq 0,
\]
such that $R_{ab}$ has a continuous infinity of single ANDs given by
all $\bl'=\bl'(z)$ satisfying the alignment equation $R_{0'0'}=0$,
i.e., $|z|=\sqrt{2}$. Thus the set of single ANDs is a
$(n-3)$-dimensional sphere. The vector $\bu=(\bl-\bn)/\sqrt{2}$
appearing in (\ref{Ricci PM}) is the unique vector realizing the PM
condition and pairs up this sphere of single ANDs.
\end{proof}

In general, the alignment types of a RPE/RPM Riemann tensor are
subject to Corollary \ref{cor: Riemann align types}. In the RPM
case, combination hereof with Proposition~\ref{prop: PM Ricci}
immediately implies:

\begin{cor} If, at a point, a spacetime
is RPM wrt $\bu$ and the Ricci tensor is non-zero, then the RPM
Riemann tensor is of alignment type I$_i$ or G (i.e., type D is
forbidden). In particular, $\bu$ realizing the RPM property is always unique. \label{cor_RPM}
\end{cor}

 More specifically, Proposition \ref{prop: Riemann boost order}
holds. For instance, if a Riemann tensor is PM wrt $\bu$ and  of
type I$_i$, then $\bu$ can be written as (\ref{u_WANDs}), where
$b_{\text{Rie}}(\bl)=b_{\text{Rie}}(\bn)=1$, and both the
corresponding Ricci and Weyl tensors are PM wrt $\bu$, where
$\max(b_{\text{Ric}}(\bl),b_C(\bl))=\max(b_{\text{Ric}}(\bn),b_C(\bn))=1$.
If a Riemann tensor is PE wrt $\bu$ and of type D, and the Ricci and
Weyl tensors are non-zero, then vectors $\bl$ and $\bn$ exist for
which (\ref{u_WANDs}) holds and along which the boost orders of the
Riemann, Ricci and Weyl tensors are all zero. Finally, from
Proposition \ref{prop: Riemann boost order}, Remark
(\ref{rem:CtypeD_PE}) and the fact that a type D Ricci tensor is
automatically PE, it follows that {\em if a type D Riemann tensor
has more than two double aligned null directions then it is RPE}.

\subsection{Direct products and explicit examples}

The first part of the following proposition is a restatement of
Proposition~\ref{prop_products u-in-M1}, while the second part is an
immediate consequence of Proposition~\ref{prop_PMproducts}.

\begin{prop}
\label{prop_RPE_RPMproducts}
Let $M^{(n)}=M^{(n_1)}\times M^{(n_2)}$
be a direct product spacetime and ${\bf U}$  a timelike vector that
lives in $M^{(n_1)}$. Then
\begin{itemize}
 \item $M^{(n)}$ is RPE wrt ${\bf U}$ iff $M^{(n_1)}$ is RPE wrt ${\bf
 U}$ (which is the case iff $M^{(n)}$ is PE wrt ${\bf U}$);
 \item $M^{(n)}$ is RPM wrt ${\bf U}$ iff $M^{(n_1)}$ is RPM wrt ${\bf
 U}$ and $M^{(n_2)}$ is flat.
\end{itemize}
\end{prop}

Recall that ${\bf U}\in M^{n_1}$ is not a restriction in the (R)PE
case, and is the only possibility in the PM or RPM cases. It is thus
evident that RPE/RPM spacetimes have a special significance in the
construction of higher-dimensional (R)PE or (R)PM spacetimes, e.g.,
from those already known in four dimensions.

\begin{rem} Regarding Proposition \ref{prop_products u-not-in-M1}, one may define the electric and magnetic parts of
the Weyl, Ricci and Riemann tensors of the Riemannian space
$M^{(n_2)}$ relative to a spacelike vector ${\bf Y}$, analogously as
for timelike vectors.
Doing so, the
duality in eqs.\ (\ref{RPE not Mn1 a})- (\ref{RPE not Mn1 b})  is
manifest. The first two equations in (\ref{RPE not Mn1 a}) and
(\ref{RPE not Mn1 b}) tell that a direct product
$M^{(n)}=M^{(n_1)}\times M^{(n_2)}$ which is (R)PE wrt a unit
timelike vector $\bu$ not living in $M^{(n_1)}$ must have factors
which are RPE wrt the respective normalized projections of $\bu$;
the last two equations are relations between electric tensors and
can be covariantly rewritten as
\[
U^cU^d\left(\Rt^{(n_1)}_{acbd}-\tfrac{1}{n_1-1}
h^{(n_1)}_{ab}R^{(n_1)}_{cd}\right)=0,\qquad
Y^cY^d\left(\Rt^{(n_2)}_{acbd}-\tfrac{1}{n_2-1}
h^{(n_2)}_{ab}R^{(n_2)}_{cd}\right)=0,
\]
where $(h^{(n_i)})^a{}_b$ is the projectors in $M^{(n_i)}$
orthogonal to ${\bf U}$ ($i=1$) or ${\bf Y}$ ($i=2$).
\end{rem}

\subsubsection{RPE spacetimes}

\label{subsub: RPE spacetimes}

We mention generic conditions under which spacetimes are RPE,
thereby taking section \ref{subsec: PE spacetimes} as a thread.

\begin{enumerate}

\item Spacetimes admitting a shear-free normal unit
timelike vector field $\bu$ are RPE wrt $\bu$ iff moreover the
expansion scalar of $\bu$ is spatially homogeneous, i.e.,
$h_a{}^b\tilde\Theta_{,b}=0$. This follows from (\ref{Riem_u}) and
(\ref{Ricci_u}). Referring to (\ref{shearfree normal line element})
this is the additional condition $\tilde\Theta=\tilde\Theta(t)$;
integrating the first equation in (\ref{kinem quant shearfree}) this
is precisely the case if $P=e^{\int V(t,x^\gamma)\tilde\Theta(t)dt}$
(after absorbing the function of integration into $\xi(x^\gamma)$).
Then $\bu=\partial_t/V$ is an eigenvector of the Ricci tensor with
eigenvalue
$-\dot{u}^a{}_{;a}+(n-1)(\dot{\tilde\Theta}+\tilde\Theta^{2})$, see
(\ref{Ruu}). Special instances are spacetimes admitting a
non-rotating rigid $\bu$ ($\sigma_{ab}=\tilde\Theta=\omega_{ab}=0$)
and the warped spacetimes with a one-dimensional timelike factor,
i.e., cases (a)-(c) in section \ref{subsubsec_sheartwistfree}. In
particular, {\em all static spacetimes are RPE}. In contrast,
doubly-warped spacetimes with a one-dimensional timelike factor
(case (d)) are PE {\em but never RPE} wrt $\bu$.

\item \label{item_RPERPM}
Spacetimes which satisfy (\ref{Ruabc=0}) are RPE and ``RPM'',
cf.\ Remark \ref{rem RPEPM 4D}(c). Within the warped class (a) of
\ref{subsubsec_sheartwistfree}, where $\dot{u}^a=0$ additionally,
this is realized iff $\dot{\tilde{\Theta}}=-\tilde\Theta^2$, see
(\ref{Riem_u}). Examples hereof are the direct product spacetimes of
the subclass (c), i.e., those spacetimes admitting a covariantly
constant unit timelike vector field $\bu$
($\sigma_{ab}=\tilde\Theta=\omega_{ab}=\dot{u}^a=0$), and the $n\geq
5$ warped spacetimes (\ref{brinkmann_metric}) with $\lambda=0$
(vacuum case).

\item {\em All direct or
warped products (\ref{warped metric}), with a RPE timelike factor
and with $\theta:M^{(n_2)}\rightarrow\R$, are RPE}. This follows
from Proposition \ref{prop_RPE_RPMproducts} and, e.g., eqs.\ (25) in
\cite{CarotdaCosta93} or (D.8) in \cite{Waldbook}. For $n_1=1$ (giving
case (a) of section \ref{subsubsec_sheartwistfree}) and $n_1=2$, the
RPE condition on the timelike factor is automatically satisfied
since then $R^{(n_1)}\sim\bg^{(n_1)}$ (see also the top of page 4415
of \cite{PraPraOrt07}, and cf.\ Remark \ref{rem: PE direct}). As an
instance of $n_1=4$ RPE spacetimes we may mention aligned perfect fluids
(for which the Weyl tensor is PE wrt the
fluid velocity $\bu$) and their Einstein space limits;
for instance, all examples mentioned at the end of
Remark \ref{rem: PE and shearfree normal} can be lifted by the above
direct or warped product construction.

\item All spacetimes with an isotropy group mentioned in Proposition
\ref{thm_isotropy} are in fact RPE (in the proof no use was made of
the tracefree property (\ref{tracefree}) of the Weyl tensor, just as
the $2k+1$-dimensional spacetimes with $U(k)$ isotropy ($k>1$)). The
spacetimes (\ref{spher}) with spherical, hyperbolical or planar
symmetry are RPE
iff the matrix
$\left[\begin{smallmatrix}R_{01}&R_{11}\\R_{00}&R_{01}\end{smallmatrix}\right]$
is of type R1 in appendix \ref{subsec: Ricci app}, relative to a
null frame $\{\bm_0=\bl,\bm_1=\bn,\bm_{\hat i}\}$ where $\bl$ and
$\bn$ live in the timelike factor $M^{(n_1)}$, $n_1=2$ (see, e.g.,
eqs.\ (25) in \cite{CarotdaCosta93}). This is precisely the case when $R_{00}R_{11}>0$ or $R_{00}=R_{11}=0$, the latter case implying Ricci type D.

\item Higher-dimensional ``Bianchi type I'' spacetimes, studied in section \ref{BTI}, are also RPE spacetimes. Again, this can be shown in two different ways; however, restricting to the second proof in section \ref{BTI} one sees that the discrete symmetry implies for the Ricci tensor, $\theta(R)=R=R_+$; consequently, the spacetime is  RPE.

\item All PE Einstein spacetimes are obviously also RPE (cf. Remark~\ref{rem_Einstein_PE}).

\end{enumerate}

\subsubsection{RPM spacetimes}

\label{subsubsec_RPM}

Evidently, RPM spacetimes are even more elusive than Weyl PM
spacetimes. The only (2-parameter) class of RPM spacetimes known so
far was derived in four dimensions by Lozanovski~\cite{Lozanovski07} (cf.\ also Remark~\ref{rem RPEPM 4D}(b)), the line element being (up to a
constant rescaling).
\begin{equation}
    \label{RPM metric}
    \d s_L^2=\exp(-2bz)\left[-dt^2+dz^2\right]+\exp(2ay)\left[dx^2+t^2x^2dy^2\right] \qquad (a,b\in\mathbb{R}) .
\end{equation}
This {{spacetime, which contains an ``imperfect fluid'' \cite{Lozanovski07},}} is RPM wrt $\bu=\exp(bz)\pa_t$, and of Petrov type
I$(M^+)$ for all values of $a$ and $b$, except when {$ab=0$}, in which
case the type is I$(M^\infty)$, cf.\ Remark~\ref{remark: IJM}.
According to Proposition~\ref{prop_RPE_RPMproducts}, explicit
examples of higher-dimensional RPM spacetimes can be produced by
taking direct products with flat Euclidean spaces. Additionally,
Weyl PM (but not RPM) spacetimes  can be generated from such direct
products by simply performing a (non-trivial) conformal
transformation (under which the Weyl tensor is invariant while the
Ricci tensor will loose its PM character, in general). For the sake
of definiteness, consider the five-dimensional line-element
\begin{equation}
 \d s^2=e^{kz}(\d s_L^2+\d w^2) ,
 \label{5D_PM}
\end{equation} with $\d s_L^2$ given by (\ref{RPM metric}). This is a spacetime
PM wrt $\bu=e^{(2b-k)z/2}\pa_t$.  It is, additionally, RPM (wrt the
same $\bu$) iff it is a direct product, i.e., $k=0$ (the necessity of this follows from the last statement of
Proposition~\ref{prop_PEPM_Ricci} and computation of the Ricci scalar
$\Rc=-3k^2e^{z(2b-k)}$, while the sufficiency follows from the second part of Proposition \ref{prop_RPE_RPMproducts}). In the latter case one has
$\Rt_{\alpha\beta\gamma w}=0$, so that $\exp(bz)\pa_t\pm\pa_w$ are
null directions aligned with the Riemann tensor when $k=0$ (thus,
the Riemann tensor is of type I$_i$ in this case). {\em A fortiori},
these are also WANDs (cf.\ Proposition \ref{prop: Riemann boost
order}), so that the Weyl tensor can not be of type G. Moreover, a direct computation shows that for this metric, the
symmetric rank 2 tensor
\begin{equation*}
T_{ag}=T_{(ag)}\equiv C_{abcd}C^{cd}{}_{ef}C^{efb}{}_g
\end{equation*}
does not vanish. Hence, by Proposition \ref{prop: PMpeculiar} the
Weyl tensor cannot be type D. Since the case $k\neq 0$ is just
obtained by a conformal transformation, it follows that {\em all
metrics (\ref{5D_PM}) are of Weyl type I$_i$, and thus PM uniquely wrt $\bu=e^{(2b-k)z/2}\pa_t$.}

To our knowledge, such products are the only examples of
higher-dimensional (R)PM spacetimes found so far.

\section{Conclusion and discussion} We introduced and elaborated a two-fold decomposition of any tensor at a point of a spacetime of arbitrary
dimension, relative to a unit timelike vector $\bu$. The splitting
is based on considering a (time) reflection of $\bu$, which itself
is a special instance of a Cartan involution (when applied to the
Lorentz group). We saw that this leads to a generalization, from
four to arbitrary dimensions, of the electric/magnetic decomposition
of the Maxwell and Weyl tensors. That this generalization is natural
has been confirmed by the extension of many four-dimensional results
regarding purely electric and magnetic curvature tensors to higher
dimensions.

In particular, we derived a close connection between purely
electric/magnetic properties and the existence of preferred null
directions. Hereby we focussed on the curvature
tensors, so crucial in (four- or higher dimensional) General
Relativity as well as in other gravity theories. However, many of these
properties generalize to arbitrary tensors and operators; as such
they are applicable to any physical theory governed by tensor
objects defined over a spacetime (manifold with Lorentzian metric),
with the potential of leading to novel interesting viewpoints and
results in such contexts.

Tensors for which one of the two parts in the splitting wrt $\bu$
vanishes are examples of tensors which are minimal wrt $\bu$, in the
sense that the sum of squares of the tensor components in any
$\bu$-adapted orthonormal frame is not larger than for any other
$\bu'$. Via a new proof of the alignment theorem we made an
intriguing connection with both the null alignment and polynomial
invariants properties of such tensors: these are precisely the
tensors characterized by their invariants or, still, the tensors
which do {\em not} possess a unique aligned null direction of boost
order $\leq 0$. Future inquiries on these facts may be important for
shedding new light on the invariant content of many modern theories
(string theory, brane world models, quantum cosmology, etc). In
particular, the classification of spacetimes themselves makes use of
the Riemann tensor and its covariant derivatives via the
Cartan-Karlhede algorithm, and thus may highly benefit from such
investigations.

This paper also demonstrates the interesting link between special classes of spacetimes and invariant theory. This link is explicitly demonstrated by the connection between the Cartan involution, which is important in the classification of Lie algebras, and a simple time-reflection. This enabled us to connect these seemingly distinct areas and use the best from both worlds to prove deep results about the existence/non-existence of certain solutions. It is believed that this bond will continue to bear fruits in investigations to come.

\section*{Acknowledgments}

We thank  Alan Barnes for reading the manuscript and Jos\'{e} M M Senovilla for useful comments and references.
M.O. has been supported by
research plan RVO: 67985840 and research grant GA\v CR P203/10/0749.
L.W.\ has been supported by an Yggdrasil mobility grant No 211109 to
Stavanger University, a BOF research grant of Ghent University, and
a FWO mobility grant No V4.356.10N to Utrecht University, where
parts of this work were performed.  M.O.\ and L.W.\ thank the
Faculty of Science and Technology of the University of Stavanger for
its hospitality during a research stay.

\appendix


\setcounter{equation}{0}
\renewcommand{\theequation}{A\arabic{equation}}

\section{Orbits of tensors, Cartan involutions, and null alignment theory}

\label{section: Preliminaries}

\subsection{Orbits of tensors; tensors characterized by their invariants}\label{subsec: Pre orbits and invchar}

Let us review some results from invariant theory and define the
appropriate concepts which we need. Furthermore,  we will consider
polynomial invariants
of tensors, and so in what follows `invariants' is to be
understood as '\emph{polynomial invariants}'.

The idea is to consider a group $G$ acting on a vector space $V$. In
our case we will consider a real $G$ and a real vector space $V$.
However, it is advantageous to review the complex case with a
complex group $\GC$ acting on a complex vector space $\VC$.  For a
vector $X\in \VC$, the \emph{orbit} of $X$ under the action of $\GC$
is defined as \be\label{OC(X)}
 \mathcal{O}_{\mathbb{C}}(X)\equiv \{  g(X)\in\VC ~\big{|}~g\in \GC \}\subset \VC.
\ee
Then (\cite{Procesi07}, p555-6):
\begin{thm}\label{thm Procesi}
If $\GC$ is a linearly reductive group acting on an affine variety
$\VC$, then the ring of invariants is finitely generated. Moveover,
the quotient $\VC/ \GC$ parameterises the closed orbits of the
$\GC$-action on $\VC$ and the invariants separate closed orbits.
\end{thm}
Here the term \emph{closed} refers to \emph{topologically closed}
with respect to the standard vector space topology and henceforth,
closed will mean topologically closed. This implies that given two
distinct closed orbits $A_1$ and $A_2$, then there is an invariant
with value $1$ on $A_1$ and $0$ on $A_2$. This enables us to define
the set of {closed} orbits: \begin{eqnarray}\label{CC(X)}
\mf{C}_{\mathbb{C}}\equiv \{\mathcal{O}_{\mathbb{C}}(X)\subset \VC
~\big{|}~ \mathcal{O}_{\mathbb{C}}(X) \text{ closed}.  \}
\end{eqnarray} Based on the above proposition we can thus say that the
invariants separate elements of $\mf{C}_{\mb{C}}$ and hence we will
say that an element of  $\mf{C}_{\mb{C}}$ is \emph{characterised by
its invariants}.

In our case we  consider the real case where we have the Lorentz
group $G=O(1,n-1)$ which is a real semisimple group. For real
semisimple groups acting on a real vector space $V$ we do not have
the same uniqueness result as for the complex case \cite{EbeJab09},
see also Remark~\ref{rem_wick}. However, by complexification,
$[G]^{\mathbb{C}}=\GC$ we have
$[O(1,n-1)]^{\mathbb{C}}=O(n,\mathbb{C})$, and by complexification
of the real vector space $V$ we get $\VC\cong V+iV$. The
complexification thus lends itself to the above theorem.

Concretely, we study tensors, $T$, belonging to some tensor space
${\cal T}^r_{s}\equiv (T_pM)^{\otimes r}\otimes(T^*_pM)^{\otimes
s}$, where $p$ is a point of a $n$-dimensional manifold $M$ with
Lorentzian metric $\bg$. Let
$\{\bm_{\alpha=1,\ldots, n}\}$ be a basis of
vectors of $T_p M$. Let $g\in G$ be a Lorentz transformation, with
representation matrix $(M^\alpha_{~\beta})$ wrt ${\mbold\omega}$,
i.e., in the natural action of $g$ on $T_p M$ we have $g(
{\bm}_\beta)=M^\alpha{}_\beta  {\bm}_\alpha$. Consider now the
following action on the components of $T$ wrt ${\mbold\omega}$: \be
 T^{\alpha_1...\alpha_r}{}_{\beta_1...\beta_s}\mapsto
 (M^{-1})^{\alpha_1}_{~~\gamma_1}...(M^{-1})^{\alpha_r}_{~~\gamma_r}T^{\gamma_1...\gamma_r}{}_{\delta_1...\delta_s}
 M^{\delta_1}_{~~\beta_1}\ldots M^{\delta_s}_{~~\beta_s}.
 \label{g(v)}
\ee
As is well-known, the real numbers on the right hand side may be interpreted as either
\begin{enumerate}
\item the components of the original tensor $T$ wrt a new basis $\{g({\mbold e}_\beta)\}$
of $T_p M$ (and the dual basis of $T_p^*M$), or
\item the components wrt the original basis $\{{\mbold e}_\beta\}$ of a new tensor $T'$,
which is the result of $g^{-1}$ acting as a tensor map on $T\in{\cal
T}^r_s$.
\end{enumerate}
In the former case one puts the components of $T$ in a vector $v\in
V=\R^m$, $m=n^{r+s}$, and one speaks about the {\em passive action}
of $O(1,n-1)$ on $V$; notice that $V$ has an $(r,s)$-tensor
structure as well here (over $\R^n$ instead of $T_p M$). In the
latter case one considers $T\mapsto T'$, referred to as the {\em
active action} of $g^{-1}\in O(1,n-1)$ on $V={\cal T}^r_s$. It is
clear that both viewpoints are essentially equivalent, although one
of them may be more natural in a specific context. In either picture
we may consider a collection (or direct sum) of tensors instead of a
single one (which just changes $V$ accordingly).

Based on the above, tensors `characterized by  invariants' are
defined as follows, in the passive viewpoint.
\begin{defn}\label{defn characterized by invariants}
Consider a (real) tensor, $T$, or a direct sum of tensors, and let
$\Tt\in V$ be the corresponding vector of components wrt a certain
basis. If the orbit of $\Tt$ under the complexified Lorentz group
$G^{\mathbb{C}}$ is an element of $\mf{C}_{\mb{C}}$, i.e.,
$\mathcal{O}_{\mathbb{C}}(\Tt) \in \mf{C}_{\mathbb{C}}$, then we say
that $T$ is \emph{characterised by its invariants}.
\end{defn}
\noindent As the invariants parametrise the set $\mf{C}_{\mb{C}}$
and since the group action defines an equivalence relation between
elements in the same orbit this definition makes sense.

In analogy with (\ref{OC(X)}) and (\ref{CC(X)}) let us define the real orbit through $X$ and the set of real closed orbits:
\begin{eqnarray}
\mathcal{O}(X)\equiv \{ g(X)\in V ~\big{|}~g\in G \}\subset V,\\
\mf{C}\equiv\{\mathcal{O}(X)\subset V ~\big{|}~ \mathcal{O}(X)
\text{ closed}  \}. \end{eqnarray} How do the results of
Proposition~\ref{thm Procesi} translate to the real case? A real
orbit $\mathcal{O}(X)$ is a real section of the complex orbit
$\mathcal{O}_{\mathbb{C}}(X)$. However, there might be more than one
such real section having the same complex orbit. Using the results
of \cite{EbeJab09}, these real closed orbits are disjoint, moreover:
\begin{thm}
$\mathcal{O}(X)$ is closed in $V$ $\Leftrightarrow$ $\mathcal{O}_{\mathbb{C}}(X)$ is closed in $\VC$.
\end{thm}

Combining this with Proposition~\ref{thm Procesi} and Definition
\ref{defn characterized by invariants} we thus have

\begin{cor}\label{cor: closed is invarchar}
A tensor $T$ is characterised by its invariants iff its orbit is closed in $V$, $\mathcal{O}(\Tt)\in\mf{C}$.
\end{cor}

\begin{rem}
The case of a direct sum of curvature tensors (i.e.,  the Riemann
tensor and its covariant derivatives)  is of particular importance
for the equivalence problem of metrics (of arbitrary signature). Let
$\tilde X=\tilde{R}_{\mbold
\omega}\equiv[{R}_{\alpha\beta\gamma\delta},{R}_{\alpha\beta\gamma\delta;\epsilon},...,{R}_{\alpha\beta\gamma\delta;\epsilon_1...\epsilon_k}]\in
\R^{m(k)}$ be the vector of components wrt a (for instance
orthonormal) frame ${\mbold \omega}=\{\bm_{\alpha=1,\ldots, n}\}$,
at a point $p$ of a manifold $M$ with metric $\bg$, of the curvature
tensors up to the $k$th derivative, where
$m(k)=n^4(n^{k+1}-1)/(n-1)$.
Then the action of $g\in
O(1,n-1)$ on $\tilde{X}$ is
\[ g(\tilde{X})=\left[ M^{\kappa}_{~\alpha}...M^{\nu}_{~\delta}{R}_{\kappa...\nu},
M^{\kappa}_{~\alpha}...M^{\nu}_{~\delta}M^{\xi}_{~\epsilon}{R}_{\kappa...\nu;\xi},\,\ldots\,,
M^{\kappa}_{~\alpha}...M^{\nu}_{~\delta}M^{\xi_1}_{~\epsilon_1}...M^{\xi_k}_{~\epsilon_k}{R}_{\kappa...\nu;\xi_1...\xi_k}\right].
\]
Let $\tilde{Y}=\tilde{R}'_{{\mbold \omega}'}\in \R^{m(k)}$ be the
analogous curvature vector for a metric $\bg'$ on $M$, wrt a frame
${\mbold \omega}'$ at $p$.  Then, if $\tilde{X}$ and $\tilde{Y}$
are in the same {\em real} orbit, we have $\tilde{Y}=g(\tilde{X})$
for certain $g\in O(1,n-1)$, i.e.,
the respective representation vectors $\tilde X$ and $\tilde Y$ are
separated by a mere rotation of frame,
${\mbold\omega}'=g({\mbold\omega})$. If this holds for $k=n(n+1)/2$
at every point $p$ of a local neighbourhood $U$ of $M$,
then a result of Cartan
(see e.g.\ \cite{Stephanibook}) tells that $\bg$ and $\bg'$ are
equivalent on $U$.
In this way the equivalence problem is reduced to a question of
classifying the various orbits.
\end{rem}

\begin{rem}
\label{rem_wick} As pointed out, different closed real orbits
$\mathcal{O}(T)$ may have the same invariants (in line with the
comments in \cite{ColHerPel09a,HerCol10}). An example of this is
given by the pair of metrics, clearly related by a double Wick
rotation \cite{HerCol10}:
\begin{eqnarray}
\d s^2_1&=& -\d t^2+\frac{1}{x^2}\left(\d x^2+\d y^2+\d z^2\right),\nonumber \\
\d s^2_2&=& \d \tau^2+\frac{1}{x^2}\left(\d x^2+\d y^2-\d
\zeta^2\right) . \end{eqnarray} These metrics are symmetric
(${R}_{abcd;e}=0$) and conformally flat; hence, the Riemann tensor
is the only non-zero curvature tensor and is equivalent to the Ricci
tensor. In both cases, at any space-time point, the Ricci operator
$R^a{}_b$ acting on tangent space has a single eigenvalue 0 and a
triple eigenvalue -2 (the space-times being homogeneous), such that
the respective Ricci tensors have the same polynomial invariants and
belong to the same complex orbit $\mathcal{O}_{\mathbb{C}}(T)$.
However, the Segre type of $R^a{}_b$ is $\{1,(111)\}$ for the former
and $\{(1,11)1\}$ for the latter metric; thus the respective Ricci
tensors lie in separate real orbits $\mathcal{O}(T)$.
\end{rem}

\subsection{Cartan involutions of the Lorentz group}

\label{subsec: Pre Cartan invol}

{\em \ref{subsec: Pre Cartan invol}.1 {\bf Representation on tensor spaces}.}
Consider the full Lorentz group $G=O(1,n-1)$.  
Let $K\cong O(n-1)$ be a maximal compact `spin' subgroup of
$O(1,n-1)$. Then there exists a {unique} Cartan involution $\theta$
of $O(1,n-1)$ with the following properties \cite{RicSlo90}:
\begin{enumerate}
\item[(i)] $\theta$ is invariant under the adjoint action of $K$:
\begin{equation}\label{k and theta commute}
Ad_{K}(\theta)=\theta,\quad\text{i.e.,}\quad k\theta=\theta
k,\;\forall k\in K;
\end{equation}
\item[(ii)] $O(1,n-1)$ is $\theta$-stable, $\theta(O(1,n-1))=O(1,n-1)$;
\item[(iii)] $\theta$ is
the automorphism $X\mapsto -X^*$ of the Lie algebra
$\mf{g}\mf{l}(n,\mb{R})$, where $^*$ denotes the adjoint (or
transpose, since the coefficients are real).
\end{enumerate}

 In general, the maximal subgroups of a semi-simple Lie group
$G$ are all conjugate, such that two Cartan involutions
are related by
$\theta_2=\mathrm{Int}(g)\theta_1\mathrm{Int}(g^{-1})$, where
$\mathrm{Int}(g)$ is the inner automorphism by a certain $g\in G$.

In our case, consider the natural representation of $G=O(1,n-1)$ on
the tangent space $T_p M$ at a point $p$ of a Lorentzian manifold
$(M,\bg)$. Then, any maximal compact subgroup $K$ is in biunivocal
relation with the timelike direction which is invariant under the
action of $K$. If this direction is spanned by the unit timelike
vector $\bu$,
then it is easy to see
that the unique Cartan involution corresponding to $K$ is simply the
reflection
\begin{equation}\label{u reflection}
\theta:\quad\bu\mapsto -\bu,\quad {\mbox{{$\bf x$}}}\mapsto {\mbox{{$\bf x$}}},\;\;\forall {\mbox{{$\bf x$}}}\bot\bu ,
\end{equation}
acting as an inner automorphism on $G$.~\footnote{If we had taken
the special Lorentz group $G=SO(1,n-1)$ instead of the full one,
then (\ref{u reflection}) would still give the Cartan involutions
for this case, but these do {\em not} have an inner action.} Thus
$\theta$ can be seen as a Lorentz transformation itself, with action
(\ref{u reflection}) on $T_p M$. In any $\bu$-ONF ${\cal
F}_u=\{\bm_1={\bf u},{\bf m}_{i=2,...,n}\}$ we have the matrix
representation:
\[ [\theta]_{{\cal F}_u}=(\theta^\alpha_{~\beta})=\diag (-1,1,...,1).\]
Obviously in such a frame $\theta$ is simply a {\em time reversal
transformation}. In abstract index notation we have
\begin{equation}
    \label{g and theta}
    \delta^a_b=g^a{}_b\equiv h^a{}_b-u^au_b, \qquad \theta^a{}_b=h^a{}_b+u^au_b,
\end{equation}
where the first part defines the projector $h^a{}_b$ of $T_pM$ orthogonal to $\bu$,
$\delta^a_b$ being the identity transformation.

Through the tensor map construction  $\theta$ acts as a reflection
($\theta^2={\sf 1}$) on any tensor space $V={\cal T}^r_{s}$  by
(\ref{g(v)}), with
$M^\alpha_{~\beta}=(M^{-1})^\alpha_{~\beta}=\theta^\alpha_{~\beta}$
(we adopt the active viewpoint here and, with a slight abuse of
notation, denote any representation of $\theta$ with $\theta$).
Denote $N_{\boldsymbol\alpha\boldsymbol\beta}$ for the number of
indices `$u$' in the tensor component
$T^{\boldsymbol\alpha}_{~\boldsymbol\beta}$ wrt ${\cal F}_u$. Notice
that $N_{\boldsymbol{\alpha}\boldsymbol{\beta}}$ is well-defined:
any other $\bu$-ONF is related to ${\cal F}$ by an $O(n-1)$-spin
preserving the number of `$u$'-indices.
Then, $\theta(T)^{\boldsymbol\alpha}_{~\boldsymbol\beta}$ equals
$+T^{\boldsymbol\alpha}_{~\boldsymbol\beta}$ if
$N_{\boldsymbol{\alpha}\boldsymbol{\beta}}$ is even, and
$-T^{\boldsymbol\alpha}_{~\boldsymbol\beta}$ if
$N_{\boldsymbol{\alpha}\boldsymbol{\beta}}$ is odd.

The following properties are immediate from the above definition:
\begin{enumerate}
\item $\theta$ commutes with any tracing $\text{Tr}_k$ over $k$ covariant and $k$ contravariant indices of a type $(r,s)$ tensor $T$, $r,s\geq k$:
\be\label{comm theta contr}
\theta(\text{Tr}_k(T))=\text{Tr}_k(\theta(T))
\ee
\item for tensors $S\in {\cal T}^{r_1}_{s_1}$ and $T\in {\cal T}^{r_2}_{s_2}$ one has
\be\label{tensor map} \theta(S\otimes T)=\theta(S)\otimes \theta(T).
\ee
\item $\theta$ commutes with lowering or raising indices of a tensor (by contraction with $g_{ab}$ or $g^{ab}$), as follows from properties 1 and 2.
\end{enumerate}

\vspace{.2cm} \noindent{\em \ref{subsec: Pre Cartan invol}.2
Orthogonal splitting.}
Since $\theta^2={\sf 1}$, we can split the
vector space $V$ into $\pm 1$ eigenspaces, $V=V_+\oplus V_-$:
\be\label{V+-} V_+=\{ T\in V~\big{|}~ \theta(T)=+T\}, \qquad  V_-=\{
T\in V~\big{|}~ \theta(T)=-T\}. \ee Consequently, for any $T\in V$,
we get the split:
\begin{equation}\label{orthsplit}
 T=T_+ + T_-, \qquad T_{\pm}=\frac 12[T\pm\theta(T)] \in V_\pm .
\end{equation}
Thus
$(T_+)^{\boldsymbol\alpha}_{~\boldsymbol\beta}=T^{\boldsymbol\alpha}_{~\boldsymbol\beta}$
when $N_{\boldsymbol{\alpha}\boldsymbol{\beta}}$ is even and
$(T_+)^{\boldsymbol\alpha}_{~\boldsymbol\beta}=0$ when
$N_{\boldsymbol{\alpha}\boldsymbol{\beta}}$ is odd, and vice versa
for $T_-$ (cf.\ supra). In covariant language, $T_+$ ($T_-$) is
constructed from $T$ by adding all possible contractions with an
even (odd) number of $-u^au_b$ projectors, completed with $h^a{}_b$
projections.
\begin{ex}\label{ex:Tabc}
For a rank 3 covariant tensor $T_{abc}=T_{a[bc]}$ we get
\begin{eqnarray}
&(T_+)_{abc}&=(h_a^{~d}h_b^{~e}h_c^{~f}+u_au^du_bu^eh_c^{~f}+u_au^dh_b^{~e}u_cu^{f}+h_a^{~d}u_bu^eu_cu^f)T_{def}\nonumber\\
&&=h_a^{~d}h_b^{~e}h_c^{~f}T_{def}+2u_au_{[b}h_{c]}^{~f}T_{uuf}\label{Tex1}\\
&(T_-)_{abc}&=-(h_a^{~d}h_b^{~e}u_cu^{f}+h_a^{~d}u_bu^eh_c^{~f}+u_au^dh_b^{~e}h_c^{~f}+u_au^{d}u_bu^eu_cu^f)T_{def}\nonumber\\
&&=2h_a^{~d}h_{[b}^{~e}u_{c]}T_{deu}+u_au^dh_b^{~e}h_c^{~f}T_{def}.\label{Tex2}
\end{eqnarray}
\end{ex}
Since $\theta$ is a Lorentz transformation we have
$\theta(\bg)=\bg$, whence $\bg=\bg_+$. As $\theta$ acts trivially on
scalars $f$ we also have $f=f_+$. Other immediate properties of this
split are the following
\begin{enumerate}
\item Recall that the metric inner product of $S,\,T\in V$ is defined by
\begin{equation}\label{innerprod on V}
\bg(S,T)=g_{a_1b_1}\cdots g_{a_rb_r}g^{c_1d_1\cdots
c_{s}d_{s}}S^{a_1\ldots a_r}{}_{c_1\ldots c_{s}}T^{b_1\ldots
b_r}{}_{d_1\ldots d_{s}}=S^{a_1\ldots a_r}{}_{c_1\ldots
c_{s}}T_{a_1\ldots a_r}{}^{c_1\ldots c_{s}}.
\end{equation}
Since
\[
\bg(S_+,T_-)=
\theta(\bg(S_+,T_-))=\theta(\bg)(\theta(S_+),\theta(T_-))=-\bg(S_+,T_-)
\]
it follows that the split (\ref{orthsplit}) is $\bg$-orthogonal,
$\bg(S_+,T_-)=0$.
Hence,
\begin{eqnarray}
\bg(S,T)&=&\bg(S_+,T_+)+\bg(S_-,T_-)\label{gST}\\
&=&\left(\sum_{N_{\boldsymbol{\alpha}\boldsymbol{\beta}}=even}-\sum_{N_{\boldsymbol{\alpha}\boldsymbol{\beta}}=odd}\right) S^{\alpha_1\alpha_2\ldots}{}_{\beta_1\beta_2\ldots}T^{\alpha_1\alpha_2\ldots}{}_{\beta_1\beta_2\ldots}.\label{gST ONF}
\end{eqnarray}
\item From (\ref{orthsplit}) and properties 1--3 of $\theta$ it follows that taking the $+$ and $-$ parts
of a tensor commutes with any tracing $\text{Tr}_k$, \be\label{comm
+- contr} \text{Tr}_k(T)_\pm=\text{Tr}_k(T_\pm), \ee as well as with
lowering and raising indices, and that for $S\in {\cal
T}^{r_1}_{s_1}$ and $T\in {\cal T}^{r_2}_{s_2}$ we have
\be
    \label{prop +- tensorprod}
    (S\otimes T)_+=S_+\otimes T_++S_-\otimes T_-,\qquad (S\otimes T)_-=S_+\otimes T_-+S_-\otimes T_+.
\ee
As a consequence of (\ref{prop +- tensorprod}) we get
\begin{equation}\label{ST+-}
S=S_\pm,\,T=T_\pm\,\Rightarrow\,S\otimes T=(S\otimes T)_+,\qquad S=S_\pm,\,T=T_\mp\,\Rightarrow\,S\otimes T=(S\otimes T)_-.
\end{equation}
In combination with (\ref{comm +- contr}) and $f=f_+$ for scalars we
thus get in particular:
\begin{prop}\label{prop: PM contr}
If $T=T_-$ then also $T^{2m+1}=(T^{2m+1})_-$ and
$\text{Tr}_k(T^{2m+1})=\text{Tr}_k(T^{2m+1})_-$ for any odd power.
In particular, if $T$ is a type $(r,r)$ tensor then
$\text{Tr}_{(2m+1)r}(T^{2m+1})=0$.
\end{prop}
\end{enumerate}

\vspace{.2cm} \noindent{\em \ref{subsec: Pre Cartan invol}.3
Euclidean inner product.} The Cartan involution $\theta$ induces an
inner product $\inn{-}{-}$ on $V$:
\begin{equation}
    \label{def_inn}
 \inn{S}{T}\equiv \bg(\theta(S),T)=\bg(S,\theta(T))=\bg(S_+,T_+)-\bg(S_-,T_-).
\ee
In any {$\bu$-ONF} 
we get \be\label{innST}
 \inn{S}{T}=\left(\sum_{N_{\boldsymbol{\alpha}\boldsymbol{\beta}}=even}+\sum_{N_{\boldsymbol{\alpha}\boldsymbol{\beta}}=odd}\right) S^{\alpha_1\ldots\alpha_p}{}_{\beta_1\ldots\beta_q}T^{\alpha_1\ldots\alpha_p}{}_{\beta_1\ldots\beta_q}=\sum_{\boldsymbol{\alpha}\boldsymbol{\beta}} S^{\alpha_1\ldots\alpha_p}{}_{\beta_1\ldots\beta_q}T^{\alpha_1\ldots\alpha_p}{}_{\beta_1\ldots\beta_q}.
\end{equation}
Compare with (\ref{gST}) and (\ref{gST ONF}). As is clear from (\ref{innST}), $\inn{-}{-}$ is Euclidean ($\inn{T}{T}\geq 0,\, \inn{T}{T}=0 \Leftrightarrow T=0$).
Notice that the norm $||T||=\inn{T}{T}^{1/2}$ associated to this
inner product is $K$-invariant, i.e., for $k\in K$ one has
$||k(T)||=||T||$,~\footnote{\label{kT=T} This is an immediate
consequence of (\ref{def_inn}), the property
(\ref{k and theta commute}) and the fact that $k$ is a Lorentz
transformation.} but it is clearly {\em not} invariant under the
full Lorentz group.

\begin{rem}\label{rem: SE density} The norm $||T||$ corresponds to the
{\em super-energy density} of the tensor $T$ relative to $\bu$ (see
\cite{Senovilla00}, pp.\ 2806, and \cite{Senovilla06}). Also compare
with \cite{syngespec}, chapter IX, for the case of Maxwell-like
tensors, and with \cite{Bel62,Robinson97} for the Bel-Robinson tensor.
\end{rem}

\begin{rem}\label{rem selfadjoint} For later use, we note that
if ${\cal O}$ is a symmetric
(self-adjoint)/antisymmetric(anti-self-adjoint) linear
transformation of $V$ wrt the inner product $\bg$, i.e., $\bg({\cal
O}(S),T)=\pm\bg(S,{\cal O}(T))$, then ${\cal O}_+$ (resp.\ ${\cal
O}_-$) is the symmetric/antisymmetric (resp.\
antisymmetric/symmetric) part of ${\cal O}$ wrt the Euclidean inner
product $\inn{-}{-}$. This follows immediately from
\[
\inn{({\cal O}_++{\cal O}_-)(S)}{T}=\inn{{\cal O}(S)}{T} =
\pm\bg(S,{\cal O}(\theta(T)))=\pm\inn{S}{\theta({\cal
O})(T)}=\pm\inn{S}{({\cal O}_+-{\cal O}_-)(T)}.
\]
\end{rem}

\subsection{Null alignment theory }
\label{subsec: null alignment}

We briefly revise the null alignment theory for tensors over a
Lorentzian space developed in \cite{Milsonetal05} (see \cite{OrtPraPra12rev} for a recent review).
Let $T_{a_1\ldots a_p}$ be a covariant rank $p$ tensor and ${\cal
F}_\ell=\{\bm_\alpha\}=\{\bm_0=\bl,\bm_1=\bn,\bm_{\hat i=3...n}\}$ a null
frame of $T_p M$. Under a positive boost
\begin{eqnarray}
b_\lambda:\; \bl\mapsto \bl'=e^\l\bl,\quad \bn\mapsto \bn'=e^{-\l}\bn,\quad
{\bf m}_{\hat i}\mapsto \bm_{\hat i}'= {\bf m}_{\hat i}\qquad (\l\in {\cal F}_M)\label{boost}
\end{eqnarray}
in the $\bl\wedge \bn$-plane, the tensor components transform
according to
\begin{equation}\label{boost tensor transform }
T_{\alpha_1\ldots \alpha_p}\mapsto T'_{\alpha_1\ldots \alpha_p}=e^{\l b_{\alpha_1\ldots \alpha_p}}T_{\alpha_1\ldots
\alpha_p},\qquad b_{\alpha_1\ldots \alpha_p}\equiv \sum_{i=1}^p
(\delta_{\alpha_i0}-\delta_{\alpha_i1}),
\end{equation}
where $\delta_{\alpha\beta}$ is the Kronecker delta symbol.
Thus the integer $b_{\alpha_1\ldots \alpha_p}$ is the difference between the
numbers of 0- and 1-indices, and is called the {\em boost weight}
(henceforth abbreviated to b.w.) of the frame component
$T_{\alpha_1\ldots \alpha_p}$ or, rather, of the $p$-tuple $(\alpha_1,\ldots,\alpha_p)$.
The maximal b.w.\ of the non-vanishing components of $T$,
in its decomposition wrt ${\cal F}_\ell$, is an invariant of Lorentz
transformations preserving the null direction spanned by
$\bl$~\cite{Milsonetal05}; it is called the {\em boost order},
$b_T(\bl)$, of $T$ along $\bl$. Let
\begin{equation}\label{def bmax}
\text{b}_\text{max}(T)\equiv \text{max}_{\{\mbold\ell\}}\,
b_T(\mbold\ell)
\end{equation}
denote the maximal value of $b_T(\bl)$ taken over all null vectors
$\bl$, based on the antisymmetries of $T$. For a generic $\bl$ one
has $b_T(\bl)=\text{b}_\text{max}(T)$; if, however,
$b_T(\bl)<\text{b}_\text{max}(T)$ then $\bl$ is said to span an {\em
aligned null direction} (AND) of {\em alignment order
$\text{b}_\text{max}(T)-b_T(\bl)$}.
An AND of alignment order 1,
2, 3,... is called single, double, triple,... . Defining
\begin{equation}
\text{b}_\text{min}(T)\equiv \text{min}_{\{\mbold\ell\}}\,
b_T(\mbold\ell),
\end{equation}
the integer
\begin{equation}
\label{def zeta} p_T\equiv
\text{b}_\text{max}(T)-\text{b}_\text{min}(T)
\end{equation}
defines the {\em primary alignment type} of $T$. Let $\bl$ be a
vector of maximal alignment ($b_T(\bl)=\text{b}_\text{min}(T)$),
then
\begin{equation}
    s_T\equiv \text{b}_\text{max}(T)-\chi_T,\qquad \chi_T \equiv
    \text{min}_{\{{\bf n}| b_T(\mbold\ell)=\text{b}_\text{min}(T),\,n^al_a=1\}}\,
b_T({\bf n})
\end{equation}
is the {\em secondary alignment type} of $T$, and the couple
$(p_T,s_T)$ the {\em (full) alignment type}.

In agreement with terminology given to the Weyl tensor (see also
below), we call a tensor $T$ of {\em type G} if it has no ANDs
($p_T=0$) and of {\em type I} if it only has one or more ANDs
{($p_T\ge1$)}. It is of {\em type II or more special} if
$\zeta_T\leq 0$ ($p_T\geq \text{b}_\text{max}$), i.e., if in a
suitable null frame only components of non-positive b.w.\ are
non-vanishing; as a particular case it is of {\em type D} if
$\zeta_T=\chi_T=0$ ($p_T=s_T=\text{b}_\text{max}$), i.e., only
components of zero boost weight are non-vanishing in some null frame
$\{\bl,\bn,\bm_{\hat i}\}$, which is then called {\em canonical}.
We define $T$ to be of {\em type III} if only components of negative
b.w. are non-zero (i.e., $p_T\geq \text{b}_\text{max}+1$). A further
special case occurs when a null vector $\bl$ exists such that
$b_T(\bl)=-\text{b}_\text{max}$; then $\bl$ spans the unique AND of
$T$ which is thus of type $(p_T,s_T)=(2\text{b}_\text{max},0)$, also
called {\em type N}. {According to these definitions type N is a
subcase of type III, which is a subcase of type II, which is, in
turn, a subcase of type I.} Of course, for tensors with many indices
and few antisymmetries there are a lot of intermediate cases, which
may be given specific names if relevant. The trivial case of $T=0$
is dubbed with type O; then one can formally define
$b_T(\bl):=-\text{b}_\text{max}-1$ or $b_T(\bl):=-\infty$.

The following properties are immediate consequences of the above
definitions.

\begin{prop}\label{prop boost order} Let $\bl$ be a null vector,
and $S\neq 0$ and $T\neq 0$ covariant tensors of arbitrary ranks $p$ and $q$,
respectively.

\begin{itemize}

\item For arbitrary $\alpha,\,\beta\in {\cal F}_M$ we have
\begin{equation}\label{boost order lincomb} b_{X}(\bl)\leq
\max(b_S(\bl),b_T(\bl)),\qquad X\equiv \alpha S+\beta T.
\end{equation}

\item For the tensor product of $S$ and $T$,
\begin{equation}\label{boost order product}
b_{S\otimes T}(\bl)=b_S(\bl)+b_T(\bl).
\end{equation}

\item If $\text{Tr}_k$ is any tracing over $k$ covariant and $k$
contravariant indices ($2k\leq q$) then
\begin{equation}\label{boost order trace}
b_{\text{Tr}_k(T)}(\bl)\leq b_T(\bl).
\end{equation}

\item If $\bn$ is a second null vector not aligned with $\bl$ ($n^al_a\neq 0$),
then $b_T(\bl)+b_T(\bn)\geq 0$.
\end{itemize}
\end{prop}

By taking $\bl$ and $\bn$ maximally aligned ($b_T(\bl)=\zeta_T$ and
$b_T(\bn)=\chi_T$), (\ref{boost order product}) and (\ref{boost
order trace}) imply:
\begin{cor}\label{cor: type II-D prod-contract} The properties `type II or more special' and `type D' are preserved by
taking powers of or contractions within a tensor.
\end{cor}

Specifically, the Weyl tensor $C_{abcd}$ of an $n$-dimensional
spacetime
obeys the Riemann-like symmetries
\begin{eqnarray}
&&C_{(ab)cd}=C_{ab(cd)}=0,\qquad C_{abcd}=C_{cdab},\qquad
C_{a[bcd]}=0\label{symm1}
\end{eqnarray}
and the tracefree property
\begin{eqnarray}
&&C^a{}_{bad}=0.\label{tracefree}
\end{eqnarray}
In terms of the Riemann and Ricci tensors and the Ricci scalar
it is given by
\begin{equation}
\label{Riemann decomp}
C_{abcd}=\Rt_{abcd}-\frac{2}{n-2}(g_{a[c}R_{d]b}-g_{b[c}R_{d]a})+\frac{2{\cal
R}}{(n-1)(n-2)}g_{a[c}g_{d]b}.
\end{equation}
For the Riemann, Weyl and Ricci tensors we have
$\text{b}_{\text{max}}=2$. Let $b_{\text{Ric}}(\bl)$ and
$b_{\text{Rie}}(\bl)$ symbolize the boost orders along $\bl$ of the
Ricci and Riemann tensor, respectively. Further consequences of
Proposition \ref{prop boost order} are:

\begin{prop}\label{prop: Riemann boost order} For any null vector
$\bl$:
\begin{equation}\label{boost order Rie}
b_{\text{Rie}}(\bl)=\max(b_C(\bl),b_{\text{Ric}}(\bl)\geq
b_C(\bl),\,b_{\text{Ric}}(\bl).
\end{equation}
\end{prop}

\begin{cor}\label{cor: Riemann align types} The alignment types $(p_R,s_R)$ and $(p_C,s_C)$
of the Ricci and Weyl tensors at a spacetime point are at least as
high as that of the Riemann tensor, i.e., $\max(p_R,p_C)\geq
p_{\Rt}$ and $\max(s_R,s_C)\geq s_{\Rt}$. In particular, if a
Riemann tensor is of type D then the Ricci and Weyl tensors are of
type D or O (but not both type O).
\end{cor}

For a non-zero Weyl tensor in particular, an AND is called a {\em
WAND}. If $p_C=0,\,1,\,2,\,3,\,4$ the primary type has been
respectively symbolized by G, I, II, III,
N~\cite{Coleyetal04,Milsonetal05}; type O symbolizes a zero Weyl
tensor. If $s_C=1,2$ this is denoted by $i$, $ii$ in subscript to
the primary symbol. In this paper we will explicitly use or meet
types G, I$_i$, II$_{ii}\equiv$ D, O and N.
In the type D case, the subtypes D(abc) and D(d) as described in
\cite{Coleyetal04,Ortaggio09} will be relevant, where the former is
the conjunction of types D(a), D(b) and D(c). Here a type D Weyl
tensor is said to be of type D(abc) (D(d)) if in some Weyl canonical null frame $\{\bm_0=\bl,\bm_1=\bn,\bm_{\hat
i}\}$ the components $C_{\hat i\hat j\hat k\hat l}$ ($C_{01\hat
i\hat j}$) all vanish (in which case they in fact vanish in {\em any} such frame).~\footnote{\label{foot: Dabcd property} It is
easy to show that the Lorentz transformations which convert a Weyl
canonical null frame into another one subjects the separate
component sets $[C_{\hat i\hat j\hat k\hat l}]$ and
$[C_{01\hat{i}\hat{j}}]$ to an invertible transformation. Hence the
vanishing of such a set is a well-defined property. The same holds
regarding the separate subtypes D(a), D(b) and D(c).}

\section{Minimal Ricci-and Maxwell-like tensors}

\label{app_minimal}

\renewcommand{\theequation}{B\arabic{equation}}
\setcounter{equation}{0}

In Example \ref{ex:minimal} we saw that, given any unit timelike
vector $\bu$, a minimal vector $\bv$ is either proportional
($\bv\sim\bu$) or orthogonal ($\bv\bot\bu$) to $\bu$ and, in
particular, cannot be null (or ``type N''). Conversely, a given
vector $\bv$ is minimal wrt the unit vector parallel to it if $\bv$
is timelike (or ``type G''), and wrt any $\bu\bot\bv$ when $\bv$
is spacelike (or ``type D''). This provides an explicit proof for
Proposition~\ref{th_minimal_ODIG} in the case of vectors. Likewise, we
give here more explicit proofs in the case of Ricci- and
Maxwell-like rank 2 tensors.

\subsection{Ricci-like tensors}\label{subsec: Ricci app}

Let $(V_n,\bg)$ be a vector space of arbitrary dimension $n$,
equipped with a (non-degenerate) metric $\bg$ of arbitrary signature
$s$. Petrov~\cite{petrov} deduced canonical forms for Ricci-like
tensors $R_{ab}=R_{(ab)}$ over
$(V_n,\bg)$, connected to the Jordan canonical forms of
$R^a{}_b\equiv g^{ac}R_{cb}$. For Lorentzian signature $s=n-2$ there
are four distinct possibilities (see also
\cite{Hallbook,RebSanTei04}), where the Segre types (but not
possible eigenvalue degeneracies) are indicated between brackets:
\begin{eqnarray*}
\text{Type R1} \quad (\{1,1\ldots 1\}):&& R_{ab}=\rho_u
u_au_b+\sum_{i=2}^n \rho_i m^i_am^i_b;\\
\text{Type R2} \quad (\{z\bar{z}1\ldots 1\}):&&
R_{ab}=2\alpha u_{(a}m^2_{b)}+\beta
(u_au_b-m^2_am^2_b)+\sum_{\hat i=3}^n \rho_i m^{\hat i}_am^{\hat i}_b,\quad \alpha\neq 0;\\
\text{Type R3} \quad (\{21\ldots 1\}):&& R_{ab}=2\alpha
l_{(a}n_{b)}\pm l_al_b+\sum_{\hat i=3}^n \rho_i m^{\hat i}_am^{\hat i}_b;\\
\text{Type R4} \quad (\{31\ldots 1\}):&&
R_{ab}=\alpha(2l_{(a}n_{b)}+m^3_am^3_b)+2l_{(a}m^3_{b)} +\sum_{\hat
i=4}^n \rho_i m^{\hat i}_am^{\hat i}_b ,
\end{eqnarray*}
where, as usual,  $\bl$ and $\bn$ are null, $\bu$ is unit timelike and the $\bm^i$ ($\bm^{\hat{i}}$) are unit spacelike.
For our purposes it is enough to mention that:
\begin{enumerate}
\item  Types R3 and R4 have, while types R1 and R2 {\em do not} have, a unique null
eigendirection (spanned by $l^a$). But for a symmetric tensor, null
eigendirections are precisely ANDs of boost order $\leq 0$
(since the equation $l_{[a}R_{b]c}l^c=0$ expresses both conditions
at the same time). Hence, {\em types R3 and R4 precisely cover the
alignment types `II or more special, but not D nor O'}. Type R1
comprises the alignment types O and D (without loss of generality
for $\rho_2=-\rho_u$, see also the proof of Proposition \ref{prop:
PMpeculiar}), while types I and G are distributed over types R1 and
R2, where type I implies type I$_i$ and at least a
$(n-3)$-dimensional surface of single ANDs (see Proposition
\ref{prop Ricci type I is Ii}). As an example, the Ricci tensor
given in eq. (\ref{example_Ricci_Ii}) is of type R1 and of alignment
type I$_i$.
\item Type R1 is the only type having one or more timelike eigendirections (one of them spanned by
$u^a$). Type R2 has two complex eigenvectors $\bu\pm i \bm_2$
corresponding to the eigenvalues $-\beta\pm i\alpha$. In the adapted
canonical null frame ${\cal F}_c=\{\bm_0=\bl,\bm_1=\bn,\bm_{\hat
i}\}$, where $\bl$ and $\bn$ are defined by (\ref{null frame corr}),
the R2 canonical form becomes
\begin{equation}\label{R2 can form}
R_{ab}=\alpha (l_al_b-n_an_b)-2\beta l_{(a}n_{b)} +\sum_{\hat i=3}^n
\rho_i m^{\hat i}_am^{\hat i}_b,\quad \alpha\neq 0.
\end{equation}
\end{enumerate}

In view of point 1 we need to show that eq.\ (\ref{minimal 2tensor
symm})
admits a solution precisely for types R1 and R2. In a $\bu$-ONF
$\{\bu,\bm_2,\bm_{\hat i}\}$, where the vector $\bm_2$ has been
isolated,
(\ref{minimal 2tensor symm}) splits into
\begin{equation}\label{minimal 2tensor symm u2}
R_{2u}(R_{uu}+R_{22})+R_{u\hat j}R_{2}{}^{\hat j}=0,\quad
R_{uu}R_{u\hat i}+R_{u2}R_{2\hat i}+R_{u\hat j}R_{\hat i}{}^{\hat
j}=0.
\end{equation}

\begin{itemize}
\item In type R1 there is at least one eigenvector $\bu$, which
satisfies $R_{iu}=0,\,\forall i$, and thus (\ref{minimal 2tensor
symm}).
\item For type R2 we take the $\bu$-ONF $\{\bu,\bm_2,\bm_{\hat i}\}$
from the canonical form. Then $R_{uu}=-R_{22}=\alpha$ and
$R_{u2}=R_{u\hat i}=R_{2\hat i}=0$, such that eq.\ (\ref{minimal
2tensor symm u2}) is satisfied and $R_{ab}$ is minimal wrt $\bu$.
\item In any null frame $\{\bm_0=\bl,\bm_1=\bn,\bm_{\hat i}\}$ adapted to $\{\bu,\bm_2,\bm_{\hat i}\}$,
the first equation of (\ref{minimal 2tensor symm u2}) becomes
\[
R_{00}^2+\sum_{\hat i=3}^n R_{0\hat i}^2=R_{11}^2+\sum_{\hat i=3}^n
R_{1\hat i}^2.
\]
We see that if $R_{ab}$ is minimal wrt a certain $\bu$ and has an
AND of boost order $\leq 0$ spanned by $\bl$ (i.e., $R_{00}=R_{0\hat
i}=0,\,\forall \hat i$), then the vector $\bn$ defined by
(\ref{u_WANDs}) necessarily spans an AND of boost order $\leq 0$ as
well. By point 1 this excludes types R3 and R4, for which there is
only one double AND (spanned by $\bl$ in their canonical forms).
\end{itemize}

This shows that Ricci-like tensors of types R1 and R2 (alignment
types G, I, D and O) are minimal wrt a certain unit timelike vector
$\bu$, whereas those of types R3 and R4 (alignment types II (not D),
III or N) are not.

We observe also that type R1 is precisely the case of a PE Ricci tensor, while type R2 contains the purely magnetic case where we can take $\beta=0=\rho_i$, cf.\ Proposition \ref{prop_PEPM_Ricci} and Remark \ref{rem_Einstein_PE}.

\subsection{Maxwell-like tensors}\label{subsec: Maxwell app}

Maxwell-like tensors $F_{ab}=F_{[ab]}$ have $\text{b}_{\text max}=1$
and can be of alignment types G, D, O, II and N (we assume hereafter $n>2$ since any non-zero bivector is trivially of type D in two dimensions). Type G (no aligned
null direction) can only occur when $n$ is odd
\cite{BergSen01,Milson04} (see also Remark~\ref{rem: antisymm class}
below). Type O is the trivial case $F_{ab}=0$. Types II and N allow
for precisely one AND (of boost order 0 and -1, respectively); in four \cite{Stephanibook,Hallbook} and higher \cite{Coleyetal04vsi} dimensions {the $F_{ab}$'s of type N are {\em null} Maxwell-like tensors in the sense that all polynomial invariants vanish. For
type D there are two or more ANDs.

Let $F_{ab}\neq 0$ and consider the symmetric tensor
$(F^2)_{ab}\equiv F_{ac}F^{c}{}_b$. In view of the minimal criterion
(\ref{minimal 2tensor antisymm}) for $F_{ab}$, we need to show that
$(F^2)^a{}_b$ has a timelike eigenvector iff $F_{ab}$ has no unique
AND (i.e., it is not of type II or N). This will follow immediately
from:

\begin{prop} $\bl$ is an AND of $F_{ab}\neq 0$ iff it is an AND for
$(F^2)_{ab}$ of boost order $\leq 0$. The symmetric tensor
$(F^2)_{ab}$ is of type R1 or R3.
\end{prop}

\begin{proof} Take an arbitrary null frame ${\cal F}=\{\bm_\alpha\}=\{\bm_0=\bl,\bm_1=\bn,\bm_{\hat
i}\}$. Then
\begin{equation}\label{ref1}
(F^2)_{1\alpha}=F_{1a}F^a{}_\alpha=F_{10}F_{1\alpha}+F_{1\hat
j}F^{\hat j}{}_{\alpha},\qquad
(F^2)_{0\beta}=F_{0a}F^a{}_\b=F_{01}F_{0\b}+F_{0\hat j}F^{\hat
j}{}_{\b}.
\end{equation}
When applied to $\alpha=1$, $\b=0$ and $\b=\hat i$ this gives
\begin{equation}\label{reff}
(F^2)_{11}=-\sum_{\hat j=3}^n F_{1\hat j}^2\leq 0,\qquad
(F^2)_{00}=-\sum_{\hat j=3}^n F_{0\hat j}^2\leq 0,\qquad
(F^2)_{0\hat i}=F_{01}F_{0\hat i}+F_{0\hat j}F^{\hat j}{}_{\hat i}.
\end{equation}
It follows from the last two equations that $F_{0\hat i}=0,\,\forall
\hat i\;\Leftrightarrow\;(F^2)_{00}=(F^2)_{0\hat i}=0,\,\forall \hat
i$, which proves the first statement. Suppose now that
$R_{ab}=(F^2)_{ab}$ were of type R2 and take the null canonical form
(\ref{R2 can form}) associated to the canonical null frame ${\cal
F}_c$. We would have $(F^2)_{11}=-(F^2)_{00}=\alpha\neq 0$, whence
$(F^2)_{00}(F^2)_{11}< 0$, in contradiction with the first two
equations of (\ref{reff}). Finally, suppose that $R_{ab}=(F^2)_{ab}$
were of type R4. In the canonical null frame associated to the
canonical form we have, in particular, $(F^2)_{13}=1$ and
$(F^2)_{11}=0$. From the latter equation and the first equation in
(\ref{reff}) we get $F_{1\hat i}=0,\,\forall \hat i$, but the first
equation of (\ref{ref1}), with $\alpha=3$, then leads to the
contradiction $(F^2)_{13}=0$.
\end{proof}

From this proposition and points 1 and 2 in section \ref{subsec:
Ricci app} we conclude: if $(F^2)_{ab}$ is of type R1 it possesses a
timelike eigenvector and not a unique AND of boost order $\leq 0$,
i.e., $F_{ab}$ doesn't have a unique AND; if $(F^2)_{ab}$ is of type
R3 it possesses {\em no} timelike eigenvector but does have a unique
AND of boost order $\leq 0$, i.e., $F_{ab}$ has a unique AND. It
follows that $F_{ab}$ is minimal wrt a certain $\bu$ iff it does not
possess a unique AND, which is the case iff $(F^2)_{ab}$ is of type
R1.

\begin{rem}\label{rem: antisymm class} In fact, these results can be shown more directly by
considering the classification of Maxwell-like tensors $F_{ab}$ into
three different types and their corresponding canonical forms. {We also indicate the Segre type; degeneracy of the eigenvalue 0 is
indicated by round brackets, but additional degeneracies may occur
in the $z\bar{z}$ parts.}
\begin{eqnarray*}
\text{Type F1}\quad (\{(1,1\ldots 1)z\bar{z}\ldots z\bar{z}\}):&&
F_{ab}=\sum_{k=1}^r 2f_k
v^{2k-1}_{[a}v^{2k}_{b]},\qquad f_k\neq 0;\\
\text{Type F2}\quad (\{11(1\ldots 1)z\bar{z}\ldots z\bar{z}\}):&&
F_{ab}=\sum_{k=1}^r 2f_k
v^{2k-1}_{[a}v^{2k}_{b]}+2\sigma l_{[a}n_{b]},\qquad f_k\neq 0\neq \sigma;\\
\text{Type F3}\quad (\{(31\ldots 1)z\bar{z}\ldots z\bar{z}\}):&&
F_{ab}=\sum_{k=1}^r 2f_k
v^{2k-1}_{[a}v^{2k}_{b]}+2l_{[a}v^{2r+1}_{b]},\qquad f_k\neq 0.
\end{eqnarray*}
Here $r\leq \lfloor \frac{n-i}{2}\rfloor$ for type F$i$. The vectors
$\bl$, $\bn$ and $\bv^l$ are part of a null frame ($\bl$ and $\bn$
being real null and the $\bv^l$ unit spacelike). A scalar $f_k$
corresponds to a complex conjugate pair of eigenvalues $\pm if_k$,
with complex null eigenvectors $v^{2k-1}\pm i v^{2k}$ and the
corresponding elementary divisors being linear.  Analogously as for
the Ricci-like (symmetric) case, this classification can be easily
derived based on the antisymmetry of $F_{ab}$ and the fact that for
Lorentzian signature orthogonal null vectors are parallel;
see also
\cite{Milson04}. The possible numbers of independent (real) null
eigendirections (ANDs) were discussed in \cite{BergSen01}, pp.\
5313; notice that a null vector $v^b$ is an eigenvector of $F^a{}_b$
iff it is an AND (joint condition $v_{[a}F_{b]c}v^c=0$); hence, in
particular, all null vectors of the kernel span ANDs.
\begin{itemize}
\item Type F1 tensors $F_{ab}$ are precisely the purely magnetic ones ($F=F_+$ wrt a certain $\bu$).
The null alignment type is G if and only if $n$ is odd and
$r=(n-1)/2$; in this case the (one-dimensional) kernel is spanned by
a unique unit timelike vector $\bu$ wrt which $F=F_+$. In all other
cases the alignment type is D (or O, corresponding to $r=0$), the
ANDs and the $\bu$ spanning precisely the null and timelike
directions of the kernel (in accordance with Remark \ref{rem:
uniqueness}). In any case $\bu$ belongs to the kernel of
$(F^2)^a{}_b$ (which is type R1) and thus $F_{ab}$ is minimal wrt
$\bu$. Notice that type G can not occur in cases F2 and F3 below, so
that all type G tensors $F_{ab}$ are necessarily PM.
\item
Type F2 tensors are all of alignment type D. There are precisely two
(real) ANDs, spanned by $\bl$ and $\bn$ and corresponding to the
real eigenvalues $+\sigma$ and $-\sigma$, respectively. We have
$F=F_-$ iff $r=0$ (this is automatically true when $n=3$). If $n\geq
4$ and when there is at least one pair of imaginary eigenvalues $\pm
i f_k$ this gives (the only) examples of minimal Maxwell-like
tensors for which $F_+\neq F\neq F_-$. In any case the
$\bl\wedge\bn$ plane is a timelike eigenplane of $(F^2)^a{}_b$
(which is type R1) such that $F_{ab}$ is minimal wrt any unit
timelike $\bu$ in this plane.
\item For type F3 tensors $F_{ab}$, $\bl$ spans the unique AND (corresponding to a cubic elementary divisor $x^3$).
Thus $F_{ab}$ is of type F3 iff it is of alignment type II or (when
$r=0$) N. The Ricci-like tensor $(F^2)_{ab}$ is of type R3 (with, in
particular, $\alpha=0$ in the corresponding canonical form); thus it
has no timelike eigenvectors and cannot be minimal wrt a unit
timelike $\bu$.
\end{itemize}
\end{rem}

\section{Timelike unit vector fields: expansion, rotation, shear, and Raychaudhuri equation}

\label{app_congruences}

\renewcommand{\theequation}{C\arabic{equation}}
\setcounter{equation}{0}

We consider a timelike unit vector field $\bu$, $u_au^a=-1$, and
follow the notation of Chapter~6 of \cite{Stephanibook}. {The
purpose here is to write parts of the Riemann and Weyl tensors in
terms of the kinematic quantities of $\bu$, as defined in (\ref{def
hab}--\ref{uab}) below (see \cite{Ellis71} for a comprehensive overview of results in four dimensions).} We first define the projector
\begin{equation}\label{def hab}
 h_{ab}=g_{ab}+u_au_b ,
\end{equation}
such that $h_{ab}u^b=0$. This enables us to define the rotation, expansion and shear tensors as
\be
 \omega_{ab}=h_a^{\,c}h_b^{\,d}u_{[c;d]} , \qquad \Theta_{ab}=h_a^{\,c}h_b^{\,d}u_{(c;d)} ,
\qquad \sigma_{ab}=\Theta_{ab}-\tilde\Theta h_{ab} ,
\label{optics}
\ee
where
$\tilde\Theta$ is a normalized (volume) expansion scalar defined by
\be
\label{def theta}
 (n-1)\tilde\Theta=\Theta\equiv h^{ab}\Theta_{ab}=u^a_{\ ;a},
\ee
 and the acceleration vector
\be
 \dot u_a=u_{a;b}u^b.
 \label{acceleration}
\ee The tensors (\ref{optics}) and (\ref{acceleration}) are all
spatial, i.e.,
$\omega_{ab}u^a=\Theta_{ab}u^a=\sigma_{ab}u^a=\dot{u}_au^a=0$. One
can write the covariant derivative of $\bu$ in the standard way,
namely \be\label{uab}
 u_{a;b}=-\dot u_au_b+\omega_{ab}+\sigma_{ab}+\tilde\Theta h_{ab} .
\ee

Using this, the Ricci identity $2u_{a;[bc]}=R^d_{\ abc}u_d$ becomes
\begin{eqnarray}
 \frac{1}{2}R^d_{\ abc}u_d= & & -\dot u_{a;[c}u_{b]}+(-\dot u_a+\tilde\Theta u_a)(-\dot u_{[b}u_{c]}+\omega_{bc})
+\omega_{a[b;c]}+\sigma_{a[b;c]} +h_{a[b}h_{c]}{}^d\tilde\Theta_{,d} \nonumber \label{Riem_u} \\
 & & {}+\tilde\Theta(\omega_{a[c}+\sigma_{a[c})u_{b]} +\left(\dot{\tilde\Theta}+\tilde\Theta^2\right) h_{a[c}u_{b]},
\end{eqnarray}
By contraction this gives
\be\label{Ricci_u}
 R^d_{\ b}u_d=-\dot u^a_{\ ;a}u_{b}+\dot u^a(\omega_{ab}-\sigma_{ab})+\omega^a_{\ b;a}+\sigma^a_{\ b;a}
-(n-2)h_b{}^c\tilde\Theta_{,c}+(n-1)\left(\dot{\tilde\Theta}+\tilde\Theta^2\right)u_b,
\ee where a dot denotes a derivative along $\bu$.

We now multiply~(\ref{Riem_u}) by $u^b$. The symmetric part of the
resulting equation can be written as
\begin{equation}
 R^d_{\ abc}u_du^b=\dot u_a\dot u_c-\omega_{ab}\omega^b_{\ c}
 -\sigma_{ab}\sigma^b_{\ c}-2\tilde\Theta\sigma_{ac}-(\dot{\tilde\Theta}+\tilde\Theta^2)h_{ac}+h_a^{\,d}h_c^{\,e}(\dot u_{(d;e)}-\dot\sigma_{de}) ,
 \label{Riem_uu} \\
\end{equation}
where we used the identities
$h_a^{\,d}h_c^{\,e}\dot u_{(d;e)}=h^b_{\ (c}\dot u_{a);b}+\dot u^bu_{(c}(\omega_{a)b}-\sigma_{a)b})-\tilde\Theta\dot u_{(a}u_{c)}$
and $h_a^{\,d}h_c^{\,e}\dot\sigma_{de}=\dot\sigma_{ac}+2u^bu_{(c}\dot\sigma_{a)b}$, while the antisymmetric part reads
\be\label{B9}
 h_a^{\,d}h_c^{\,e}\dot\omega_{de}=2\sigma^b_{\ [a}\omega_{c]b}-2\tilde\Theta\omega_{ac}+h_a^{\,d}h_c^{\,e}\dot u_{[d;e]} ,
\ee
in which the identities $h_a^{\,d}h_c^{\,e}\dot u_{[d;e]}=h^b_{\ [c}\dot u_{a];b}-\dot u^b(\omega_{b[c}+\sigma_{b[c})u_{a]}+\tilde\Theta\dot u_{[a}u_{c]}$ and $h_a^{\,d}h_c^{\,e}\dot\omega_{de}=\dot\omega_{ac}+2\dot u^b\omega_{b[a}u_{c]}$ have been employed.

Further, the trace of~(\ref{Riem_uu}) gives the Raychaudhuri equation
\be
 R^d_{\ b}u_du^b=\dot u^d_{\ ;d}+\omega_{ab}\omega^{ab}-\sigma_{ab}\sigma^{ab}-(n-1)(\dot{\tilde\Theta}+\tilde\Theta^2 ).
 \label{Ruu}
\ee

Substituting in~(\ref{Riem_uu}) the standard definition of the Weyl
tensor and using~(\ref{Ruu}) and the identities
$h_a^{\,d}h_c^{\,e}R_{de}=R_{ac}+2u_{(a}R_{c)u}+u_au_cR_{uu}$ and
$h^{de}R_{de}=R+R_{uu}$, we can write the (electric) components
$C^d_{\ abc}u_du^b$ of the Weyl tensor as
\begin{eqnarray}\label{Cuu}
 C^d_{\ abc}u_du^b= & & \dot u_a\dot u_c-\omega_{ab}\omega^b_{\ c}-\sigma_{ab}\sigma^b_{\ c}-2\tilde\Theta\sigma_{ac}+h_a^{\,d}h_c^{\,e}\left(\dot u_{(d;e)}-\dot\sigma_{de}+\frac{R_{de}}{n-2}\right) \nonumber \label{Weyl_uu} \\
 & & {}-h_{ac}\frac{1}{n-1}\left(\dot u^d_{\ ;d}+\omega_{de}\omega^{de}-\sigma_{de}\sigma^{de}+\frac{h^{de}R_{de}}{n-2}\right) .
\end{eqnarray}

The magnetic components can be expressed in terms of

\begin{eqnarray}\label{Weyl_u}
 C^{dg}{}_{bc}u_dh^b_{~e}h^c_{~f}= & & 2h^{ag}h^b_{~e}h^c_{~f}(-\dot u_a\omega_{bc}+\omega_{a[b;c]}+\sigma_{a[b;c]}) \nonumber \\
                & &  {}+\frac{2}{n-2}h^g_{~[e}\left[(\omega^a_{~f]}-\sigma^a_{~f]})\dot u_a+h^b_{~f]}(\omega^a_{\ b;a}+\sigma^a_{\ b;a})\right] .
\end{eqnarray}

The above equations reduce to formulae (6.26)--(6.30) in
\cite{Stephanibook} when $n=4$.~\footnote{Formula (\ref{B9}) is not
displayed in \cite{Stephanibook}, but agrees with (4.22) of
\cite{HawEll73}. However, in equation (4.27) of this last standard
reference the $\dot u_a\dot u_c$ term appearing in (\ref{Cuu}),
necessary for tracelessness since $h^{de}\dot u_{(d;e)}=\dot u^d_{\
;d}-\dot u^d\dot u_d$, is missing.} Remember, however, that the
electric part of the Weyl tensor consists also of $C_{ijkl}$, which
is not described by~(\ref{Weyl_uu}) for $n\geq 5$. Note that in the
special case of a geodesic $\bu$ one has ${\mbox{{$\bf\dot u$}}}=0$
and the above equations get a simpler form, cf., e.g.,
\cite{KarSen07}.


\providecommand{\href}[2]{#2}\begingroup\raggedright\endgroup

\end{document}